%% file: main.tex
\documentclass[11pt]{article}
\usepackage{graphicx} %
\usepackage[margin=1in]{geometry}
\usepackage{amsmath}
\usepackage{amsthm} %
\usepackage{amsfonts}
\usepackage{tcolorbox}
\usepackage{authblk}
\usepackage{hyperref}
\hypersetup{
    colorlinks=true,
    linkcolor=blue,
    filecolor=blue,      
    urlcolor=blue,
    pdftitle={Computationally Efficient Collaborative Communication},
    }
\usepackage{cleveref}
\usepackage{nicefrac}
\usepackage{url}
\usepackage{xcolor}
\usepackage{pgfplots}
\usepackage{parskip}
\usepackage{enumitem}
\usepackage{mathrsfs}
\usepackage{dsfont}
\usepackage{thmtools}
\usepackage{mathtools}
\usepackage{thm-restate}
\usepackage{algorithm}
\usepackage{algpseudocode}
\usepackage{tikz}
\usepackage{mathdots}
\usepackage{wrapfig}
\pgfplotsset{compat=1.18}
\usepgfplotslibrary{fillbetween}
\usetikzlibrary{calc, positioning, arrows.meta, backgrounds, patterns, trees, shapes, patterns, fit}

\usepackage{adjustbox}
\pgfplotsset{compat=1.17}

\allowdisplaybreaks

\definecolor{darkgreen}{RGB}{0, 100, 0} %
\definecolor{darkred}{RGB}{139, 0, 0}   %

\usepackage{natbib}
\makeatletter
\def\input@path{{}{sections/}{figures/}}
\def\bibfile{bibliography}
\makeatother

\input{macros}

\title{
Computationally Efficient Collaborative Communication\thanks{Authors are listed in alphabetical order. For correspondence, please contact Mark Bedaywi and Nika Haghtalab at
\texttt{\{mark\_bedaywi,nika\}@berkeley.edu}}
\\
\large{
Via Regularity-Based Coarsening}

}

\author{
Mark Bedaywi
\quad
Scott Emmons
\quad
Nika Haghtalab
\quad
Stuart Russell
}
\affil{University of California, Berkeley}

\date{}

\newtheoremstyle{openproblem}
  {3pt} %
  {3pt} %
  {\itshape} %
  {} %
  {\bfseries} %
  {.} %
  {.5em} %
  {} %

\newtheoremstyle{goalstyle}
  {3pt} %
  {3pt} %
  {} %
  {} %
  {\bfseries} %
  {.} %
  {.5em} %
  {} %

\theoremstyle{goalstyle}

\theoremstyle{openproblem}
\newtheorem{theorem}{Theorem}[section]

\newtheorem{definition}[theorem]{Definition}
\newtheorem{proposition}[theorem]{Proposition}

\newtheorem{lemma}[theorem]{Lemma}
\newtheorem{claim}[theorem]{Claim}

\newtheorem{remark}[theorem]{Remark}

\begin{document}

\maketitle

\begin{abstract}
 
Our results show that the existence of a short high-utility protocol already suffices for efficient communication.

In particular, in a game with $n$ possible observations and $m$ actions:
\begin{itemize}
    \item For any achievable target utility $\alpha$, we give an algorithm with $\mathrm{poly}(n, m, 1/\epsilon)$  runtime that designs a protocol achieving utility at least $\alpha-\epsilon$ using only $2^{\mathcal O(CC_\alpha(G))}/\epsilon^2$  bits of communication.
    Here,  $CC_\alpha(G)$ is the minimum number of bits used by any protocol, even a computationally inefficient one, to achieve utility $\alpha$.

    \item We prove that this exponential dependence on $CC_\alpha(G)$ is tight up to a constant. That is, unless $\mathrm P=\mathrm{NP}$, no polynomial-time algorithm can in general find optimal protocols using fewer than $2^{CC_\alpha(G) -2}$ bits.
\end{itemize}

We note that our results strictly weaken the assumptions required by prior work in the multi-agent information aggregation literature, filling a gap that had remained elusive even for games with constant $CC_\alpha(G)$. In particular, prior guarantees for agreement-based information aggregation rely on structural assumptions such as informational substitutes \citep{frongillo2023agreement} or weak learnability \citep{collina2025collaborative}. We show that these assumptions already imply $CC_\alpha(G) = O(1)$ and are therefore more restrictive conditions than required by our protocol to succeed.
Moreover, we show that Aumann agreement~\cite{aumann1976agreeing,aaronson_complexity_2005} can be exponentially suboptimal even in low-complexity games: there are games of communication complexity $CC(G)$ where agents reach agreement in $O(1)$ rounds but obtain no more than $2^{-CC(G)/2}$ utility. We also introduce a notion of lasting agreement and show that while lasting agreement guarantees at least $2^{-CC(G)/2}$ utility, reaching lasting agreement may require $\Omega(n)$ rounds.

On a technical level, our results involve a novel strengthening of the Frieze--Kannan weak regularity lemma \citep{frieze1999quick} and yield the following powerful polynomial-time transformation tool: for every communication game $G$, it constructs a game $\hat G$ that is a ``coarsening'' of the agents' observation spaces into constant-size partitions, such that $G$ and $\hat G$ are ``indistinguishable'' with respect to every short communication protocol.
This coarsening theorem is the engine behind our algorithm and may be of independent interest.

\end{abstract}
\newpage

\tableofcontents
\newpage

\input{sections/01-introduction}
\input{sections/02-roadmap}
\input{sections/03-Preliminaries}
\input{sections/04-Frieze-Kannan-Communication}

\input{sections/05-Hardness-of-Communication}
\input{sections/06-Aumann-Agreement-Guarantees}

\input{sections/07-related-work}

\input{sections/08-Discussion}

\bibliography{\bibfile}
\appendix

\input{sections/09-extra-proofs}

\end{document}

%% file: macros.tex
\newcommand{\E}{\mathop{\mathbb{E}}}
\newcommand{\abs}[1]{\left|#1\right|}
\newcommand{\norm}[1]{\lVert#1\rVert}

\newcommand{\bO}[1]{\mathcal{O}\left(#1\right)}

\newcommand{\acts}{\mathcal{A}}
\newcommand{\dist}{\mu}

\DeclareMathOperator*{\argmax}{arg\,max}

\newcommand{\cutnorm}[1]{\left\lVert#1\right\rVert_{\square, \mu}}

\newcommand{\ceil}[1]{\left\lceil #1\right\rceil}
\newcommand{\floor}[1]{\left\lfloor #1\right\rfloor}

\newcommand{\inner}[1]{\left\langle #1\right\rangle}

\newcommand{\rank}{\mathrm{rank}}

%% file: sections/01-introduction.tex
\section{Introduction}

How should agents communicate when the goal is not to compute a prescribed function exactly, but to make a good decision? And how should a long-context agent decide what information from earlier computation is worth remembering for completing an ongoing task?

These questions arise whenever relevant information is distributed, either across several agents with common interests or across the stages of a single long computation by an agent: A physician and a diagnostic model may each observe different, complementary evidence about the same patient and need to arrive, after a short and cognitively manageable interaction, at a good treatment decision. 
Drones in a swarm may each observe local data and need to exchange only a few bits before committing to a coordinated maneuver. 
Or an agentic assistant, such as Claude~\citep{anthropic_claude} or Codex~\citep{openai_codex}, used for coding or scientific research, has to decide what findings from its earlier reasoning trajectory to preserve under a limited context window.
In all of these examples, the goal is not to reconstruct all hidden information, but to communicate just enough while respecting communication bandwidth and computational (or cognitive) ease
---  to other agents or to a future self --- so that decisions with high utility can be taken based on the shared information.

A protocol for such a task must balance three goals that are usually studied separately. It should lead to a \emph{high-utility action}, use \emph{little communication}, and be simple enough to \emph{find and execute in polynomial time.} Classical information theory is well suited to questions about reconstructing distributed information, but here full reconstruction may be unnecessary. On the other hand, while communication complexity asks how many bits are needed to compute a prescribed function exactly \citep{kushilevitz1996communication}, it can use \emph{computationally intractable} protocols to do so. 

Once we care simultaneously about utility, brevity of communication, and the protocol itself being computationally or cognitively easy, neither framework by itself addresses the full question.

One influential recent line of work has approached decision-centric aggregation of decentralized information through the lens of Aumann agreement, which states that in the classical Bayesian model, two agents with a common prior and common knowledge of one another's posteriors cannot agree to disagree \citep{aumann1976agreeing}. Subsequent work showed that approximate agreement can be reached by exchanging posteriors back and forth \citep{geanakoplos_we_1982}, and that this only requires a small number of rounds \citep{aaronson_complexity_2005}. More recent work has developed tractable variants of this perspective \citep{collina2025collaborative} without the need for priors. This makes agreement-style communication appealing as a model of ``natural'' decentralized interaction: agents repeatedly report their current beliefs or best responses, and useful coordination is hoped to emerge from the conversation.
However, agreement does not in general imply high utility. It has long been understood that agents can agree while still failing to aggregate all useful information \citep{geanakoplos_we_1982}. So to ensure that agreement protocols are useful for decision-making, recent positive results have required additional structure, such as assumptions on the shared prior or on the geometry of the induced decision problem \citep{kong_false_2022, frongillo2023agreement, collina2025collaborative}. While these works provide important motivation, our aim is to study the problem at a broader algorithmic level rather than focus only on agreement-style protocols. Instead, we ask the following question: 
\begin{quote}
    \emph{Among all communication protocols, when can one \textbf{efficiently} find one that is simultaneously \textbf{short} and \textbf{useful}?
    }
\end{quote}

More concretely, we study \emph{common-interest} communication games in which Alice and Bob receive correlated private observations about the state of the world and may communicate before an action is taken.
We consider cases where one of the informed agents is the action taker --- for example, a doctor may consult an AI system and then personally decide on the treatment --- or a bystander, Charlie, who has no other private information takes the action --- e.g., a command center that controls the swarm of drones.\footnote{Our formal model uses Charlie as the action taker because it cleanly separates information revelation from decision-making, but the two viewpoints are very close and they differ by at most a small communication overhead.} Once an action is taken, all agents receive the same utility.
We seek protocols satisfying:
\begin{enumerate}[leftmargin=2em]
    \item \textbf{Utility.} The resulting action should have high expected reward.
    \item \textbf{Communication Efficiency.} The number of communicated bits should be small whenever the problem admits short useful protocols.
    \item \textbf{Computational Efficiency.} The protocol itself should be computable in polynomial time from the explicit description of the game.
\end{enumerate}

Any two of these goals are easy to satisfy in isolation. Revealing all private information is useful and easy to compute, but can take a number of bits that is arbitrarily larger than necessary. Sending nothing is certainly brief and computationally trivial, but may not  lead to high-utility actions. Brute-force search over all short protocols optimizes utility and communication length, but is computationally hopeless. The core problem is therefore to understand when one can approximately satisfy all three desiderata at once.

\subsection{Our Model and Results}

We study a communication game
$
    G=\langle \acts,\Omega_A,\Omega_B,\dist,r\rangle
$
in which Alice observes $\omega_A \in \Omega_A$ and Bob observes $\omega_B \in \Omega_B$ where
the observation pair is drawn from a known joint distribution
$
    \dist \in \Delta(\Omega_A \times \Omega_B)$.
We write $n$ for the size of the observation space per player and $m$ for the size of the action set. Charlie takes an action $a \in \acts$
and the agents receive a shared reward
$
    r(a;\omega_A,\omega_B) \in [-1,1].
$

For a target expected utility level $\alpha$, we write $CC_\alpha(G)$ for the minimum number of communicated bits required to achieve expected utility at least $\alpha$ when computational efficiency is ignored. If expected utility $\alpha$ is unattainable, we set $CC_\alpha(G)=\infty$.
We also define a decision-theoretic analogue of the monochromatic tiling number from communication complexity. The quantity $N_\alpha(G)$ is the partition number of communication game $G$: it is the minimum number of rectangles in a partition of Alice's and Bob's observation space $\Omega_A\times \Omega_B$ into
\[
    R_1 \times C_1,\dots,R_{N_\alpha(G)} \times C_{N_\alpha(G)},
\]
such that if Charlie is told only which rectangle (i.e., $R_i\times C_i$) contains the realized observation pair, then Charlie can still choose an action of expected utility at least $\alpha$.

Our main result is that whenever $CC_\alpha(G)$ is constant, there is a communication protocol that can be found in polynomial time using only a constant number of bits while still achieving utility close to $\alpha$.

\begin{theorem}[\Cref{thm:one-way} Informally]
\label{thm:informal:main:upper}
    For any choice of $\alpha$ and $\delta > 0$, there exists a communication protocol that runs in time $poly(n,m,1/\delta)$ and allows Charlie to take an action that achieves expected utility at least $\alpha - \delta$, while using
     \[
    \bO{\frac{N_\alpha(G)\min\{m,N_\alpha(G)\}}{\delta^2}} \leq \bO{\frac{2^{CC_\alpha(G)}\min\{m,2^{CC_\alpha(G)}\}}{\delta^2}}
\]
bits.
\end{theorem}

\begin{remark}
Theorem~\ref{thm:informal:main:upper} shows that the bit complexity of efficient communication is independent of the size of the observation space or the action set.

By definition, if computational efficiency is ignored there are protocols using only $CC_\alpha(G)$ bits. Importantly, Theorem~\ref{thm:informal:main:upper} shows that with computational efficiency as a constraint, the number of bits communicated is still independent of the size of the observation space but exponential in $CC_\alpha(G)$. We provide evidence that this exponential blowup is unavoidable in general.
\end{remark}

\begin{theorem}[\Cref{thm:perfect-hardness} Informally]
    Unless $\mathrm{P} = \mathrm{NP}$, any efficient communication protocol allowing Charlie to make perfect decisions must, in the worst case, send roughly $2^{CC(G)}$ bits.
\end{theorem}

\paragraph{Comparison to Recent Work on Agreement}
Let us compare Theorem~\ref{thm:informal:main:upper} with recent results that use Aumann \emph{agreement} (or related dynamics approximating agreement) as a mechanism for decision-centric information aggregation. The point of comparison is the strength of the assumptions: taking  $CC_\alpha(G)$ to be a constant is a weaker constraint than  the structural conditions under which prior work shows that agreement itself leads to good decisions.

Since agreement protocols are known not to imply high utility in general \citep{geanakoplos_we_1982}, recent positive results have imposed additional structure on the information environment. For example, \cite{frongillo2023agreement} show that agreement implies accuracy under a strong substitutes-type condition they call \emph{rectangle substitutes}, while \cite{kong_false_2022} obtain a related positive result under a specialized Bayesian substitutes condition based on conditional independence of signals given the state. More recently, \cite{collina2025collaborative} introduce a weaker and more operational condition, which they call \emph{$w$-weak learnability}. Among these assumptions, weak learnability is the weakest assumption leading to the most general guarantee that agreement-style procedures obtain.

In \Cref{prop:cc-weak}, we show that $w$-weak learnability already implies bounded communication complexity: for every fixed target utility level $\alpha$, if a game is $w$-weak learnable then $CC_\alpha(G)=O(1)$ (with the constant depending on $w$ and $\alpha$). Thus the low-communication regime in which Theorem~\ref{thm:informal:main:upper} is strongest already contains the settings captured by the agreement-based results of \cite{collina2025collaborative,frongillo2023agreement}. In this sense, assuming that $CC_\alpha(G)$ is small is a strictly more permissive starting point than past alternatives.\footnote{However, the exact results of our paper and those of \cite{collina2025collaborative} in the Bayesian setting are incomparable. \cite{collina2025collaborative} aims for optimality with respect to a class of functions over the joint feature space, and the runtime of their algorithm depends on how easy it is to optimize over this function class. On the other hand, we compete with the utility of the optimal action (analogous to allowing any function), and as a result our runtime depends on the size of the observation and action spaces.
}

We complement this with an analysis of the performance of Aumann agreement on low-complexity games. In \Cref{prop:AA-upper-bound}, we exhibit communication games of complexity $CC(G)$ in which the agents are immediately in agreement over what the best action is, but only achieve a utility of $2^{-CC(G)/2}$. While it is known that agreement can be reached in constant time \cite{aaronson_complexity_2005, collina_tractable_2024}, agreement is brittle and may not be lasting. We introduce a notion of lasting agreement, defined as the fixed point of an agreement protocol, the point at which the agreed-upon action never changes. In \Cref{thm:AA-guarantee}, we prove that when the agents are in lasting agreement, this bound is tight: the utility agents achieve in lasting agreement over the best action to take in a communication game $G$ is at least $2^{-CC(G)/2}$. In \Cref{prop:AA-too-long}, we give communication games with constant communication complexity in which lasting agreement requires $\Omega(n)$ steps. In such a game, running an agreement protocol requires exponentially more communicated bits to arrive at lasting agreement than simply communicating the full observation outright.

By contrast, whenever $CC_\alpha(G)$ is constant for a target utility level $\alpha$, our protocol achieves utility close to $\alpha$ using a number of bits independent of the size of the game.

\paragraph{Our Strengthening of the Frieze--Kannan Regularity Lemma}
In Section~\ref{sec:1.2}, we describe our technical tools in more detail. Here, let us highlight just one of our contributed technical theorems, a strengthening of the Frieze--Kannan weak regularity lemma \citep{frieze1999quick}. This theorem (stated more formally in \Cref{lem:grid-FK}) serves as a technical engine for our paper and we believe it may be of independent interest to the community.

\begin{theorem}[Informal generalized coarsening theorem]
Let $A^1,\dots,A^k\in[-1,1]^{n_1\times n_2}$ be matrices, let $\dist$ be any distribution over their entries, and let $s$ denote their total $\dist$-weighted squared mass.
There are collections of row sets $\mathcal P$ and column sets $\mathcal Q$, each of size $\bO{s/\epsilon^2}$, such that if $A^i_{\mathcal P,\mathcal Q}$ is obtained by replacing each entry of $A^i$ by its $\dist$-conditional average on the row and column membership pattern induced by $\mathcal P,\mathcal Q$, then for every matrix $i$ and every rectangle $R\times C$,
\[
    \left|
    \E_{(x,y)\sim \dist}
    \left[
        \mathds{1}\{x\in R,\,y\in C\}
        \left(A^i(x,y)-A^i_{\mathcal P,\mathcal Q}(x,y)\right)
    \right]
    \right|
    \leq \epsilon .
\]
Moreover, the sets $\mathcal P,\mathcal Q$ can be found in polynomial time.
\end{theorem}
For communication games, the relevant value of $s$ is constant, so the number of row and column sets is $\bO{\epsilon^{-2}}$ rather than depending on the number of actions.

Applied simultaneously to the reward matrix of every action, this theorem gives a polynomial-time transformation \emph{from any communication game $G$ to a coarsened game $\hat G$ whose reward discrepancies are small on every rectangle-action test.} This is the \emph{indistinguishability} test that then lets us prove that $G$ and $\hat G$ are indistinguishable on every short protocol. Communication in the coarsened game proceeds by having each agent reveal only the membership pattern of their observation in $\mathcal{P}$ or $\mathcal{Q}$, leading to short communication.

At a high level, this result can be viewed as a multicalibration variant of the Frieze--Kannan cut decomposition, in the spirit of recent work connecting regularity and multicalibration \citep{casacuberta2024complexity}. The important strengthening is that the approximation is itself a genuine ``coarsening'' --- i.e., a conditional expectation of the original matrices on the row and column membership patterns --- and applies to joint distributions $\dist$. Our approach constructs the coarsened object directly, which allows our communication protocol to use the coarsened observations within the framework of tests protected by the regularity and indistinguishability argument.

%% file: sections/02-roadmap.tex
\subsection{Techniques and Proof Overview}
\label{sec:1.2}

Our aim is to prove that whenever a short communication protocol exists, there is a short protocol that is efficiently computable.

\paragraph{Search Is Hard, Even For One Bit.}

Let us note the computational challenges that arise from the combinatorial nature of the problem. Even deciding how to send a single bit of information in games whose communication complexity is $1$ is non-trivial as the space of candidate one-bit protocols is exponential: each one-bit protocol for Alice
is analogous to a partition of her observation space into two sets --- one set of observations leading to message $\texttt{1}$ being sent and the other to message $\texttt{0}$. This exponentially large search space implies that low communication complexity does not, by itself, make brute-force search over the short protocol space tractable.

To make matters worse, we will show in \Cref{sec:hardness} that even on games with at most a constant communication complexity, computing perfect communication protocols without needing an exponential blowup in the number of bits is NP-hard --- one can solve graph coloring problems with a good communication scheme.

\paragraph{Learning a Compressed Abstraction for Communication Problems}

Instead, notice that the brute-force hardness of communication search completely disappears if the observation space has constant size. This is because each agent's private signal can be encoded in a constant number of bits, and the number of constant-bit protocols on a constant-size observation space is itself constant. This means that the agents can iterate over all such protocols in constant time.

At a high level, our key result is to show that tractable communication is possible in general because \textbf{\emph{any communication problem is approximately ``effectively equal'' to a constant-observation game}.} We will formalize this insight and what ``effectively equal'' means through the language of indistinguishability.

 Two communication games $G$ and $\hat{G}$ having the same reward on all possible actions and observation pairs $\omega_A, \omega_B \in \Omega_A, \Omega_B$ is enough, but is too strong a notion of equality. Instead, we ask for approximate equality of the expected reward of taking any action, conditioned on any rectangle of observations $S \times T$ that has sufficiently high likelihood (where $S\subseteq \Omega_A$ and $T\subseteq  \Omega_B$) to be $\epsilon$ close.

Formally, every action $a \in \acts$ and every rectangle $S \times T \subseteq \Omega_A \times \Omega_B$ defines a test function $f_{a, S \times T}$ that we can run on any communication game $G$:
\begin{equation*}
    f_{a, S \times T}(G) = \E_{\omega_A, \omega_B \sim \dist}\left[\mathds{1}(\omega_A \in S, \omega_B \in T) \cdot r(a; \omega_A, \omega_B)\right].
\end{equation*}
We say two communication games $G$ and $\hat{G}$ are $\epsilon$-indistinguishable when for \emph{every single possible choice of $S\subseteq\Omega_A$, $T\subseteq\Omega_B$ and $a \in \acts$}, the two test values do not differ by more than $\epsilon$:
\begin{equation*}
    \abs{f_{a, S \times T}(G) - f_{a, S \times T}(\hat{G})} \leq \epsilon.
\end{equation*}

Why is this set of tests natural? It is possible to show that if $G$ and $\hat{G}$ are $\epsilon$-indistinguishable with respect to the set of tests defined by rectangles, then the expected utility of any one-bit protocol run on $G$ and $\hat{G}$ differs by no more than $4\epsilon$. Such games are approximately indistinguishable to Alice, Bob, and Charlie if Alice and Bob are not allowed to send more than a single bit.

\paragraph{Warm-up Algorithm for a Somewhat Tractable Communication Protocol.}
This already suggests a natural warm-up algorithm for designing \emph{somewhat tractable} communication protocols that compete with the optimal single-bit protocol, even when that optimal protocol may itself be computationally inefficient: Compute the compressed game
$\hat G$, find the best one-bit protocol $\hat \pi^*$ in $\hat G$, and then run
$\hat \pi^*$ back in the original game $G$. Then by the
indistinguishability argument, when $\pi^*$ is the optimal one-bit protocol in
$G$ that achieves value $\alpha$, we have that $\pi^*$ achieves value at least
$\alpha-4\epsilon$ in $\hat G$. Since $\hat \pi^*$ is optimal in $\hat G$, it must be at least as good as $\pi^*$ in $\hat{G}$ and so also achieves value at least $\alpha-4\epsilon$ there. Applying
indistinguishability once more, $\hat \pi^*$ achieves value at least
$\alpha-8\epsilon$ back in $G$. Thus, the warm-up algorithm computes an additive
$8\epsilon$-approximation to the best one-bit protocol in $G$.

However, this warm-up algorithm is not yet computationally satisfactory. Even
after compression, the number of effective observations is
$2^{\bO{1/\epsilon^2}}$. Since a one-bit messaging rule for Alice is a subset of her
compressed observation space, the number of possible one-bit messaging rules
for Alice is roughly $2^{2^{\bO{1/\epsilon^2}}}$, and similarly for Bob. Thus
brute-forcing all one-bit protocols in the compressed game gives running time
$2^{2^{\bO{1/\epsilon^2}}}$, independent of the original observation size but
still doubly exponential in $1/\epsilon$.

\paragraph{Why Indistinguishability (Weak Regularity) is not Enough.}
A natural attempt to avoid this brute-force step is to have Alice and Bob simply communicate their full observations in the compressed game. Since the compressed game has only $2^{\bO{1/\epsilon^2}}$ observations, this would use only $\bO{1/\epsilon^2}$ bits. The hope would be that this protocol is an $O(\epsilon)$-additive approximation to the best one-bit protocol.

This is where the indistinguishability (weak-regularity) abstraction runs into a vicious cycle. The compressed game is guaranteed to be close to $G$ for one-bit protocols, but the protocol that communicates the compressed observation uses $\bO{1/\epsilon^2}$ bits. Asking for indistinguishability against all $\bO{1/\epsilon^2}$-bit protocols requires a finer approximation, which increases the size of the compressed game, which increases the number of bits needed to communicate inside it, and so on. The bit requirement explodes.

In what follows, we make a crucial change to this indistinguishability argument, yielding a stronger form of game compression. This stronger compression allows us to compete with the best one-bit protocol using only $\bO{1/\epsilon^2}$ bits, while running in time polynomial in $1/\epsilon$. We then show how the same ideas extend to competing with optimal (intractable) protocols of arbitrary length.

\paragraph{Strengthened Compression via Indistinguishable
Coarsenings}
Our goal is to ensure that communicating the full observation from the compressed game is equivalent to communicating coarse information about the true observation (such as what set it belongs to) in the original game. For this to happen, we need a compression algorithm that leaks no additional information about the original observation space beyond the abstraction itself. We achieve this by introducing the concept of ``indistinguishable coarsening''.

Given a communication game $G$ and a family of sets $\mathcal{P}=\{P_1,\dots,P_\ell \}\subseteq 2^{\Omega_A}$, a coarsening hides Alice's exact observation and reveals only the membership pattern of $\omega_A$ in these sets. The same is done for Bob using sets $\mathcal{Q}=\{Q_1,\dots,Q_\ell\}\subseteq 2^{\Omega_B}$. We write the resulting game as $G_{\mathcal{P},\mathcal{Q}}$, where the reward of any observation is its conditional expected reward computed on the $\sigma$-algebra induced by the membership pattern of $\mathcal{P},\mathcal{Q}$. See Definition~\ref{def:coarsening} for more details.

We then prove a strengthened regularity lemma and, with it, a more powerful version of the indistinguishability result that uses coarsened games.

\begin{theorem}[Informal version of \Cref{thm:small-coarsening}]
For every communication game $G$ and every $\epsilon>0$, there are collections of sets $\mathcal{P}\subseteq 2^{\Omega_A}$ and $\mathcal{Q}\subseteq 2^{\Omega_B}$ with $\abs{\mathcal{P}},\abs{\mathcal{Q}}=\bO{\epsilon^{-2}}$ such that the coarsened game $G_{\mathcal{P},\mathcal{Q}}$ is $\epsilon$-indistinguishable from $G$. Moreover, such collections can be found in time $\mathrm{poly}(n,m,\epsilon^{-1})$.
\end{theorem}

The technical engine behind this statement is a stronger variant of the Frieze--Kannan weak regularity lemma we introduce in this paper. Unlike the usual weak regularity decomposition, the approximation here must be produced by actual row and column coarsenings, because we later need the compressed object to be a genuine communication game.

\begin{theorem}[Informal version of our Generalized Frieze-Kannan Regularity Lemma (\Cref{lem:grid-FK})]
Let $A^1,\dots,A^k\in[-1,1]^{n_1\times n_2}$ be matrices, let $\dist$ be any distribution over their entries, and let $s$ denote their total $\dist$-weighted squared mass.
There are collections of row sets $\mathcal P$ and column sets $\mathcal Q$, each of size $\bO{s/\epsilon^2}$, such that every residual $A^i-A^i_{\mathcal{P},\mathcal{Q}}$ has weighted cut norm at most $\epsilon$. The collections can be found efficiently.
\end{theorem}

The important point is that coarsening gives two kinds of control at once.

\begin{enumerate}
    \item For any rectangle $S\times T$ and action $a$, the expected reward of taking action $a$ on $S\times T$ differs by at most $\epsilon$ between $G$ and $G_{\mathcal{P},\mathcal{Q}}$.
    \item For any rectangle that is measurable with respect to the coarsening, the equality is exact: the average reward of every action on that rectangle is the same in $G$ and in $G_{\mathcal{P},\mathcal{Q}}$.
\end{enumerate}

The first property lets us compare the utility of rectangle-action communication between the original and coarsened games. The second property is what breaks the vicious cycle: once Alice and Bob communicate the coarsened observations, the protocol is measurable with respect to the coarsening, so its utility transfers back to the original game with no additional approximation loss.

This already gives a clean protocol for competing with the best one-bit communication scheme. Compute the $\epsilon$-indistinguishable coarsening $G_{\mathcal{P},\mathcal{Q}}$ with
$|\mathcal{P}|,|\mathcal{Q}| \in \bO{\epsilon^{-2}}$. Let Alice and Bob send the membership patterns of their observations in $\mathcal{P}$ and $\mathcal{Q}$, and let Charlie play the action with highest expected utility conditioned on those patterns. If $\pi^*$ is the best one-bit protocol in $G$, indistinguishability implies that $\pi^*$ has nearly the same value in $G_{\mathcal{P},\mathcal{Q}}$. Optimal play $\hat \pi$ in the coarsened game is at least as good as $\pi^*$ in the coarsened game. Using the fact that this is a coarsening, the value of $\hat \pi$ transfers back to $G$ with no additional loss. Thus coarsening gives an $\bO{\epsilon}$-additive approximation to the best one-bit protocol using only $\bO{1/\epsilon^2}$ bits and runtime $\mathrm{poly}(n, m, \epsilon^{-1})$.

\paragraph{From One-Bit Tests to General Short Protocols.}
The argument above focuses on competing with optimal one-bit protocols. To extend this to competing with protocols of any length we give a strong connection between indistinguishable communication games using the machinery of factorization norms and Grothendieck's inequality.

\begin{theorem}[Informal Combination of \Cref{lem:indist-weak} and \Cref{lem:indist-strong}]
\label{thm:informal:indist-strong-kbit}
    If $G$ and $\hat{G}$ are $\epsilon$-indistinguishable on the set of all test functions induced by one-bit protocols, then the value of running a $k$-bit protocol $\pi$ in $G$ and $\hat{G}$ cannot differ by more than $\epsilon\min\{2^k, 4K_{\mathrm{Gr}} \sqrt{m 2^k}\}$, where $m$ is the number of actions and $K_{\mathrm{Gr}}$ is the Grothendieck constant.
\end{theorem}
The first term in this bound can be achieved by a rather naive argument establishing that when communication games $G$ and $\hat{G}$ are $\epsilon$-indistinguishable on the set of all test functions induced by one-bit protocols, then $G$ and $\hat{G}$ are $(\epsilon \cdot 2^k)$-indistinguishable on the set of all test functions induced by $k$-bit protocols (see \Cref{lem:indist-weak}). The second term involves a more detailed analysis using factorization norms and Grothendieck's inequality to sharpen this bound to have a $\sqrt{2^k}$ dependence instead (see \Cref{lem:indist-strong}). This sharpening allows our upper and lower bounds to match in their dependence on the communication complexity $CC(G)$.

\paragraph{The General Tractable Protocol.}
 We can now state the full protocol for competing with the optimal protocol of any length.
Given a communication game $G$, let $CC_\alpha(G)$ be the number of bits that \emph{any protocol} (even one computed intractably) must communicate to ensure utility $\alpha$ in $G$. We aim to guarantee $\alpha-\delta$ utility with a computationally efficient protocol.
We first compute the indistinguishable coarsening $G_{\mathcal{P}, \mathcal{Q}}$ with $\abs{\mathcal{P}}, \abs{\mathcal{Q}} \leq \bO{1/\epsilon^2}$ for $$\epsilon^{-1} = \delta^{-1}\cdot
\min\left(4K_{\mathrm{Gr}} \sqrt{2^{CC_\alpha(G)}m},~ 2^{CC_\alpha(G)}\right).$$
Then, when Alice receives the observation $\omega_A$, she communicates to Charlie the membership pattern of $\omega_A$ in $\mathcal{P}$. Similarly, Bob communicates the membership pattern of $\omega_B$ in $\mathcal{Q}$. Charlie then takes the action with highest expected utility conditioned on this information. Because $\mathcal{P}$ and $\mathcal{Q}$ have at most $\bO{1/\epsilon^2}$ sets, the communication procedure uses $\bO{1/\epsilon^2}$ bits and runs in time polynomial in $n$, $m$, and $1/\epsilon$.

The proof that this procedure has high utility proceeds in three parts:
\begin{enumerate}
    \item Let $\pi^*$ be an optimal (intractable) protocol sending $CC_\alpha(G)$ bits and achieving utility $\alpha$ in $G$. By \Cref{thm:informal:indist-strong-kbit},
    $\pi^*$ achieves utility at least $\alpha-\epsilon\cdot\min\{2^{CC_\alpha(G)},4K_{\mathrm{Gr}}\sqrt{m 2^{CC_\alpha(G)}}\}\geq \alpha - \delta$ in $G_{\mathcal{P}, \mathcal{Q}}$.
    \item The protocol $\hat{\pi}^*$ that communicates the membership patterns in $\mathcal{P}$ and $\mathcal{Q}$ and then plays the best action on the resulting atom is optimal in $G_{\mathcal{P}, \mathcal{Q}}$, so it performs at least as well as $\pi^*$ there. Thus it achieves utility at least $\alpha-\delta$.
    \item By construction, the expected utility of $\hat{\pi}^*$ in $G$ is exactly the same as its utility in $G_{\mathcal{P},\mathcal{Q}}$, which is $\geq \alpha-\delta$.
\end{enumerate}

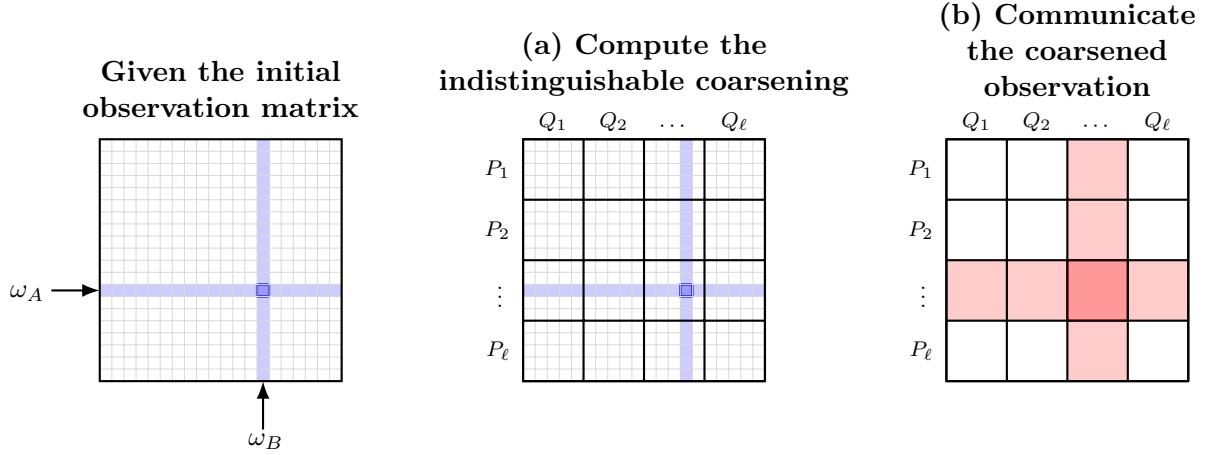
\begin{figure}
    \centering
    \makebox[\textwidth][c]{
        \input{figures/Frieze-Kannan}
    }
    \caption{A visual representation of our communication protocol. The arrows denote the observations Alice and Bob receive. In (a), \Cref{thm:small-coarsening} is used to compute a coarsening of the communication game that is $\epsilon$-indistinguishable. In (b), Alice and Bob then send their membership patterns in the coarsening to Charlie. The rightmost matrix represents the coarsened game after communication. Finally, Charlie takes the best action conditioned on this information.}
    \label{fig:FK-visualization}
\end{figure}

%% file: figures/Frieze-Kannan.tex
\begin{tikzpicture}[
    scale=0.8,
    every node/.style={font=\small},
    matrix_style/.style={draw, thick, minimum size=4cm},
    highlight/.style={fill=blue!20},
    intersection/.style={fill=blue!40, draw=blue!80!black, thick},
    coarse_highlight/.style={fill=red!20},
    coarse_intersection/.style={fill=red!40, draw=red!80!black, thick},
    grid_line/.style={draw=gray!30, very thin},
    partition_line/.style={draw=black, thick},
    arrow_style/.style={-Latex, thick},
    part_label/.style={font=\scriptsize, fill=white, inner sep=1.2pt}
]

    \def\matSize{4}
    \def\cellSize{0.2}

    \def\selRowY{1.4}
    \def\selColX{2.6}

    \def\coarseRowY{1}
    \def\coarseColX{2}
    \def\coarseCell{1}

    \begin{scope}[local bounding box=mat1]
        \node[above, align=center, font=\bfseries] at (\matSize/2, \matSize+0.2) { Given the initial \\ observation matrix};

        \fill[highlight] (0, \selRowY) rectangle (\matSize, \selRowY+\cellSize);
        \fill[highlight] (\selColX, 0) rectangle (\selColX+\cellSize, \matSize);
        \fill[intersection] (\selColX, \selRowY) rectangle (\selColX+\cellSize, \selRowY+\cellSize);

        \draw[grid_line, step=\cellSize] (0,0) grid (\matSize,\matSize);
        \draw[thick] (0,0) rectangle (\matSize,\matSize);

        \node[] at (-1.2, \selRowY +  \cellSize/2-0.05) {$\omega_A$};
        \node[] at (\selColX + \cellSize/2 + 0.05, -1) {$\omega_B$};
        \draw[arrow_style] (-0.8, \selRowY + \cellSize/2) -- (0, \selRowY + \cellSize/2);
        \draw[arrow_style] (\selColX + \cellSize/2, -0.8) -- (\selColX + \cellSize/2, 0);
    \end{scope}

    \begin{scope}[shift={(7,0)}, local bounding box=mat2]
        \node[above, align=center, font=\bfseries] at (\matSize/2, \matSize+0.55) {(a) Compute the \\indistinguishable coarsening};

        \fill[highlight] (0, \selRowY) rectangle (\matSize, \selRowY+\cellSize);
        \fill[highlight] (\selColX, 0) rectangle (\selColX+\cellSize, \matSize);
        \fill[intersection] (\selColX, \selRowY) rectangle (\selColX+\cellSize, \selRowY+\cellSize);

        \draw[grid_line, step=\cellSize] (0,0) grid (\matSize,\matSize);
        \draw[partition_line, step=\coarseCell] (0,0) grid (\matSize,\matSize);
        \draw[thick] (0,0) rectangle (\matSize,\matSize);

        \foreach \i/\lab in {3/P_1,2/P_2,0/P_\ell} {
            \node[part_label, anchor=east] at (-0.15, \i+0.5) {$\lab$};
        }
        \node[part_label] at (-0.38, 1.5) {$\vdots$};
        \foreach \i/\lab in {0/Q_1,1/Q_2,2/\cdots,3/Q_\ell} {
            \node[part_label, anchor=south] at (\i+0.5, 4.08) {$\lab$};
        }
    \end{scope}

    \begin{scope}[shift={(14,0)}, local bounding box=mat3]
        \node[above, align=center, font=\bfseries] at (\matSize/2, \matSize+0.55) {(b) Communicate \\the coarsened \\ observation};

        \fill[coarse_highlight] (0, \coarseRowY) rectangle (\matSize, \coarseRowY+\coarseCell);
        \fill[coarse_highlight] (\coarseColX, 0) rectangle (\coarseColX+\coarseCell, \matSize);
        \fill[coarse_intersection] (\coarseColX, \coarseRowY) rectangle (\coarseColX+\coarseCell, \coarseRowY+\coarseCell);

        \draw[partition_line, step=\coarseCell] (0,0) grid (\matSize,\matSize);
        \draw[thick] (0,0) rectangle (\matSize,\matSize);

        \foreach \i/\lab in {3/P_1,2/P_2,0/P_\ell} {
            \node[part_label, anchor=east] at (-0.15, \i+0.5) {$\lab$};
        }
        \node[part_label] at (-0.38, 1.5) {$\vdots$};
        \foreach \i/\lab in {0/Q_1,1/Q_2,2/\cdots,3/Q_\ell} {
            \node[part_label, anchor=south] at (\i+0.5, 4.08) {$\lab$};
        }
    \end{scope}

\end{tikzpicture}

%% file: sections/03-Preliminaries.tex
\section{Preliminaries}

\subsection{Notation}

For a set $X$, let $2^X$ denote the set of all subsets, and let $\Delta(X)$ denote the set of all possible probability distributions over $X$. For any natural number $\ell \in \mathbb{N}$, let $[\ell] = \{1, \dots, \ell\}$. Given a matrix $A \in \mathbb{R}^{n \times m}$, we denote its $(i, j)$ entry by $A(i, j)$.

Given two matrices, $A, B \in \mathbb{R}^{n_1 \times n_2}$, the Hadamard product $A \circ B$ is defined entrywise by $(A \circ B)(i, j) = A(i, j) \cdot B(i, j)$. Given a set $S$ in a universe $\mathcal{U}$, the complement $\mathcal{U} \setminus S$ is written $S^C$.

\subsection{The Communication Game}

The central object of study in this paper is the communication game. Three agents --- Alice, Bob, and Charlie --- are playing a game of common-interest together. Alice and Bob each receive some hidden observation, Alice observing $\omega_A \in \Omega_A$ and Bob observing $\omega_B \in \Omega_B$. Charlie is a bystander who receives no private information but must take an action from $\mathcal{A}$. A communication game starts with Alice and Bob receiving a draw from a distribution over observations $\dist \in \Delta(\Omega_A \times \Omega_B)$. Alice and Bob then exchange bits back and forth. Finally, Charlie takes an action $a \in \mathcal{A}$, and all the agents receive the common reward $r(a; \omega_A, \omega_B)$.\footnote{Modeling actions as being taken by an observer who only sees the transcript is akin to the open communication model of \cite{fontes2023communication}. Communication protocols derived for this model can be used in the setting where Charlie's identity is the same as either of Alice or Bob, i.e., the decision is taken by one of the agents with partial information.}

\begin{definition}[Communication Game]
    A communication game is a tuple $G = \langle \acts, \Omega_A, \Omega_B, \dist, r\rangle$, where
    \begin{itemize}
        \item $\acts$ is the set of actions,
        \item $\Omega_A$ is the set of observations for Alice and $\Omega_B$ is the set of observations for Bob,
        \item $\dist$ is a joint distribution over $\Omega_A \times \Omega_B$,
        \item $r: \acts \times \Omega_A \times \Omega_B \to [-1, 1]$ is the common objective the agents are attempting to maximize.
    \end{itemize}
\end{definition}

Throughout this paper, we refer to $\abs{\Omega_A}$ as $n_A$, to $\abs{\Omega_B}$ as $n_B$, to $\abs{\acts}$ as $m$, and we write $n = \max(n_A, n_B)$.
For convenience in the analysis below, we will assume rewards are $\ell_2$-normalized,\footnote{This assumption that the rewards are $\ell_2$-normalized holds, for example, in the standard formulation of communication complexity in which Alice and Bob jointly compute a function $f(\omega_A, \omega_B)$, where $r(f(\omega_A, \omega_B); \omega_A, \omega_B) = 1$ and the reward of any other action is 0.} so that
\begin{equation*}
    \sum_{a \in \mathcal{A}} \E_{\omega_A, \omega_B \sim \dist}\left[r(a; \omega_A, \omega_B)^2\right] \leq 1.
\end{equation*}
This assumption is not restrictive; our results can apply to arbitrary rewards after rescaling them, though the bounds then acquire an appropriate dependence.

For Charlie to take useful actions, Alice and Bob need to communicate their hidden information. They do this by running a communication policy.

\begin{definition}[Communication Policy]
    \label{def:communication-policy}
    A communication policy $\pi$ for a communication game $G$ is a finite rooted binary tree in which each internal node is labeled by
    \begin{enumerate}
        \item an agent $i \in \{A, B\}$, and
        \item a deterministic message function $\sigma: \Omega_i \to \{0, 1\}$,
    \end{enumerate}
    and each leaf is labeled by an action $a \in \mathcal{A}$.

    On an observation pair $(\omega_A, \omega_B)$ in a communication game $G$, the policy is run as follows. Alice and Bob start at the root of the binary tree. On a node labeled with agent $i$:
    \begin{enumerate}
        \item If the current node is internal, agent $i$ sends the message $m_i = \sigma(\omega_i)$. If $m_i$ is 0, the agents descend to the left subtree. If $m_i$ is 1, the agents descend to the right subtree.
        \item Otherwise, when the agents are on a leaf node instead, Charlie plays an action $a$ on the leaf node.
    \end{enumerate}

    The action played by $\pi$ on observation pair $(\omega_A, \omega_B)$ is written $\pi(\omega_A, \omega_B)$.
    The value of a communication policy $\pi$ on a game $G$ is  $$\pi(G) = \E_{\omega_A, \omega_B \sim \dist}[r(\pi(\omega_A, \omega_B); \omega_A, \omega_B)].$$ 
    The length of a communication policy $\pi$, written $CC(\pi)$, is the height of the associated binary tree. The number of leaves in the tree is written $N(\pi)$. In particular $$N(\pi) \leq 2^{CC(\pi)}.$$
\end{definition}

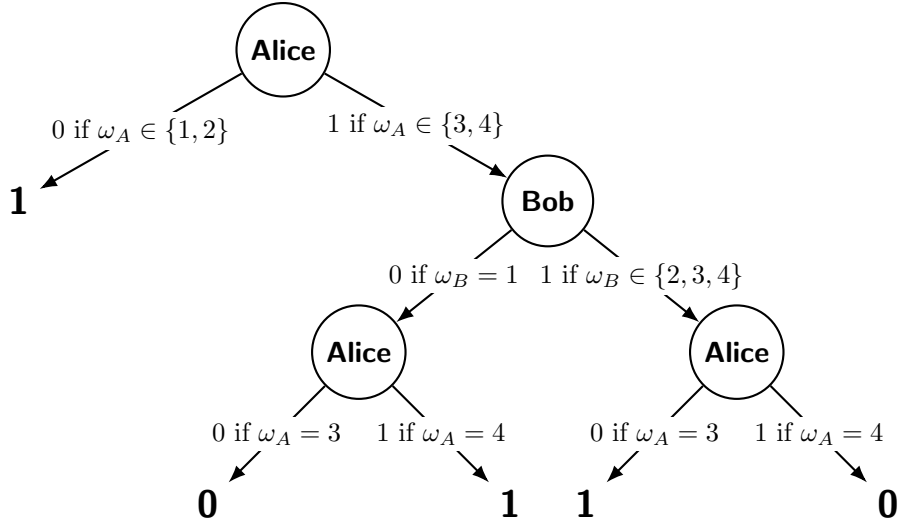
\begin{figure}
    \centering
    \input{figures/communication_tree}
    \caption{An illustration of a communication policy in binary tree form.}
    \label{fig:communication_tree}
\end{figure}

Without loss of generality, we do not allow randomness in Charlie's actions or in the bits Alice or Bob send. Indeed, any randomized policy can be written as a distribution over deterministic policies. So if there exists a randomized communication policy achieving utility $\alpha$, there must exist some deterministic policy achieving utility at least $\alpha$.

\begin{remark}[The Action Taker is Irrelevant]
    One could have defined a communication game in which either Alice or Bob takes the action instead. Because Charlie has no private information and only observes the public transcript, both Alice and Bob can always simulate the action Charlie would take. This means that any short policy in this model automatically translates to a short policy when either agent is an action taker. 

    Conversely, the shortest communication policy when Charlie is the action taker is never more than $\ceil{\log(m)}$ bits longer than the shortest policy when Alice or Bob is the action taker. Indeed, when the action to be taken is known, Alice or Bob can simply communicate that action to Charlie using $\ceil{\log(m)}$ bits.
\end{remark}

It will be useful to have a notion of a subgame of a communication game.
\begin{definition}[Subgames]
    Given a communication game $G = \langle \mathcal{A}, \Omega_A, \Omega_B, \dist, r\rangle$ and subsets of observations $\Omega_A' \subseteq \Omega_A$ and $\Omega_B' \subseteq \Omega_B$ where $\Pr_\dist(\Omega_A', \Omega_B')$ is non-zero, the subgame $G'$ of $G$ induced by $\Omega_A'$ and $\Omega_B'$ is the communication game $G' = \langle \mathcal{A}, \Omega_A', \Omega_B', \dist', r\rangle$, where $\dist'$ is the distribution $\dist$ conditioned on the event $\omega_A \in \Omega_A'$ and $\omega_B \in \Omega_B'$:
    \begin{equation*}
        \dist'(\omega_A, \omega_B) = \Pr_\dist(\omega_A, \omega_B \mid \Omega_A', \Omega_B').
    \end{equation*}
\end{definition}

In summary, we can describe the process of playing a communication game as in the box below.

\begin{tcolorbox}[colback=gray!5!white, colframe=gray!75!black, title=Running policy $\pi$ on game $G$]
A communication policy $\pi$ is run on a communication game $G$ as follows:
\begin{enumerate}
    \item \textbf{Observation phase:} The environment samples $(\omega_A, \omega_B) \sim \dist$, shows Alice $\omega_A$ and shows Bob $\omega_B$.
    \item \textbf{Cheap-talk phase:} The agents run the communication policy, taking turns relaying messages to each other.
    \item \textbf{Action phase:} Charlie observes the transcript, plays an action $a \in \acts$, and the agents receive the reward $r(a; \omega_A, \omega_B)$.
\end{enumerate}
\end{tcolorbox}

Finally, we define a \textit{communication protocol} to be a mapping from communication games to communication policies. The goal of this work is to describe communication protocols that can be run in polynomial time, that when given a communication game as input, output a short length communication policy of high value.

\subsection{Communication Complexity}

We now define the notion of the communication complexity of a communication game. This is a measure of how distributed the information required to take good decisions is throughout the system.

\begin{definition}[Communication Complexity]
    For any $\alpha \in [-1, 1]$ and communication game $G$, let $\mathcal{C}_\alpha(G)$ be the set of communication policies that achieve a reward of at least $\alpha$. The communication complexity of $G$ is the minimum-length communication policy that achieves an expected reward of at least $\alpha$,
    \begin{equation*}
        CC_\alpha(G) = \min_{\pi \in \mathcal{C}_\alpha(G)} CC(\pi).
    \end{equation*}
    If $\mathcal{C}_\alpha(G)$ is empty, then we set $CC_\alpha(G) = \infty$.
    Let $$\alpha^*(G) = \E_{\omega_A, \omega_B \sim \dist}\left[\max_{a \in \acts} r(a; \omega_A, \omega_B)\right]$$ be the utility achieved by playing perfectly in communication game $G$. We will write $CC(G)$ to mean $CC_{\alpha^*(G)}(G)$, the number of bits required to achieve maximal utility. 
\end{definition}

A particularly useful measure of the complexity of a communication game $G$ is the partition number. Imagine partitioning the observation space $\Omega_A \times \Omega_B$ with rectangles (axis-aligned subsets of the observation space).

The value of a rectangular partitioning is the best utility in expectation that Charlie can achieve when only being told the rectangle that the current observation pair lies in.

\begin{definition}[Rectangular Partition]
    \label{def:rectangular-partition}
    A rectangular partition of the observation space $\Omega_A \times \Omega_B$ is a set of products $R_1 \times C_1, \dots, R_\ell \times C_\ell$ where $R_i \subseteq \Omega_A$ and $C_i \subseteq \Omega_B$ for all $i \in [\ell]$, the rectangles are disjoint, so that for all $i \neq j$, $(R_i \times C_i) \cap (R_j \times C_j) = \emptyset$, and the rectangles cover the space,
    \begin{equation*}
        \bigcup_{i = 1}^\ell R_i \times C_i = \Omega_A \times \Omega_B.
    \end{equation*}
    Let $\mathcal{R}_\ell$ be the set of all possible rectangular partitions with $\ell$ rectangles.
\end{definition}

We will define the utility of a rectangular partition to be how actionable the partition is.

\begin{definition}[Partition Number]
    \label{def:covering-number}
    For a rectangular partition $D = \{R_1 \times C_1, \dots, R_\ell \times C_\ell\}$, let $\Pi_D$ be the set of policies $\pi:\Omega_A\times\Omega_B\to\acts$ that are constant on every rectangle of $D$. That is, for every $R_i\times C_i\in D$, there exists an action $a_i\in\acts$ such that $\pi(\omega_A,\omega_B)=a_i$ for all $(\omega_A,\omega_B)\in R_i\times C_i$. We call these the $D$-measurable policies.

    Given a policy $\pi \in \Pi_D$, define $\pi(G) = \E_{\omega_A, \omega_B \sim \dist}\left[r(\pi(\omega_A, \omega_B); \omega_A, \omega_B)\right]$. The value of the rectangle partition $D$ is defined to be
    \begin{equation*}
        U(D) = \max_{\pi \in \Pi_D} \pi(G).
    \end{equation*}

    The $\alpha$-partition number of a communication game $G$, written $N_\alpha(G)$, is the minimum $\ell$ such that a rectangular partition with $\ell$ parts has value at least $\alpha$. We will write $N(G)$ to mean $N_{\alpha^*(G)}(G)$, the size of the smallest rectangular partition that achieves maximal utility. If no policy achieves value $\alpha$, we write $N_\alpha(G) = \infty$.
\end{definition}

Every communication policy that achieves a utility of $\alpha$ on average induces an $\alpha$-partitioning of the observation space. The intuition for this, which we will make formal below, is that the subset of observation pairs that arrive at a leaf node of the binary tree corresponding to a policy $\pi$ is always a rectangle. This allows us to relate the following two quantities.

\begin{proposition}[Rectangle Decomposition of Communication Policies]
    \label{prop:communication-tiles}
    For any communication game $G$ and any communication policy $\pi$, there exists a rectangular partition $D = \{R_1 \times C_1, \dots, R_\ell \times C_\ell\}$ of $\Omega_A \times \Omega_B$ with $\ell = N(\pi)$ such that the action taken by $\pi$ is constant on each rectangle $R_i \times C_i$.
    
    It follows that for any communication game $G$ and any $\alpha \in [-1, 1]$,
    $$N_\alpha(G) \leq 2^{CC_\alpha(G)}.$$
\end{proposition}
\begin{proof}
    We first show that for every node $v$ of the policy tree, the set of observations passing through $v$ is a rectangle.

    Let $B = (V, E)$ be the binary tree corresponding to $\pi$, and define $f: V \to 2^{\Omega_A \times \Omega_B}$ to be the function mapping $v$ to the set of observation pairs passing through $v$ when the protocol is run on them.

    Let $v_{\mathrm{root}}$ be the root node of $B$. Because every observation pair in $\Omega_A \times \Omega_B$ starts here, $f(v_{\mathrm{root}}) = \Omega_A \times \Omega_B$.
    
    We continue by induction on the tree. Suppose that for a node $v$, we can write $f(v) = S \times T$ for $S \subseteq \Omega_A$ and $T \subseteq \Omega_B$. Let $v_0$ be the left child of $v$ and $v_1$ be the right child. Suppose, without loss of generality, that it is Alice's turn on node $v$ and let $\sigma_v: \Omega_A \to \{0, 1\}$ be its messaging function. Let $S_0 = \{\omega_A \in S \mid \sigma_v(\omega_A) = 0\}$ and $S_1 = \{\omega_A \in S \mid \sigma_v(\omega_A) = 1 \}$. By definition, because the agents descend to the left when $\omega_A \in S_0$ and to the right when $\omega_A \in S_1$, it must hold that $f(v_0) = S_0 \times T$ and $f(v_1) = S_1 \times T$. If it is Bob's turn, the same argument applies by splitting $T$ instead.
    
    Let $L = \{v_1, \dots, v_{N(\pi)}\}\subseteq V$ be the set of leaf nodes in $B$. This implies that $f(v_i) = R_i \times C_i$ for some $R_i \subseteq \Omega_A$ and $C_i \subseteq \Omega_B$. Because communication is deterministic, every observation pair must end at exactly one leaf node. Therefore, all the $R_i \times C_i$ are pairwise disjoint and cover all of $\Omega_A \times \Omega_B$.

    Because Charlie takes a deterministic action at every leaf node, the policy $\pi$ must play a fixed action at every set $R_i \times C_i$, and there are $N(\pi)$ such rectangles, as claimed.

    For the second claim, let $\pi$ be the $CC_\alpha(G)$-bit protocol achieving utility $\alpha$. By our argument above, there exists an $N(\pi) \leq 2^{CC_\alpha(G)}$ sized rectangular partition $D$ on which $\pi$ takes constant actions. Therefore there exists a $2^{CC_\alpha(G)}$ sized rectangular partition on which playing some constant action in each part results in average utility at least $\alpha$, i.e., $N_\alpha(G) \leq 2^{CC_\alpha(G)}$ as claimed.
\end{proof}

\begin{figure}
  \centering
  \input{figures/rectangular_tiling}
  \caption{An illustration of \Cref{prop:communication-tiles}, the associated rectangular partition of \Cref{fig:communication_tree}.}
  \label{fig:tiling}
\end{figure}

\subsection{\texorpdfstring{$\sigma$-Algebras}{Sigma-Algebras}}

The arguments in this paper will rely on the language of $\sigma$-algebras.
A $\sigma$-algebra on a set $\Omega$ is a nonempty collection $\mathcal{P}$ of subsets of $\Omega$ that are closed under complement, unions, and intersections. A universe $\Omega$ together with a $\sigma$-algebra $\mathcal{P}$ defines a measurable space. 

For a $\sigma$-algebra $\mathcal{P}$, we say that a set $S \subseteq \Omega$ is $\mathcal{P}$-measurable when $S \in \mathcal{P}$. The \emph{atoms} of a $\sigma$-algebra are the nonempty sets in the $\sigma$-algebra that cannot be split further into smaller nonempty measurable sets.
 
Equivalently, a $\sigma$-algebra can be thought of as being generated by membership patterns in a collection of subsets of $\Omega$ as we describe below.

\begin{definition}[$\sigma$-algebra]
    \label{def:sigma-algebra}
    Consider a list of sets $X_1, \dots, X_k \subseteq \Omega$.
    For any $x \in \Omega$, the membership pattern of $x$ with respect to $X_1, \dots, X_k$ is the binary string $b \in \{0, 1\}^k$ whose $i$th bit $b_i$ is 1 if $x \in X_i$ and 0 otherwise.

    For any membership pattern $b \in \{0, 1\}^k$, the atom induced by $b$, written $P_b \subseteq \Omega$, is the set of all elements in $\Omega$ with the membership pattern $b$. The $\sigma$-algebra generated by $X_1, \dots, X_k$ is defined to be the set of possible unions of atoms,
    \begin{equation*}
        \sigma(X_1, \dots, X_k) = \left\{\bigcup_{b \in B} P_b \mid B \subseteq \{0, 1\}^k\right\}.
    \end{equation*}
\end{definition}

Given two $\sigma$-algebras $\mathcal{P}$ and $\mathcal{Q}$, we can define the product $\sigma$-algebra, denoted $\mathcal{P} \otimes \mathcal{Q}$, as being the $\sigma$-algebra generated by the sets $P \times Q$ for all $P \in \mathcal{P}, Q \in \mathcal{Q}$. Equivalently, if $\mathcal{P}$ has atoms $P_1,\dots,P_r$ and $\mathcal{Q}$ has atoms $Q_1,\dots,Q_s$, then $\mathcal{P}\otimes\mathcal{Q}$ consists of all unions of atom products $P_i\times Q_j$.

We can define the operation of adding information to a $\sigma$-algebra with the notion of a refinement.

\begin{definition}[Refinement]
    Given two $\sigma$-algebras $\mathcal{P}$ and $\mathcal{P}'$ over a set $\Omega$, $\mathcal{P}$ is said to refine $\mathcal{P}'$, written $\mathcal{P}' \subseteq \mathcal{P}$ when every $\mathcal{P}'$-measurable set is $\mathcal{P}$-measurable. A sequence of $\sigma$-algebras $\mathcal{P}_1, \dots, \mathcal{P}_k$ is said to be a filtration when each $\mathcal{P}_{i} \subseteq \mathcal{P}_{i + 1}$.
\end{definition}

One key operation we will perform with $\sigma$-algebras is to take expectations conditioned on them.
Throughout this subsection, expectations are taken with respect to an ambient distribution $\mu$ on $\Omega$.

\begin{definition}[Expectation Conditioned on a $\sigma$-Algebra]
    Given a $\sigma$-algebra $\mathcal{P}$ over $\Omega$ with atoms $P_1, \dots, P_k$ and a random variable $X: \Omega \rightarrow \mathbb{R}$, the expectation of $X$ conditioned on $\mathcal{P}$, written $\E[X \mid \mathcal{P}]$, is the random variable that is constant on each atom of $\mathcal{P}$ and equal to the average value of $X$ on that atom. That is, it is a function mapping any $x \in \Omega$ to
    \begin{equation*}
        \E[X \mid \mathcal{P}](x) = \sum_{i \in [k]: \Pr_\dist(P_i) > 0} \E[X \mid P_i] \cdot \mathds{1}(x \in P_i),
    \end{equation*}
    where $\E[X \mid P_i]$ is taken with respect to an ambient distribution $\mu$.
\end{definition}

A random variable $X$ defined over a set $\Omega$ is said to be measurable with respect to a $\sigma$-algebra $\mathcal{P}$ when it is constant on every atom of $\mathcal{P}$. Equivalently, $\E[X \mid \mathcal{P}] = X$ with probability one.

Expectations conditioned on a $\sigma$-algebra satisfy the handy tower property: for any $\sigma$-algebras $\mathcal{P}$, $\mathcal{P}'$ such that one refines the other $\mathcal{P}' \subseteq \mathcal{P}$,
\begin{equation*}
    \E[\E[X \mid \mathcal{P}] \mid \mathcal{P}'] = \E[X \mid \mathcal{P}'].
\end{equation*}
It will be useful to have the notion of martingales defined.

\begin{definition}[Martingales]
    Let $\mathcal{P}_1, \mathcal{P}_2, \dots$ be a filtration, i.e., a sequence of $\sigma$-algebras such that
    $\mathcal{P}_i \subseteq \mathcal{P}_{i + 1}$. The sequence of random variables $X_1, X_2, \dots$ is called a martingale with respect to the filtration when for any $i$, $X_i$ is $\mathcal{P}_i$-measurable and
    \begin{equation*}
        \E[X_{i + 1} \mid \mathcal{P}_i] = X_i.
    \end{equation*}
\end{definition}

Finally, we will need the orthogonality of martingale increments. Suppose that the sequence of random variables $X_1, X_2, \dots$ satisfies $\E[X_i^2] < \infty$ for all $i$, and is a martingale with respect to the filtration $\mathcal{P}_1, \mathcal{P}_2, \dots$. The orthogonality of martingale increments is the statement that for all $i$,
\begin{equation*}
    \E[X_i (X_{i + 1} - X_i)] = 0.
\end{equation*}

\subsection{Cut Norms}

In the analysis below, we will rely on the notion of a cut norm. This measures the rectangle in the matrix with the largest absolute sum.
\begin{definition}[Cut Norm]
    \label{def:cut-norm}
    Given a matrix $A \in \mathbb{R}^{n_1 \times n_2}$,  the cut norm is defined as
    \begin{equation*}
        \norm{A}_\square := \max_{\substack{R \subseteq [n_1] \\C \subseteq[n_2]}} \abs{\sum_{\substack{x \in R\\ y \in C}} A(x, y)}.
    \end{equation*}
    Given a probability distribution $\dist  \in \Delta([n_1] \times [n_2])$, the weighted cut norm is defined as
    \begin{equation*}
        \cutnorm{A} := \max_{\substack{R \subseteq [n_1] \\C \subseteq[n_2]}} \abs{\sum_{\substack{x \in R\\ y \in C}} A(x, y) \dist(x, y)}.
    \end{equation*}
\end{definition}

Finding the cut norm exactly is computationally intractable in the worst case \citep{alon2004approximating}, but algorithms have been proposed for approximating the largest absolute sum rectangle in a matrix. We adapt these algorithms for computing the weighted cut norm too.

\begin{lemma}[Approximating the Weighted Cut Norm Efficiently based on  \citep{alon2004approximating}]
    \label{lem:Alon-Naor}
   There exists a universal constant $c > 0$ such that given a matrix $A \in \mathbb{R}^{n_1 \times n_2}$ and a probability distribution $\dist \in \Delta([n_1] \times [n_2])$, there exists a deterministic algorithm running in $\mathrm{poly}(n_1, n_2)$ time that outputs subsets $R \subseteq [n_1]$ and $C \subseteq [n_2]$ with
    \begin{equation*}
        \abs{\sum_{\substack{x \in R \\ y \in C}} A(x, y)\dist(x, y)} \geq c \cutnorm{A}.
    \end{equation*}
\end{lemma}
\begin{proof}
    This is the Alon-Naor cut norm algorithm given in \cite[Section 5]{alon2004approximating}. The difference is that they prove this for the unweighted cut norm.

    Define $W_\dist$ to be the matrix of probabilities, where $(W_\dist)(i, j) = \dist(i, j)$. By definition, $\cutnorm{A} = \norm{W_\dist \circ A}_\square$. Therefore, in the distributional case, it is enough to run the Alon-Naor algorithm on the Hadamard product of $A$ and $W_\dist$ to compute an approximation of $\norm{W_\dist \circ A}_\square$.
\end{proof}

%% file: figures/communication_tree.tex
\begin{tikzpicture}[
    level distance=2cm,
    level 1/.style={sibling distance=7cm},
    level 2/.style={sibling distance=5cm},
    level 3/.style={sibling distance=4cm},
    edge from parent/.style={draw, -{Latex}, thick},
    player/.style={draw, circle, thick, minimum size=1.2cm, align=center, font=\sffamily\bfseries},
    leaf/.style={font=\sffamily\bfseries\Large},
    label/.style={midway, fill=white, font=\small, inner sep=2pt}
]

\node[player] {Alice}
    child {
        node[leaf] {1}
        edge from parent node[label] {0 if $\omega_A \in \{1,2\}$}
    }
    child {
        node[player] {Bob}
        child {
            node[player] {Alice}
            child {
                node[leaf] {0}
                edge from parent node[label] {0 if $\omega_A=3$}
            }
            child {
                node[leaf] {1}
                edge from parent node[label] {1 if $\omega_A=4$}
            }
            edge from parent node[label] {0 if $\omega_B=1$}
        }
        child {
            node[player] {Alice}
             child {
                node[leaf] {1}
                edge from parent node[label] {0 if $\omega_A=3$}
            }
            child {
                node[leaf] {0}
                edge from parent node[label] {1 if $\omega_A=4$}
            }
            edge from parent node[label] {1 if $\omega_B \in \{2,3,4\}$}
        }
        edge from parent node[label] {1 if $\omega_A \in \{3,4\}$}
    };

\end{tikzpicture}

%% file: figures/rectangular_tiling.tex
\begin{tikzpicture}[scale=1]
    \fill[gray!30] (0,0) rectangle (1,1);
    \fill[gray!30] (1,1) rectangle (4,2);
    \fill[gray!30] (0,2) rectangle (4,4);

    \draw[step=1cm, gray!20, thin] (0,0) grid (4,4);

    \draw[thick] (0,0) rectangle (4,4);
    
    \draw[very thick] (0,2) -- (4,2);
    
    \draw[very thick] (1,0) -- (1,2);
    
    \draw[very thick] (0,1) -- (4,1);

    \node[font=\Huge] at (2,3) {1};
    
    \node[font=\Large] at (0.5, 1.5) {0};
    \node[font=\Large] at (0.5, 0.5) {1};
    
    \node[font=\Large] at (2.5, 1.5) {1};
    \node[font=\Large] at (2.5, 0.5) {0};

    \foreach \i in {1,...,4} {
        \node at (\i-0.5, 4.2) {$\i$};
        \node at (-0.3, 4.5-\i) {$\i$};
    }

\end{tikzpicture}

%% file: sections/04-Frieze-Kannan-Communication.tex
\section{Indistinguishable Coarsenings and How to Find Them}
\label{sec:FK}
We begin by describing how to construct the abstractions used by our communication protocol. The starting point is that directly searching for a good short communication policy is combinatorially hard: even a one-bit policy for Alice corresponds to choosing an arbitrary subset of her observation space. Our approach is instead to replace the original game by a much smaller effective game that preserves the value of all short communication policies.

The right notion of approximation is not entrywise closeness of rewards. A short communication policy only accesses the games through more structured tests, such as rectangles of the observation space and actions taken on those rectangles. This motivates the following definition.

\subsection{Indistinguishable Coarsenings and the Main Construction Theorem}

\begin{definition}[$\epsilon$-Indistinguishable Communication Games]
    \label{def:indistinguishable}
    Consider two communication games $G$ and $\hat{G}$ with reward functions $r$ and $\hat{r}$ respectively, but over the same observation spaces $\Omega_A, \Omega_B$, distribution over observations $\dist$, and actions $\mathcal{A}$. For any $\epsilon >0$, $G$ and $\hat{G}$ are said to be $\epsilon$-indistinguishable when for any $R \subseteq \Omega_A$, $C \subseteq \Omega_B$, and $a \in \mathcal{A}$,
    \begin{equation*}
        \abs{\E_{\omega_A, \omega_B \sim \dist}\left[\mathds{1}(\omega_A \in R,~ \omega_B \in C)\cdot \left(r(a; \omega_A, \omega_B) - \hat{r}(a; \omega_A, \omega_B)\right)\right]} \leq \epsilon.
    \end{equation*}
\end{definition}

As we described in Section~\ref{sec:1.2}, indistinguishability alone is not enough for our protocol. If Alice and Bob communicate the full observation in the compressed game, that communication may reveal more information than the approximation was designed to preserve. We therefore introduce an abstraction that is a coarsening of the original game, so that the abstraction is faithful on communication policies much longer than a single bit. 

\begin{definition}[Coarsenings of Communication Games]
    \label{def:coarsening}
    Let $G$ be a communication game.
    Given collections of subsets $\mathcal{P} = \{P_1, \dots, P_{\ell_A}\}$ of $\Omega_A$ and $\mathcal{Q} = \{Q_1, \dots, Q_{\ell_B}\}$ of $\Omega_B$, the coarsening of $G$ defined by $\mathcal{P}, \mathcal{Q}$ is the communication game $G_{\mathcal{P}, \mathcal{Q}} = \langle \mathcal{A}, \Omega_A, \Omega_B, \mu, r_{\mathcal{P}, \mathcal{Q}}\rangle$ with reward function
    \begin{equation*}
        r_{\mathcal{P}, \mathcal{Q}}(a; \omega_A, \omega_B) = \E_\dist\left[r(a; \cdot, \cdot) \mid \sigma(\mathcal{P}) \otimes \sigma(\mathcal{Q})\right](\omega_A, \omega_B).
    \end{equation*}
\end{definition}

The main result of this section is that such coarsened abstractions always exist and can be found efficiently.

\begin{restatable}[Every Communication Game Has a Low Complexity $\epsilon$-Indistinguishable Coarsening]{theorem}{IndistCoarseningExists}
    \label{thm:small-coarsening}
    For any $\epsilon > 0$ and communication game $G$, there exist collections of subsets $\mathcal{P} \subseteq 2^{\Omega_A}$, $\mathcal{Q}\subseteq 2^{\Omega_B}$ with $\abs{\mathcal{P}}, \abs{\mathcal{Q}} \leq \bO{\epsilon^{-2}}$ such that the coarsening $G_{\mathcal{P}, \mathcal{Q}}$ is $\epsilon$-indistinguishable from $G$.
    Moreover, such $\mathcal{P}$ and $\mathcal{Q}$ can be found in $\mathrm{poly}(n, m, \epsilon^{-1})$ time.
\end{restatable}

Coarsened games are easy to communicate optimally in a few bits, at most $\bO{\epsilon^{-2}}$. This is because all Alice and Bob need to do is communicate which subsets in $\mathcal{P}$ and $\mathcal{Q}$ their observations lie in respectively. So \Cref{thm:small-coarsening} further implies that every communication game $G$ is $\epsilon$-indistinguishable from a communication game $\hat{G}$ with constant communication complexity $CC(\hat{G}) \leq \bO{\epsilon^{-2}}$.

The proof goes through our introduction of a generalized, weighted, simultaneous version of the Frieze-Kannan weak regularity lemma. In \Cref{sec:regularity}, we prove the matrix lemma behind this construction. Then, in \Cref{sec:indist-games}, we apply it simultaneously to all action-reward matrices of a communication game to prove \Cref{thm:small-coarsening}.

\subsection{A Generalized Weighted Simultaneous Frieze-Kannan Lemma}
\label{sec:regularity}

We will show a more general version of the coarsened weak regularity lemma that applies not just to communication but to matrices more generally. First, we must define what a coarsening of a matrix means.
\begin{definition}[Coarsenings of Matrices]
    \label{def:step-function-approximations}
    Consider an arbitrary matrix $A \in [-1, 1]^{n_1 \times n_2}$, a probability distribution $\dist \in \Delta([n_1] \times [n_2])$, a set of subsets of $[n_1]$, $\mathcal{P} = \{P_1, \dots, P_{\ell_1}\} \subseteq 2^{[n_1]}$, and a set of subsets of $[n_2]$, $\mathcal{Q} = \{Q_1, \dots, Q_{\ell_2}\} \subseteq 2^{[n_2]}$.

    The coarsening $A_{\mathcal{P}, \mathcal{Q}}$ with respect to $\dist$ is defined as the matrix such that for any $x, y \in [n_1] \times [n_2]$,
    \begin{align*}
        A_{\mathcal{P}, \mathcal{Q}}(x, y)
        &= \E_\dist\left[A \mid \sigma(\mathcal{P}) \otimes \sigma(\mathcal{Q})\right](x, y).
    \end{align*}
\end{definition}

The $\sigma$-algebra induced by $\mathcal{P} \times \mathcal{Q}$ could potentially contain $2^{2\ell}$ different atoms. Despite this, it is always possible, given collections of sets $\mathcal{P}$ and $\mathcal{Q}$ both of size at most $\ell$ each, to compute the coarsening in $\mathrm{poly}(\ell_1, \ell_2, n_1, n_2)$ time. One simple way of doing this is to take three passes over the matrix. In the first pass, all entries are labeled with a $2\ell$-bit membership string denoting whether or not the entry is contained in each $P_i$ and each $Q_i$. There are at most $n_1n_2$ entries in the matrix so at most $n_1n_2$ membership strings. In the second pass, compute the average value under $\dist$ for each membership string encountered, storing at most $n_1n_2$ averages. In the third pass, every entry is set to the average of entries with the same membership string.

With this notion in mind, we can state and prove the main technical tool that this work relies on.

\begin{theorem}[Efficiently Computing Simultaneous Regular Coarsenings]
    \label{lem:grid-FK}Consider an arbitrary list of matrices $A^1, \dots, A^k \in [-1, 1]^{n_1 \times n_2}$, a fidelity parameter $\epsilon > 0$, and a distribution $\dist \in \Delta([n_1] \times [n_2])$. Let $c > 0$ be a universal constant and let $s = \sum_{i = 1}^k \E_{x, y \sim \dist}\left[A^i(x, y)^2\right]$. There exist collections of subsets $\mathcal{P} \subseteq 2^{[n_1]}$ and $\mathcal{Q} \subseteq 2^{[n_2]}$ of size at most $\ceil{s/(c\epsilon)^2}$ each that simultaneously induce an $\epsilon$-indistinguishable coarsening for all $A^i$. That is, for all $i \in [k]$,
    \begin{equation*}
        \cutnorm{A^i - A^i_{\mathcal{P}, \mathcal{Q}}} \leq \epsilon.
    \end{equation*}
    These subsets can be found in $\bO{\mathrm{poly}(n_1, n_2) \cdot ks/\epsilon^2}$ time.
\end{theorem}
\begin{proof}
    We will run an energy increment argument in the language of martingales.

    \begin{algorithm}[ht]
    \caption{Efficiently Computing a Regular Coarsening}
    \label{alg:FK-grid}
    \begin{algorithmic}[1]
    \State Initialize $\mathcal{P}_0 = \mathcal{Q}_0 = \emptyset$, and set $t = 1$.
    \Loop
        \State Using the Alon-Naor algorithm (\Cref{lem:Alon-Naor}), for each $i \in [k]$, compute sets $P_t^i \subseteq [n_1]$ and $Q_t^i \subseteq [n_2]$ corresponding to the weighted cut norm of the residual $A^i - A^i_{\mathcal{P}_{t-1}, \mathcal{Q}_{t-1}}$. These satisfy
        \begin{equation*}
            V_t^i := \abs{\dist(P_t^i \times Q_t^i) \cdot \E_{x, y \sim \dist}\left[A^{i}(x, y) - A^{i}_{\mathcal{P}_{t - 1}, \mathcal{Q}_{t - 1}}(x, y) ~\middle|~ P_t^i \times Q_t^i\right]} \geq c \cdot \cutnorm{A^{i} - A^{i}_{\mathcal{P}_{t - 1}, \mathcal{Q}_{t - 1}}}.
        \end{equation*}
        \If{there exists some $i_t$ such that $V_t^{i_t} \geq c\epsilon$}
            \State Set $\mathcal{P}_{t} = \mathcal{P}_{t-1} \cup \{P_t^{i_t}\}$ and
            $\mathcal{Q}_{t} = \mathcal{Q}_{t-1} \cup \{Q_t^{i_t}\}$.
            \State Increment $t \gets t + 1$.
        \Else
            \State Break the loop and output $\mathcal{P}_{t-1}, \mathcal{Q}_{t-1}$.
        \EndIf
        \EndLoop
    \end{algorithmic}
    \end{algorithm}

    Consider \Cref{alg:FK-grid} for computing the partition with the desired property. The claim is that after $s/(c\epsilon)^2$ iterations, it must be that for all $i \in [k]$, the partition is a good approximation of the cut norm $\cutnorm{A^i - A^i_{\mathcal{P}_t, \mathcal{Q}_t}} < \epsilon$. The algorithm terminates quickly, and when it terminates, it must return a good coarsening by definition.

    Define the following potential function at time step $t$
    \begin{equation*}
        \Phi_t = \sum_{i = 1}^k \E_{x, y \sim \dist}\left[A^i_{\mathcal{P}_t, \mathcal{Q}_t}(x, y)^2\right],
    \end{equation*}
    to be the sum of the average squared entries under $\dist$ in the matrices $A^i_{\mathcal{P}_t, \mathcal{Q}_t}$.
    First, we will show that if there exists some $i$ such that $\cutnorm{A^i - A^i_{\mathcal{P}_t, \mathcal{Q}_t}} \geq \epsilon$, then the potential function must increase by $(c\epsilon)^2$. Then we will start by showing that this potential can never be larger than $s$. The result will follow from these two key facts.

    First, to study how the potential function increases at every iteration, we will prove the following claim.
    \begin{claim}
        \label{claim:key-claim}
        For every $i$ and iteration $t \geq 1$,
        \begin{equation*}
            \E_{x, y \sim \dist}\left[A^i_{\mathcal{P}_t, \mathcal{Q}_t}(x, y)^2\right] = \E_{x, y \sim \dist}\left[A^i_{\mathcal{P}_{t - 1}, \mathcal{Q}_{t - 1}}(x, y)^2\right] +
            \E_{x, y \sim \dist}\left[\left(A^i_{\mathcal{P}_t, \mathcal{Q}_t}(x, y) - A^i_{\mathcal{P}_{t - 1}, \mathcal{Q}_{t - 1}}(x, y)\right)^2\right].
        \end{equation*}
    \end{claim}
    The proof of \Cref{claim:key-claim} is deferred to \Cref{sec:proof-key-claim}.

    Notice that because $\E_{x, y \sim \dist}\left[\left(A^i_{\mathcal{P}_t, \mathcal{Q}_t}(x, y) - A^i_{\mathcal{P}_{t - 1}, \mathcal{Q}_{t - 1}}(x, y)\right)^2\right] \geq 0$, it follows from \Cref{claim:key-claim} that each refinement never decreases the second moment: for all $i$ and all $t$ it must be that $$\E_{x, y \sim \dist}\left[A^i_{\mathcal{P}_t, \mathcal{Q}_t}(x, y)^2\right] \geq \E_{x, y \sim \dist}\left[A^i_{\mathcal{P}_{t - 1}, \mathcal{Q}_{t - 1}}(x, y)^2\right].$$
    Refinement can never hurt the potential function.

    Now we will show that at every iteration, the potential function must increase by $(c\epsilon)^2$. We will do this by proving the following claim.
    \begin{claim}
        \label{claim:key-claim-2}
        For all iterations $t$, when $i_t$ is the index of the matrix chosen at iteration $t$,
        \begin{equation}
            \E_{x, y \sim \dist}\left[\left(A^{i_t}_{\mathcal{P}_{t}, \mathcal{Q}_{t}}(x, y) - A^{i_t}_{\mathcal{P}_{t - 1}, \mathcal{Q}_{t - 1}}(x, y)\right)^2\right] \geq (c\epsilon)^2.
        \end{equation}
    \end{claim}
    The proof of \Cref{claim:key-claim-2} is deferred to \Cref{sec:proof-key-claim-2}.

    Applying \Cref{claim:key-claim} and \Cref{claim:key-claim-2}, the potential function must increase by $(c\epsilon)^2$ each step:
    \begin{align*}
        \Phi_{t}
        &= \sum_{i = 1}^k \E_{x, y \sim \dist}\left[A^i_{\mathcal{P}_{t}, \mathcal{Q}_{t}}(x, y)^2\right] \\
        &= \sum_{i = 1}^k \left(\E_{x, y \sim \dist}\left[A^i_{\mathcal{P}_{t - 1}, \mathcal{Q}_{t - 1}}(x, y)^2\right] + \E_{x, y \sim \dist}\left[\left(A^i_{\mathcal{P}_{t}, \mathcal{Q}_{t}}(x, y) - A^i_{\mathcal{P}_{t - 1}, \mathcal{Q}_{t - 1}}(x, y)\right)^2\right]\right) && (\text{By \Cref{claim:key-claim}})\\
        &\geq \E_{x, y \sim \dist}\left[\left(A^{i_t}_{\mathcal{P}_{t}, \mathcal{Q}_{t}}(x, y) - A^{i_t}_{\mathcal{P}_{t - 1}, \mathcal{Q}_{t - 1}}(x, y)\right)^2\right] + \sum_{i = 1}^k \E_{x, y \sim \dist}\left[A^i_{\mathcal{P}_{t - 1}, \mathcal{Q}_{t - 1}}(x, y)^2\right] \\
        &\geq (c\epsilon)^2 + \Phi_{t - 1}. && (\text{By \Cref{claim:key-claim-2}})
    \end{align*}

    The final step of this argument is to prove that the potential function can never exceed $s$. Indeed, for any iteration $t$,
    \begin{align*}
        \Phi_t &= \sum_{i = 1}^k \E_{x, y \sim \dist}\left[\left(A^i_{\mathcal{P}_t, \mathcal{Q}_t}(x, y)\right)^2\right] \\
        &= \sum_{i = 1}^k \E_{x, y \sim \dist}\left[\E_\dist\left[A^i \mid \sigma(\mathcal{P}_t) \otimes \sigma(\mathcal{Q}_t)\right](x, y)^2\right] \\
        &\leq \sum_{i = 1}^k \E_{x, y \sim \dist}\left[\E_\dist\left[(A^i)^2 \mid \sigma(\mathcal{P}_t) \otimes \sigma(\mathcal{Q}_t)\right](x, y)\right] && (\text{By Jensen's inequality}) \\
        &= \sum_{i = 1}^k \E_{x, y \sim \dist}\left[A^i\left(x, y\right)^2\right] && (\text{By the tower rule})\\
        &= s.
    \end{align*}
    Therefore, after at most $t \leq s/(c\epsilon)^2$ iterations, there cannot exist any $i \in [k]$ such that $\cutnorm{A^i - A^i_{\mathcal{P}_t, \mathcal{Q}_t}} \geq \epsilon$.
\end{proof}

This result can be thought of as a multicalibration variant of the Frieze-Kannan cut decomposition (see \cite{hebert2018multicalibration} for an introduction to multicalibration) of the kind studied in work such as \cite{casacuberta2024complexity}.

\begin{remark}
    As a remark, the standard way to prove a coarsening variant of the regularity lemma in the unweighted case (e.g. in \cite{alon2002random, lovasz2007szemeredi}) is to first run the standard Frieze-Kannan regularity lemma \citep{frieze1999quick} to find an $\epsilon$-indistinguishable matrix that is the sum of $k \in \bO{1/\epsilon^2}$ matrices of the form $\mathds{1}_{S_1} \mathds{1}_{T_1}^T, \dots, \mathds{1}_{S_k} \mathds{1}_{T_k}^T$, where each $S_i \subseteq [n_1]$ and $T_i \subseteq [n_2]$ and $\mathds{1}_{S_i}$ and $\mathds{1}_{T_i}$ are the indicator vectors of the sets $S_i$ and $T_i$ respectively. Then, take the partitions induced by the $\sigma$-algebra $\mathcal{P}$ of the sets $S_1, \dots, S_k$ for the rows and $\mathcal{Q}$ of the sets $T_1, \dots, T_k$ for the columns. It can be proven that replacing the value of each grid cell in $\mathcal{P} \times \mathcal{Q}$ in the approximation with the average in the original matrix cannot worsen the approximation in the cut norm by more than a factor of 2, so coarsening by $\mathcal{P}$ and $\mathcal{Q}$ results in a coarsened $2\epsilon$-indistinguishable matrix.

    This argument fails when the distribution over entries in the matrix is not a product. Replacing the value that the approximation takes in each grid cell in $\mathcal{P} \times \mathcal{Q}$ with the weighted average in the original matrix can worsen the approximation in the weighted cut norm.

    Another interpretation of this failure is that when the distribution is a product, it can be shown that any $\epsilon$-multiaccurate partition is automatically a $2\epsilon$-multicalibrated one. This statement is not necessarily true for all possible distributions. It would be interesting to study other natural settings more generally in which multiaccuracy suffices for multicalibration.
\end{remark}

\begin{remark}
    The constant $c$ in this theorem is $0.03$. One can improve this to $c = 0.56$ with a randomized cut norm algorithm, as given in \cite{alon2004approximating}, that succeeds with high probability. In doing so, however, we suffer a multiplicative factor of $\log(1/\epsilon)$ in the runtime to get each step of boosting to succeed with high enough probability.
\end{remark}

\subsubsection{Proof of the Martingale Energy Identity Claim (Claim~\ref{claim:key-claim})}
\label{sec:proof-key-claim}

\begin{proof}[Proof of \Cref{claim:key-claim}]
    Notice that the sequence of $\sigma(\mathcal{P}_t) \otimes \sigma(\mathcal{Q}_t)$ defines a filtration by construction. This is because at every iteration of the algorithm, the row and column partition is refined.

    For every $i$, the sequence of $A^i_{\mathcal{P}_t, \mathcal{Q}_t}$ in fact defines a martingale.
    We show this below using the tower property and the fact that the sequence of $\sigma$-algebras is a filtration.
    \begin{align*}
        \E_\dist\left[A^i_{\mathcal{P}_t, \mathcal{Q}_t} \mid \sigma\left(\mathcal{P}_{t - 1}\right) \otimes \sigma\left(\mathcal{Q}_{t - 1}\right)\right]
        &= \E_\dist\left[\E_\dist\left[A^i \mid \sigma\left(\mathcal{P}_t\right) \otimes \sigma\left(\mathcal{Q}_t\right)\right] \mid \sigma\left(\mathcal{P}_{t - 1}\right) \otimes \sigma\left(\mathcal{Q}_{t - 1}\right)\right] \\
        &= \E_\dist\left[A^i \mid \sigma\left(\mathcal{P}_{t - 1}\right) \otimes \sigma\left(\mathcal{Q}_{t - 1}\right)\right] \\
        &= A^i_{\mathcal{P}_{t - 1}, \mathcal{Q}_{t - 1}}.
    \end{align*}
    Therefore, by the orthogonality of martingale increments, $$\E_{x, y \sim \dist}\left[A^i_{\mathcal{P}_{t - 1}, \mathcal{Q}_{t - 1}}(x, y)\cdot \left(A^i_{\mathcal{P}_t, \mathcal{Q}_t}(x, y) - A^i_{\mathcal{P}_{t - 1}, \mathcal{Q}_{t - 1}}(x, y)\right)\right] = 0.$$

    Hence, we get the key equality,
    \begin{align*}
        \E_{x, y \sim \dist}[A^i_{\mathcal{P}_t, \mathcal{Q}_t}(x, y)^2] &= \E_{x, y \sim \dist}\left[\left(A^i_{\mathcal{P}_{t - 1}, \mathcal{Q}_{t - 1}}(x, y) + \left(A^i_{\mathcal{P}_t, \mathcal{Q}_t}(x, y) - A^i_{\mathcal{P}_{t - 1}, \mathcal{Q}_{t - 1}}(x, y)\right)\right)^2\right]  \\
        &=\E_{x, y \sim \dist}\left[A^i_{\mathcal{P}_{t - 1}, \mathcal{Q}_{t - 1}}(x, y)^2\right] +
        \E_{x, y \sim \dist}\left[\left(A^i_{\mathcal{P}_t, \mathcal{Q}_t}(x, y) - A^i_{\mathcal{P}_{t - 1}, \mathcal{Q}_{t - 1}}(x, y)\right)^2\right] \\
        &\qquad+ 2\E_{x, y \sim \dist}\left[A^i_{\mathcal{P}_{t - 1}, \mathcal{Q}_{t - 1}}(x, y) \cdot \left(A^i_{\mathcal{P}_t, \mathcal{Q}_t}(x, y) - A^i_{\mathcal{P}_{t - 1}, \mathcal{Q}_{t - 1}}(x, y)\right)\right] \\
        &= \E_{x, y \sim \dist}\left[A^i_{\mathcal{P}_{t - 1}, \mathcal{Q}_{t - 1}}(x, y)^2\right] +
        \E_{x, y \sim \dist}\left[\left(A^i_{\mathcal{P}_t, \mathcal{Q}_t}(x, y) - A^i_{\mathcal{P}_{t - 1}, \mathcal{Q}_{t - 1}}(x, y)\right)^2\right]. && \qedhere
    \end{align*}
\end{proof}

\subsubsection{Proof of the Energy Increment Claim (Claim~\ref{claim:key-claim-2})}
\label{sec:proof-key-claim-2}

\begin{proof}[Proof of \Cref{claim:key-claim-2}]
    By the acceptance condition at iteration $t$, the algorithm has found $P_t, Q_t$ such that
    \begin{equation*}
        \abs{\dist(P_t \times Q_t) \cdot \E_{x, y \sim \dist}\left[A^{i_t}(x, y) - A^{i_t}_{\mathcal{P}_{t - 1}, \mathcal{Q}_{t - 1}}(x, y) ~\middle|~ P_t \times Q_t\right]} \geq c\epsilon.
    \end{equation*}
    Because we refine the partition at step $t$ to include $P_t$ and $Q_t$, it follows that $P_t \times Q_t$ is $\sigma(\mathcal{P}_t) \otimes \sigma(\mathcal{Q}_t)$-measurable. Therefore,
    \begin{equation*}
        \E_{x, y \sim \dist}\left[A^{i_t}(x, y) - A^{i_t}_{\mathcal{P}_{t}, \mathcal{Q}_{t}}(x, y) ~\middle|~ P_t \times Q_t\right] = 0.
    \end{equation*}
    Via the triangle inequality and some algebraic manipulation, we get
    \begin{align*}
        c\epsilon &\leq \dist(P_t \times Q_t) \cdot \left(\abs{\E_{x, y \sim \dist}\left[A^{i_t}(x, y) - A^{i_t}_{\mathcal{P}_{t - 1}, \mathcal{Q}_{t - 1}}(x, y) ~\middle|~ P_t \times Q_t\right]}\right) \\
        &= \dist(P_t \times Q_t) \cdot \left(\abs{\E_{x, y \sim \dist}\left[A^{i_t}(x, y) - A^{i_t}_{\mathcal{P}_{t - 1} , \mathcal{Q}_{t - 1}}(x, y) ~\middle|~ P_t \times Q_t\right]} - 0\right) \\
        &= \dist(P_t \times Q_t) \cdot \bigg(\abs{\E_{x, y \sim \dist}\left[A^{i_t}(x, y) - A^{i_t}_{\mathcal{P}_{t - 1}, \mathcal{Q}_{t - 1}}(x, y) ~\middle|~ P_t \times Q_t\right]} \\
        &\qquad\qquad\qquad\qquad- \abs{\E_{x, y \sim \dist}\left[A^{i_t}(x, y) - A^{i_t}_{\mathcal{P}_{t}, \mathcal{Q}_{t}}(x, y) ~\middle|~ P_t \times Q_t\right]}\bigg) \\
        &\leq \dist(P_t \times Q_t) \cdot \Bigg|\E_{x, y \sim \dist}\left[A^{i_t}(x, y) - A^{i_t}_{\mathcal{P}_{t - 1}, \mathcal{Q}_{t - 1}}(x, y) ~\middle|~ P_t \times Q_t\right] \\
        &\qquad\qquad\qquad\qquad- \E_{x, y \sim \dist}\left[A^{i_t}(x, y) - A^{i_t}_{\mathcal{P}_{t}, \mathcal{Q}_{t}}(x, y) ~\middle|~ P_t \times Q_t\right]\Bigg| \\
        &= \dist(P_t \times Q_t) \cdot \abs{\E_{x, y \sim \dist}\left[A^{i_t}_{\mathcal{P}_{t}, \mathcal{Q}_{t}}(x, y) - A^{i_t}_{\mathcal{P}_{t - 1}, \mathcal{Q}_{t - 1}}(x, y) ~\middle|~ P_t \times Q_t\right]}.
    \end{align*}
    Squaring both sides,
    \begin{equation*}
        (c\epsilon)^2 \leq \dist(P_t \times Q_t)^2 \cdot \E_{x, y \sim \dist}\left[A^{i_t}_{\mathcal{P}_{t}, \mathcal{Q}_{t}}(x, y) - A^{i_t}_{\mathcal{P}_{t - 1}, \mathcal{Q}_{t - 1}}(x, y) ~\middle|~ P_t \times Q_t\right]^2.
    \end{equation*}
    Because the cut norm is larger than $c\epsilon > 0$, it must be that $\dist(P_t \times Q_t) > 0$. This means that it is safe to divide both sides by $\dist(P_t \times Q_t)$ to get
    \begin{equation*}
        (c\epsilon)^2/\dist(P_t \times Q_t) \leq \dist(P_t \times Q_t) \cdot \E_{x, y \sim \dist}\left[A^{i_t}_{\mathcal{P}_{t}, \mathcal{Q}_{t}}(x, y) - A^{i_t}_{\mathcal{P}_{t - 1}, \mathcal{Q}_{t - 1}}(x, y) ~\middle|~ P_t \times Q_t\right]^2.
    \end{equation*}
    By Jensen's,
    \begin{equation*}
         (c\epsilon)^2/\dist(P_t \times Q_t) \leq \dist(P_t \times Q_t) \cdot \E_{x, y \sim \dist}\left[\left(A^{i_t}_{\mathcal{P}_{t}, \mathcal{Q}_{t}}(x, y) - A^{i_t}_{\mathcal{P}_{t - 1}, \mathcal{Q}_{t - 1}}(x, y)\right)^2 ~\middle|~ P_t \times Q_t\right]
    \end{equation*}
    Because it is non-negative, $$(1 - \dist(P_t \times Q_t))\E_{x, y \sim \dist}\left[\left(A^{i_t}_{\mathcal{P}_t, \mathcal{Q}_t}(x, y) - A^{i_t}_{\mathcal{P}_{t - 1}, \mathcal{Q}_{t - 1}}(x, y)\right)^2 ~\middle|~ (P_t \times Q_t)^C\right] \geq 0.$$ By the tower rule,
    \begin{align*}
        (c\epsilon)^2 / \dist(P_t \times Q_t)
        &\leq \dist(P_t \times Q_t) \cdot \E_{x, y\sim \dist}\left[\left(A^{i_t}_{\mathcal{P}_{t}, \mathcal{Q}_{t}}(x, y) - A^{i_t}_{\mathcal{P}_{t - 1}, \mathcal{Q}_{t - 1}}(x, y)\right)^2 ~\middle|~ P_t \times Q_t\right] \\
        &\qquad + (1 - \dist(P_t \times Q_t))\E_{x, y \sim \dist}\left[\left(A^{i_t}_{\mathcal{P}_{t}, \mathcal{Q}_{t}}(x, y) - A^{i_t}_{\mathcal{P}_{t - 1}, \mathcal{Q}_{t - 1}}(x, y)\right)^2 ~\middle|~ (P_t \times Q_t)^C\right] \\
        &= \E_{x, y \sim \dist}\left[\left(A^{i_t}_{\mathcal{P}_t, \mathcal{Q}_t}(x, y) - A^{i_t}_{\mathcal{P}_{t - 1}, \mathcal{Q}_{t - 1}}(x, y)\right)^2\right].
    \end{align*}
    Because $\dist(P_t \times Q_t) \leq 1$, we finally get the desired inequality
    \begin{equation*}
        (c\epsilon)^2 \leq \E_{x, y \sim \dist}\left[\left(A^{i_t}_{\mathcal{P}_{t}, \mathcal{Q}_{t}}(x, y) - A^{i_t}_{\mathcal{P}_{t - 1}, \mathcal{Q}_{t - 1}}(x, y)\right)^2\right]. \qedhere
    \end{equation*}
\end{proof}

\subsection{Indistinguishable Coarsenings of Communication Games}
\label{sec:indist-games}

With this powerful tool in hand, we can take any communication game $G$ and compute, in polynomial time, collections of subsets of Alice's and Bob's observation spaces $\mathcal{P}, \mathcal{Q}$ such that the coarsening $G_{\mathcal{P}, \mathcal{Q}}$ is indistinguishable from $G$.

For convenience, we restate the main coarsening theorem before proving it.
\IndistCoarseningExists*

\begin{proof}[Proof of \Cref{thm:small-coarsening}]
    The idea of the proof is to coarsen along the sets guaranteed to us by \Cref{lem:grid-FK}. Let $U^a$ be the matrix whose rows are indexed by $\Omega_A$ and whose columns are indexed by $\Omega_B$, where $U^a(\omega_A, \omega_B) = r(a; \omega_A, \omega_B)$.

    By applying \Cref{lem:grid-FK} on the matrices $U^a$ for all $a \in \mathcal{A}$, there exist collections $\mathcal{P}$ of subsets of $\Omega_A$ and $\mathcal{Q}$ of subsets of $\Omega_B$ such that for all $a \in \mathcal{A}$
    \begin{equation*}
        \cutnorm{U^a - U^a_{\mathcal{P}, \mathcal{Q}}} \leq \epsilon.
    \end{equation*}
    By the definition of the cut norm, for any $a \in \mathcal{A}$, any $R \subseteq \Omega_A$, and any $C \subseteq \Omega_B$, it follows that
    \begin{align*}
    \epsilon &\geq \abs{\sum_{\substack{\omega_A \in R \\\omega_B \in C}} \dist(\omega_A, \omega_B) \left(U^a(\omega_A, \omega_B) - U^a_{\mathcal{P}, \mathcal{Q}}(\omega_A, \omega_B)\right)} \\
            &= \abs{\E_{\omega_A, \omega_B \sim \dist}\left[\mathds{1}(\omega_A \in R,~ \omega_B \in C)\cdot \left(U^a(\omega_A, \omega_B) - U^a_{\mathcal{P}, \mathcal{Q}}(\omega_A, \omega_B)\right)\right]}.
    \end{align*}
    By definition, $U^a_{\mathcal{P}, \mathcal{Q}}$ is the reward of taking action $a$ under observations $\omega_A, \omega_B$ in the coarsening of $G$. It follows that $G$ and $G_{\mathcal{P}, \mathcal{Q}}$ are $\epsilon$-indistinguishable.
\end{proof}

\section{Communicating with Indistinguishable Coarsenings}

We begin the design of our communication protocol. The protocol will work by computing an $\epsilon$-indistinguishable coarsening of the communication game and solving the communication problem in the coarsened game instead. In \Cref{sec:planning}, we prove that doing so is safe: the value of a short policy on $G$ and an $\epsilon$-indistinguishable coarsening cannot differ too much. In \Cref{sec:one-way}, we describe our protocol and combine the results in \Cref{sec:planning} to prove that it is effective.

The final guarantee is the following. The number of bits communicated depends on the intrinsic partition number $N_\alpha(G)$, independently of the size of the observation spaces.

\begin{restatable}[Tractable Effective Communication]{theorem}{Communication}
    \label{thm:comm-exists}
    There exists a communication protocol that when given $\alpha \in [-1, 1]$, a communication game $G$ with $N_\alpha(G) < \infty$, and $\delta > 0$: runs in $\bO{\mathrm{poly}(n, m, 1/\delta)}$ time, achieves a utility of at least $\alpha - \delta$ in expectation, and sends no more than
    \[
        \bO{\frac{N_\alpha(G) \cdot \min(m, N_\alpha(G))}{\delta^2}}
    \]
    bits.
\end{restatable}

\subsection{Transfer Between the Original and Coarsened Games}
\label{sec:planning}

The proof follows a bidirectional argument regarding the utility of optimal policies in the original and coarsened games.
In the forward direction, because $G_{\mathcal{P}, \mathcal{Q}}$ is a coarsening of $G$, the expected utility of any policy based on coarsened information transfers back to $G$ with no loss of utility at all.
In the backwards direction, because $G_{\mathcal{P}, \mathcal{Q}}$ is $\epsilon$-indistinguishable with respect to rectangle tests (equivalently $1$-bit policies), any policy $\pi$ achieves utilities in $G$ and $G_{\mathcal{P}, \mathcal{Q}}$ at most $\epsilon \cdot \min\left(N(\pi), \bO{\sqrt{m N(\pi)}}\right)$ apart.

\subsubsection{Exact Transfer of Utility From Coarsened Observations to the Original Game}

We first show that the utility of any policy based on the coarsened observations --- that is, any coarsening-measurable policy --- in the coarsened game exactly matches its value in the original game. This exact transfer of utility motivates our choice of using coarsenings, as discussed in Section~\ref{sec:1.2}.

\begin{restatable}[Exact Utility Transfer Coarsening $\rightarrow$ Original Game]{lemma}{CoarseningNoError}
    \label{lem:coarsening-has-no-error}
    Let $G$ be a communication game and let $\mathcal{P}$ and $\mathcal{Q}$ be collections of subsets of the observation spaces. Let $\pi: \Omega_A \times \Omega_B \to \acts$ be a policy that is $(\sigma(\mathcal{P}) \otimes \sigma(\mathcal{Q}))$-measurable. Then,
    \begin{equation*}
        \pi(G) = \pi(G_{\mathcal{P}, \mathcal{Q}}).
    \end{equation*}
\end{restatable}
\begin{proof}
    Let $r$ be the reward function of $G$ and $r_{\mathcal{P}, \mathcal{Q}}$ be the reward function of $G_{\mathcal{P}, \mathcal{Q}}$.

    For ease of notation, for any action $a \in \acts$, define $R^a: \Omega_A \times \Omega_B \to \mathbb{R}$ to be the function mapping $(\omega_A, \omega_B) \mapsto r(a; \omega_A, \omega_B)$.

    By definition,
    \begin{equation*}
        r_{\mathcal{P}, \mathcal{Q}}(a; \omega_A, \omega_B) = \E_{\dist}[R^a \mid \sigma\left(\mathcal{P}\right) \otimes \sigma\left(\mathcal{Q}\right)](\omega_A, \omega_B).
    \end{equation*}
    So the expected reward of running $\pi$ on game $G_{\mathcal{P}, \mathcal{Q}}$ is
    \begin{align*}
        \pi(G_{\mathcal{P}, \mathcal{Q}})
        &= \E_{\omega_A, \omega_B \sim \dist}\left[r_{\mathcal{P}, \mathcal{Q}}\left(\pi(\omega_A, \omega_B); \omega_A, \omega_B\right)\right] \\
        &= \E_{\omega_A, \omega_B \sim \dist}\left[\sum_{a \in \acts} \mathds{1}(\pi(\omega_A, \omega_B) = a) \cdot \E_{\dist}[R^a \mid \sigma\left(\mathcal{P}\right) \otimes \sigma\left(\mathcal{Q}\right)](\omega_A, \omega_B)\right]
    \end{align*}
    Because $\pi$ is $(\sigma(\mathcal{P}) \otimes \sigma(\mathcal{Q}))$-measurable, it follows that each indicator $\mathds{1}(\pi(\omega_A, \omega_B) = a)$ is $(\sigma(\mathcal{P}) \otimes \sigma(\mathcal{Q}))$-measurable. For ease of notation, let $\pi^a: \Omega_A \times \Omega_B \to \{0, 1\}$ be the indicator function mapping $(\omega_A, \omega_B) \mapsto \mathds{1}(\pi(\omega_A, \omega_B) = a)$.

    We can take the measurable function into the expectation to get
    \begin{align*}
        \pi(G_{\mathcal{P}, \mathcal{Q}}) &= \E_{\omega_A, \omega_B \sim \dist}\left[\E_{\dist}\left[\sum_{a \in \acts} R^a \cdot \pi^a \mid \sigma\left(\mathcal{P}\right) \otimes \sigma\left(\mathcal{Q}\right)\right](\omega_A, \omega_B)\right].
    \end{align*}
    By the tower rule,
    \begin{align*}
        \pi(G_{\mathcal{P}, \mathcal{Q}})
        &= \E_{\omega_A, \omega_B \sim \dist}\left[\sum_{a \in \acts} R^a \cdot \pi^a\right] \\
        &= \E_{\omega_A, \omega_B \sim \dist}\left[r(\pi(\omega_A, \omega_B); \omega_A, \omega_B)\right] \\
        &= \pi(G). \tag*{\qedhere}
    \end{align*}
\end{proof}

\subsubsection{Approximate Utility Transfer of
any Short Policy between the Coarsened and Original Games
}

Now we show that given two communication games $G$ and $\hat{G}$ that are $\epsilon$-indistinguishable, the expected utility of short communication policies cannot differ too much. The first bound loses linearly in the number of rectangles induced by the protocol. The second bound improves this dependence when the number of actions is smaller than the number of rectangles.
At their core, both results concern the fact that policies communicating $k$ bits of information are composed of simpler one-bit communications, which are the objects of the indistinguishability tests we designed using rectangle actions. Thus, both results can be understood as showing how indistinguishability with respect to one-bit policies can be boosted to indistinguishability with respect to $k$-bit policies.

\begin{restatable}[Weak $\epsilon$-indistinguishability Bound]{lemma}{IndistWeak}
    \label{lem:indist-weak}
    Let $G$ and $\hat{G}$ be two $\epsilon$-indistinguishable communication games. Let $D$ be an arbitrary rectangular partition of $\Omega_A \times \Omega_B$. When $\pi$ is a $D$-measurable policy,
    \begin{equation*}
        \abs{\pi(G) - \pi(\hat{G})} \leq \epsilon \cdot \abs{D}.
    \end{equation*}
    In particular, for an arbitrary communication policy $\pi$,
    \begin{align*}
        \abs{\pi(G) - \pi(\hat{G})} &\leq \epsilon \cdot N(\pi) \\
        &\leq \epsilon \cdot 2^{CC(\pi)}.
    \end{align*}
\end{restatable}

\begin{restatable}[Strong $\epsilon$-indistinguishability Bound]{lemma}{IndistStrong}
    \label{lem:indist-strong}
    Let $G$ and $\hat{G}$ be two $\epsilon$-indistinguishable communication games. Let $D$ be an arbitrary rectangular partition of $\Omega_A \times \Omega_B$. When $\pi$ is a $D$-measurable policy,
    \begin{equation*}
        \abs{\pi(G) - \pi(\hat{G})} \leq 4\epsilon K_{\mathrm{Gr}} \cdot \sqrt{m\abs{D}}.
    \end{equation*}
    Here, $K_{\mathrm{Gr}} \approx 1.7$ is Grothendieck's constant \citep{grothendieck1956resume}.
    In particular, for an arbitrary communication policy $\pi$,
    \begin{align*}
        \abs{\pi(G) - \pi(\hat{G})} &\leq 4\epsilon K_{\mathrm{Gr}} \cdot \sqrt{mN(\pi)} \\
        &\leq 4\epsilon K_{\mathrm{Gr}} \cdot \sqrt{m} \cdot 2^{CC(\pi)/2},
    \end{align*}
\end{restatable}

\subsubsection{Proof of the Weak Transfer Bound (Lemma~\ref{lem:indist-weak})}

\begin{proof}[Proof of \Cref{lem:indist-weak}]
    Let $r: \mathcal{A} \times \Omega_A \times \Omega_B \to [-1, 1]$ be the reward function for $G$, and $\hat{r}: \mathcal{A} \times \Omega_A \times \Omega_B \to [-1, 1]$ be the reward function for $\hat{G}$. Let the rectangular partition be $D = \{R_i \times C_i\}_{i = 1}^\ell$ and let $a_i$ be the constant action $\pi$ takes in rectangle $R_i \times C_i$.

    By $\epsilon$-indistinguishability, we know that on each rectangle $R_i \times C_i$ the average utility of taking any action $a \in \mathcal{A}$ in $G$, and therefore especially $a_i$, differs by at most $\epsilon$ from taking the action in $\hat{G}$. Formally,
    \begin{equation*}
        \abs{\sum_{\substack{\omega_A \in R_i\\\omega_B \in C_i}} \dist(\omega_A, \omega_B)\left(r(a_i; \omega_A, \omega_B) - \hat{r}(a_i; \omega_A, \omega_B)\right)} \leq \epsilon.
    \end{equation*}

    The utility of running policy $\pi$ on game $G$ is
    \begin{equation*}
        \pi(G) = \sum_{\substack{\omega_A \in \Omega_A\\\omega_B \in \Omega_B}} \dist(\omega_A, \omega_B) \cdot r(\pi(\omega_A, \omega_B); \omega_A, \omega_B) = \sum_{i = 1}^\ell \sum_{\substack{\omega_A \in R_i\\\omega_B \in C_i}} \dist(\omega_A, \omega_B) \cdot r(a_i; \omega_A, \omega_B).
    \end{equation*}
    Similarly, the utility of running policy $\pi$ on game $\hat{G}$ is
    \begin{equation*}
        \pi(\hat{G}) = \sum_{\substack{\omega_A \in \Omega_A\\\omega_B \in \Omega_B}} \dist(\omega_A, \omega_B) \cdot \hat{r}(\pi(\omega_A, \omega_B); \omega_A, \omega_B) = \sum_{i = 1}^\ell \sum_{\substack{\omega_A \in R_i\\\omega_B \in C_i}} \dist(\omega_A, \omega_B) \cdot \hat{r}(a_i; \omega_A, \omega_B).
    \end{equation*}
    So, by the triangle inequality, and by the bound on the absolute difference in each rectangle, the total difference in utility when running $\pi$ in two $\epsilon$-indistinguishable games is
    \begin{align*}
        \abs{\pi(G) - \pi(\hat{G})} &= \abs{\sum_{i = 1}^\ell \sum_{\substack{\omega_A \in R_i\\\omega_B \in C_i}} \dist(\omega_A, \omega_B) \cdot \left(r(a_i; \omega_A, \omega_B) - \hat{r}(a_i; \omega_A, \omega_B)\right)} \\
        &\leq \sum_{i = 1}^\ell \abs{\sum_{\substack{\omega_A \in R_i\\\omega_B \in C_i}} \dist(\omega_A, \omega_B) \cdot \left(r(a_i; \omega_A, \omega_B) - \hat{r}(a_i; \omega_A, \omega_B)\right)} \\
        &\leq \epsilon \abs{D},
    \end{align*}
    proving the first part of the claim.

    By \Cref{prop:communication-tiles}, there exists a rectangular partition $D = \{R_1 \times C_1, \dots, R_\ell \times C_\ell\}$ with $\ell = N(\pi)$, of the set $\Omega_A \times \Omega_B$ such that $\pi$ takes the constant action $a_i \in \mathcal{A}$ in each $R_i \times C_i$, for all $i \in [\ell]$. Applying the above to this rectangular partition, we get that $\abs{\pi(G) - \pi(\hat{G})} \leq \epsilon \cdot N(\pi) \le \epsilon \cdot 2^{CC(\pi)}$.
\end{proof}

\subsubsection{Proof of the Strong Transfer Bound (Lemma~\ref{lem:indist-strong})}

In the regime where $m$ is smaller than $N(\pi)$, we can show a much stronger result. When two games are $\epsilon$-indistinguishable, then they cannot differ by more than $\bO{\epsilon\sqrt{m N(\pi)}}$, as opposed to $\epsilon N(\pi)$. Proving this is more involved and will require the use of the notion of a factorization norm.

To show this result, we will first introduce the notion of an operator norm.

\begin{definition}[$\ell_a \to \ell_b$ norm]
    Given a matrix $A \in \mathbb{R}^{n_1 \times n_2}$, the $\ell_a \to \ell_b$ operator norm is defined as
    \begin{equation*}
        \norm{A}_{\ell_a \to \ell_b} = \max_{v \in \mathbb{R}^{n_2}:~ \norm{v}_{\ell_a} = 1} \norm{Av}_{\ell_b}.
    \end{equation*}
\end{definition}
    We also need to define the notion of a factorization norm.
    \begin{definition}[Factorization Norm]
        Given a matrix $A \in \mathbb{R}^{n_1 \times n_2}$, the factorization norm $\gamma$ is defined as
        \begin{equation*}
            \gamma(A) = \min_{X, Y :~ A = XY}~ \norm{X}_{\ell_2 \to \ell_\infty}\norm{Y}_{\ell_1 \to \ell_2}.
        \end{equation*}
        The dual factorization norm is defined as
        \begin{equation*}
            \gamma^*(A) = \sup_{B \neq 0} \frac{\abs{\inner{A, B}_F}}{\gamma(B)},
        \end{equation*}
        where $$\inner{A, B}_F = \sum_{\substack{i \in [n_1] \\ j \in [n_2]}} A(i, j) B(i, j)$$ is the Frobenius inner product of $A$ and $B$.
    \end{definition}
With the operator norm, we will show a handy lemma that will allow us to get a square root dependence on the partition number.

\begin{lemma}\citep{alon2004approximating, linial2007complexity}
    \label{lem:norm-relationships}
    For any matrices $A, B \in \mathbb{R}^{n_1 \times n_2}$,
    \begin{equation*}
        \abs{\inner{A, B}_F} \leq 4K_{\mathrm{Gr}} \norm{A}_\square\cdot \norm{B}_{\ell_1 \to \ell_\infty} \cdot \sqrt{ \mathrm{rank}(B)}.
    \end{equation*}
\end{lemma}
\begin{proof}
We start by showing $\gamma^*(A) \leq 4K_{\mathrm{Gr}} \norm{A}_\square$. By \cite[Equation 4]{linial2007lower},
    \begin{equation*}
        \gamma^*(A) \leq K_{\mathrm{Gr}} \norm{A}_{\ell_\infty \to \ell_1}.
    \end{equation*}

    By \cite[Lemma 3.1]{alon2004approximating},
    \begin{equation*}
        \norm{A}_{\ell_\infty \to \ell_1} \leq 4\norm{A}_\square.
    \end{equation*}
    Therefore,
    \begin{equation}
        \label{eq:norm-relationship-1}
        \gamma^*(A) \leq 4K_{\mathrm{Gr}} \norm{A}_\square.
    \end{equation}

    Next, we note that
    \begin{equation}
        \label{eq:norm-relationship-2}
        \gamma(B) \leq \norm{B}_{\ell_1 \to \ell_\infty}\sqrt{\mathrm{rank}(B)}
    \end{equation}
    is proven in \cite[Lemma 9]{linial2007complexity}.\footnote{We note that \cite[Lemma 9]{linial2007complexity} state their bound as $\gamma(B)^2 \leq \norm{B}_{\ell_1 \to \ell_\infty}\mathrm{rank}(B)$ which appears to be a typo missing the squared norm.}
    To prove the lemma, by duality, we know that
    \begin{equation*}
        \abs{\inner{A, B}_F} \leq \gamma^*(A) \cdot \gamma(B).
    \end{equation*}
    Upper bounding the factorization norm and its dual using \Cref{eq:norm-relationship-1} and \Cref{eq:norm-relationship-2} proves the lemma.
    \begin{equation*}
        \abs{\inner{A, B}_F} \leq 4K_{\mathrm{Gr}} \norm{A}_\square\cdot \norm{B}_{\ell_1 \to \ell_\infty}\sqrt{\mathrm{rank}(B)}. \tag*{\qedhere}
    \end{equation*}
\end{proof}

With this handy lemma, we are ready to prove \Cref{lem:indist-strong}.

\begin{proof}[Proof of \Cref{lem:indist-strong}]
    Let $r: \mathcal{A} \times \Omega_A \times \Omega_B \to [-1, 1]$ be the reward function for $G$, and $\hat{r}: \mathcal{A} \times \Omega_A \times \Omega_B \to [-1, 1]$ be the reward function for $\hat{G}$. Let $D=\{R_i\times C_i\}_{i=1}^\ell$ be a rectangular partition of $\Omega_A \times \Omega_B$.

    For any action $a \in \acts$, let $D^a$ be the set of rectangles on which $\pi$ plays action $a$. Let $M^a \in \{0, 1\}^{\abs{\Omega_A} \times \abs{\Omega_B}}$ be the matrix indicating whether action $a$ is played when observing $\omega_A, \omega_B$:
    \begin{equation*}
        M^a(\omega_A, \omega_B) = \mathds{1}(\pi(\omega_A, \omega_B) = a).
    \end{equation*}
    Let $U^a \in \mathbb{R}^{\abs{\Omega_A} \times \abs{\Omega_B}}$ be the matrix of the reward of action $a$ in communication game $G$, scaled by the probability of the observation:
    \begin{equation*}
        U^a(\omega_A, \omega_B) = \dist(\omega_A, \omega_B) \cdot r(a; \omega_A, \omega_B).
    \end{equation*}
    Define $\hat{U}^a \in \mathbb{R}^{\abs{\Omega_A} \times \abs{\Omega_B}}$ for game $\hat{G}$ analogously.
    Notice that the expected reward of playing policy $\pi$ on game $G$ can be written as
    \begin{align*}
        \pi(G) &= \sum_{\omega_A, \omega_B} \dist(\omega_A, \omega_B) r(\pi(\omega_A, \omega_B); \omega_A, \omega_B) \\
        &= \sum_{a \in \acts} \sum_{\omega_A, \omega_B} \dist(\omega_A, \omega_B) r(a; \omega_A, \omega_B) \cdot \mathds{1}(\pi(\omega_A, \omega_B) = a) \\
        &= \sum_{a \in \acts} \sum_{\omega_A, \omega_B} U^a(\omega_A, \omega_B) \cdot M^a(\omega_A, \omega_B) \\
        &= \sum_{a \in \acts} \inner{U^a, M^a}_F.
    \end{align*}
    By the same argument,
    \begin{equation*}
        \pi(\hat{G}) = \sum_{a \in \acts} \inner{\hat{U}^a, M^a}_F.
    \end{equation*}
    Therefore,
    \begin{equation}
        \label{eq:sum-over-actions}
        \abs{\pi(G) - \pi(\hat{G})} = \abs{\sum_{a \in \acts} \inner{U^a - \hat{U}^a, M^a}_F} \leq \sum_{a \in \acts} \abs{\inner{U^a - \hat{U}^a, M^a}_F}.
    \end{equation}

    By \Cref{lem:norm-relationships}, for each action $a$, we can upper bound the right-hand side by
    \begin{equation}
        \label{eq:norm-relationship}
        \abs{\inner{U^a - \hat{U}^a, M^a}_F} \leq 4K_{\mathrm{Gr}} \norm{U^a - \hat{U}^a}_\square \cdot \norm{M^a}_{\ell_1 \to \ell_\infty}\cdot \sqrt{ \mathrm{rank}(M^a)}.
    \end{equation}
    Notice that $M^a$ can be written as the sum of $\abs{D^a}$ rectangles
    \begin{equation*}
        M^a = \sum_{R_i \times C_i \in D^a} \mathds{1}_{R_i}\mathds{1}_{C_i}^T,
    \end{equation*}
    where $\mathds{1}_{R_i} \in \{0, 1\}^{n_A}$ and $\mathds{1}_{C_i} \in \{0, 1\}^{n_B}$ are the indicator vectors of sets $R_i$ and $C_i$. This equality holds because the rectangles are disjoint. Therefore, $\rank(M^a) \leq \abs{D^a}$.

    Moreover, notice that $\norm{M^a}_{\ell_1 \to \ell_\infty} \leq 1$. This is because for any vector $v \in \mathbb{R}^{n_B}$,
    \begin{equation*}
        \norm{M^a v}_\infty = \max_{i \in [n_A]} \abs{\sum_{j = 1}^{n_B} M^a(i, j) v(j)} \leq \max_{i \in [n_A]} \sum_{j = 1}^{n_B} M^a(i, j)\abs{v(j)} \leq \sum_{j = 1}^{n_B} \abs{v(j)} \leq \norm{v}_1,
    \end{equation*}
    where the first inequality follows by the triangle inequality, and the second follows because $M^a(i, j) \in \{0, 1\}$. Therefore,
    \begin{equation*}
        \norm{M^a}_{\ell_1 \to \ell_\infty} = \sup_{v \in \mathbb{R}^{n_B}:~ \norm{v}_1 = 1} \norm{M^av}_\infty \leq 1.
    \end{equation*}
    Substituting this back into \Cref{eq:norm-relationship}, we get
     \begin{equation*}
        \abs{\inner{U^a - \hat{U}^a, M^a}_F} \leq 4K_{\mathrm{Gr}} \norm{U^a - \hat{U}^a}_\square \cdot \sqrt{\abs{D^a}}.
    \end{equation*}
    By the definition of $\epsilon$-indistinguishability, we can show $\norm{U^a - \hat{U}^a}_\square \leq \epsilon$. To see this, let $S \times T$ be the rectangle that realizes the cut norm in $U^a - \hat{U}^a$. By definition,
    \begin{align*}
        \norm{U^a - \hat{U}^a}_\square
        &= \abs{\sum_{\substack{\omega_A \in S \\ \omega_B \in T}} (U^a(\omega_A, \omega_B) - \hat{U}^a(\omega_A, \omega_B))} \\
        &= \abs{\sum_{\substack{\omega_A \in S \\ \omega_B \in T}} \dist(\omega_A, \omega_B) \cdot \left(r(a; \omega_A, \omega_B) - \hat{r}(a;\omega_A, \omega_B)\right)} \\
        &\leq \epsilon.
    \end{align*}
    All in all, we can say that
    \begin{equation}
        \label{eq:last-part}
        \abs{\inner{U^a - \hat{U}^a, M^a}_F} \leq 4\epsilon K_{\mathrm{Gr}} \sqrt{\abs{D^a}}.
    \end{equation}
    Combining everything,
    \begin{align*}
        \abs{\pi(G) - \pi(\hat{G})}
        &\leq \sum_{a \in \acts} \abs{\inner{U^a - \hat{U}^a, M^a}_F} && (\text{By \Cref{eq:sum-over-actions}}) \\
        &\leq \sum_{a \in \acts} 4\epsilon K_{\mathrm{Gr}} \sqrt{\abs{D^a}} && (\text{By \Cref{eq:last-part}}) \\
        &= 4\epsilon K_{\mathrm{Gr}}  \sum_{a \in \acts} \sqrt{\abs{D^a}} \\
        &\leq 4\epsilon K_{\mathrm{Gr}}\sqrt{m\sum_{a \in \acts}\abs{D^a}} && (\text{By Cauchy--Schwarz}) \\
        &= 4\epsilon K_{\mathrm{Gr}} \sqrt{m \abs{D}},
    \end{align*}
    as claimed!

    For communication policy $\pi$, let $D = \{R_1 \times C_1, \dots, R_k \times C_k\}$ be the rectangular partition associated with policy $\pi$, so that $\pi$ takes a constant action in each $R_i \times C_i$ rectangle. By \Cref{prop:communication-tiles}, $k = N(\pi) \leq 2^{CC(\pi)}$. Applying the above, we get that $\abs{\pi(G) - \pi(\hat{G})} \leq 4 \epsilon K_{\mathrm{Gr}} \cdot \sqrt{m N(\pi)} \leq 4\epsilon K_{\mathrm{Gr}} \sqrt{m} \cdot 2^{CC(\pi)/2}$.
\end{proof}

\subsection{A Provably Effective Communication Protocol}
\label{sec:one-way}

We can now specify an algorithm for communication that is computationally efficient, achieves constant utility, and sends a small number of bits when the communication complexity of the game is small.

We first compute $\mathcal{P} \times \mathcal{Q}$ for the communication game $G$ that makes the coarsened game $G_{\mathcal{P}, \mathcal{Q}}$ $\epsilon$-indistinguishable from the original game $G$. Then we have Alice send which subsets in $\mathcal{P}$ her observation lies in using $\bO{\epsilon^{-2}}$ bits, and Bob send which subsets of $\mathcal{Q}$ his observation lies in using $\bO{\epsilon^{-2}}$ bits. For any $\delta > 0$, when $\epsilon$ is appropriately set to $\bO{\delta/\min\left(N_\alpha(G), \sqrt{mN_\alpha(G)}\right)}$, this achieves a utility of $\alpha - \delta$ in $\bO{N_\alpha(G) \min(m, N_\alpha(G))/\delta^2}$ bits.

\begin{algorithm}[h]
\caption{Effective Communication}
\label{alg:One-Way-Communication}
\begin{algorithmic}[1]
\Require A fidelity parameter $\epsilon$, and a communication game $G$.
\State Use \Cref{thm:small-coarsening} to compute a collection $\mathcal{P} \times \mathcal{Q}$ of subsets of $\Omega_A \times \Omega_B$ such that $G_{\mathcal{P}, \mathcal{Q}}$ is $\epsilon$-indistinguishable from $G$. Write $\mathcal{P} = \{P_1, \dots, P_k\}$ and $\mathcal{Q} = \{Q_1, \dots, Q_k\}$ where $k \in \bO{\epsilon^{-2}}$.
\State Alice sends the binary string $b^A \in \{0, 1\}^k$ such that for all $i \in [k]$, $b^A_i = \mathds{1}(\omega_A \in P_i)$.
\State Bob sends the binary string $b^B \in \{0, 1\}^k$ such that for all $i \in [k]$, $b^B_i = \mathds{1}(\omega_B \in Q_i)$.
\State Charlie plays an action $a^*(b^A,b^B)$ maximizing expected utility given all available information
\begin{equation*}
    a^*(b^A,b^B) \in \argmax_{a \in \acts} \E_{\omega_A, \omega_B \sim \dist}\left[r(a; \omega_A, \omega_B) \mid b^A, b^B\right].
\end{equation*}
\end{algorithmic}
\end{algorithm}

At a high level, $\epsilon$-indistinguishability allows us to take the $k$-bit communication policy achieving the optimal utility of $\alpha^*$ in $G$ and use it to say that in $G_{\mathcal{P}, \mathcal{Q}}$, the optimal policy must achieve utility at least $\alpha^* - \epsilon \cdot 2^k$. By construction, getting perfect utility in $G_{\mathcal{P}, \mathcal{Q}}$ will be easy, only requiring $\bO{\epsilon^{-2}}$ bits. Finally, we will use the calibration condition afforded to us by the fact that we are planning in a coarsened game to say that this optimal policy in $G_{\mathcal{P}, \mathcal{Q}}$ will achieve exactly the same utility in $G$, without having to suffer the very large loss of $\epsilon \cdot 2^{\epsilon^{-2}}$ utility that $\epsilon$-indistinguishability alone would otherwise suggest is necessary.

\begin{theorem}[Tractable Efficient Communication]
    \label{thm:one-way}
    For any $\epsilon > 0$ and $\alpha \in [-1, 1]$, running \Cref{alg:One-Way-Communication} takes time $\mathrm{poly}(\epsilon^{-1}, n, m)$, sends $\bO{\epsilon^{-2}}$ bits, and achieves a utility of at least
    \begin{equation*}
        \alpha - \epsilon \cdot \min\left(N_\alpha(G), 4K_{\mathrm{Gr}}\sqrt{mN_\alpha(G)}\right).
    \end{equation*}
\end{theorem}
\begin{proof}
    By \Cref{thm:small-coarsening}, the protocol can be run in $\textrm{poly}(\epsilon^{-1}, n, m)$ time. Because $\mathcal{P}$ and $\mathcal{Q}$ contain no more than $\bO{\nicefrac{1}{\epsilon^2}}$ sets, it only takes $\bO{\nicefrac{1}{\epsilon^2}}$ bits to communicate whether or not the observation lies in each part. It remains to show that the policy achieves high utility on average.

    Let $\mathcal{P} = \{P_1, \dots, P_\ell\}, \mathcal{Q} = \{Q_1, \dots, Q_\ell\}$ be the subsets of $\Omega_A$ and $\Omega_B$ produced by running \Cref{alg:One-Way-Communication}. By definition, there exists a rectangular partition $D$ of $\Omega_A \times \Omega_B$ with $\abs{D} \leq N_\alpha(G)$ such that a $D$-measurable policy $\pi^*$ achieves utility at least $\alpha$ in $G$. By \Cref{lem:indist-weak}, $\pi^*$ achieves a utility of at least $\alpha - \epsilon N_\alpha(G)$ in $G_{\mathcal{P}, \mathcal{Q}}$ and by \Cref{lem:indist-strong}, a utility of at least $\alpha - 4\epsilon K_{\mathrm{Gr}} \sqrt{mN_\alpha(G)}$ in $G_{\mathcal{P}, \mathcal{Q}}$.

    Let $\pi$ be the policy in \Cref{alg:One-Way-Communication} that has Alice and Bob send which subsets of $\mathcal{P}$ and $\mathcal{Q}$ contain their observations, then play the action that achieves the highest utility on average conditioned on this information. To continue, we will show that $\pi$ is perfect in $G_{\mathcal{P}, \mathcal{Q}}$, and therefore achieves utility at least as high as $\pi^*$ in the coarsened game.

    \begin{claim}
        \label{claim:claim-1-protocol}
        $\pi$ achieves optimal utility in the coarsened game $G_{\mathcal{P}, \mathcal{Q}}$. In particular, compared to $\pi^*$,
        $$\pi(G_{\mathcal{P}, \mathcal{Q}}) \geq \pi^*(G_{\mathcal{P}, \mathcal{Q}}).$$
    \end{claim}
    \begin{proof}
        Run $\pi$ on game $G_{\mathcal{P}, \mathcal{Q}}$, and let $\omega_A, \omega_B$ be the observations given to Alice and Bob at the start of the game.
        Because Charlie is taking the action that maximizes utility conditioned on which subsets in $\mathcal{P}, \mathcal{Q}$ that $\omega_A, \omega_B$ are contained in, it follows that the utility that Charlie achieves is
        \begin{equation*}
            \max_{a \in \acts} \E[r(a; \cdot, \cdot) \mid \sigma(\mathcal{P}) \otimes \sigma(\mathcal{Q})](\omega_A, \omega_B).
        \end{equation*}
        By the definition of a coarsening, it follows that the expectation is the reward function of the coarsened game
        \begin{equation*}
            \max_{a \in \acts} r_{\mathcal{P}, \mathcal{Q}}(a; \omega_A, \omega_B).
        \end{equation*}
        Thus, Charlie is taking the optimal action in the coarsened game after the communication policy $\pi$ terminates. Because it is optimal, $\pi$ must perform at least as well as $\pi^*$ in $G_{\mathcal{P}, \mathcal{Q}}$. That is,
        \begin{equation*}
            \pi(G_{\mathcal{P}, \mathcal{Q}}) \geq \pi^*(G_{\mathcal{P}, \mathcal{Q}}). \qedhere
        \end{equation*}
    \end{proof}
    For the second part of the proof, we will show that the utility of $\pi$ in $G_{\mathcal{P}, \mathcal{Q}}$ matches the utility of $\pi$ in $G$ \textit{exactly}. Because the action played by policy $\pi$ on observation pair $\omega_A, \omega_B$ only depends on the parts of $\mathcal{P}$ and $\mathcal{Q}$ that $\omega_A$ and $\omega_B$ lie in, it follows that $\pi$ is $(\sigma(\mathcal{P}) \otimes \sigma(\mathcal{Q}))$-measurable. Therefore, by \Cref{lem:coarsening-has-no-error}, $\pi(G_{\mathcal{P}, \mathcal{Q}}) = \pi(G)$.

    All in all, \Cref{alg:One-Way-Communication} outputs, in polynomial time, a communication policy $\pi$ that sends $\bO{\nicefrac{1}{\epsilon^2}}$ bits and achieves utility
    \begin{align*}
        \pi(G) &= \pi(G_{\mathcal{P}, \mathcal{Q}}) && \text{(By \Cref{lem:coarsening-has-no-error})} \\
             &\geq \pi^*(G_{\mathcal{P}, \mathcal{Q}}) && \text{(By \Cref{claim:claim-1-protocol})} \\
             &\geq \pi^*(G) - \epsilon \cdot \min\left(N_\alpha(G), 4 K_{\mathrm{Gr}} \sqrt{m N_\alpha(G)}\right). && \text{(By \Cref{lem:indist-weak} and \Cref{lem:indist-strong})}
    \end{align*}
    This completes the proof.
\end{proof}

From here, to complete the proof of the protocol, we only need to instantiate \Cref{thm:one-way} with the correct choice of $\epsilon$. To set $\epsilon$ appropriately in one shot, we need to know $N_\alpha(G)$. An algorithm for computing this exactly would give an algorithm for approximating the communication complexity. Computing the communication complexity exactly is known to be hard under cryptographic assumptions \citep{kushilevitz_complexity_2009}, and recently has been discovered to be NP-hard \citep{gaspers2025np, hirahara2025communication}. Nevertheless, we can use a doubling trick and achieve the same utility without this privileged information, suffering no more than a log factor in $N_\alpha(G)$ in running time, and a constant multiplicative factor in the number of bits sent.

We are finally ready to prove \Cref{thm:comm-exists}.
\begin{proof}[Proof of \Cref{thm:comm-exists}]
    When $\epsilon$ is set to a value at most $$\frac{\delta}{\min\left(4K_{\mathrm{Gr}} \sqrt{N_\alpha(G)m},~ N_\alpha(G)\right)},$$ the utility achieved by \Cref{thm:one-way} is at least
    \begin{equation*}
        \alpha -\frac{\delta}{\min(4K_{\mathrm{Gr}} \sqrt{N_\alpha(G)m},~ N_\alpha(G))}\cdot \min\left(4K_{\mathrm{Gr}} \sqrt{N_\alpha(G)m},~ N_\alpha(G)\right) = \alpha - \delta.
    \end{equation*}
    This uses $$\bO{\frac{1}{\epsilon^2}} = \bO{\frac{N_\alpha(G) \cdot \min(m, N_\alpha(G))}{\delta^2}}$$ bits.

    Computing $N_\alpha(G)$ is hard, so instead, $\epsilon$ can initially be set to $\nicefrac{1}{2}$, and the protocol is simulated repeatedly with $\epsilon$ halving each time. The average utility of the protocol can always be computed in $\bO{n^2m}$ time by simply running the protocol for each observation pair. We can keep halving $\epsilon$ until the average utility is at least $\alpha - \delta$. Once a small enough choice of $\epsilon$ is found so that in expectation, the protocol results in $\alpha - \delta$ utility, the agents send bits by running the protocol with that choice of $\epsilon$. In the worst case, this uses an $\epsilon$ that is half as small as it needs to be, meaning a constant-factor blow up in the bit complexity.
    This only requires $$\log\left(\frac{\min\left(4K_{\mathrm{Gr}} \sqrt{N_\alpha(G)m},~ N_\alpha(G)\right)}{\delta}\right) \in \bO{\log(nm/\delta)}$$ repetitions of the protocol to be simulated, so the choice of $\epsilon$ can be computed in polynomial time.
\end{proof}

%% file: sections/05-Hardness-of-Communication.tex
\section{The Computational Intractability of Perfect Communication}
\label{sec:hardness}

We now provide evidence that the communication problem in its most general form is intractable. Hardness results for the problem of computing the numerical value of the communication complexity have been given in the literature \citep{kushilevitz_complexity_2009, ilango2020np, hirahara2025communication, gaspers2025np}. A particularly strong hardness result available is the following.

\begin{proposition}\citep[Corollary 34]{ilango2020np}
    Unless $\mathrm{P} = \mathrm{NP}$, for any $\epsilon \in (0, 1)$, there is no algorithm that can, in $\mathrm{poly}(n, m)$ time, output a communication policy for a communication game $G$ achieving perfect utility in less than $(1 - \epsilon)\log(n) + CC(G)$ bits.
\end{proposition}

Using the same construction as \cite[Theorem 32]{ilango2020np}, we can strengthen the argument to show that the required number of bits communicated must blow up exponentially on games with constant communication complexity.

\begin{theorem}[Fixed Parameter Lower Bound of Perfect Play in Communication Games]
    \label{thm:perfect-hardness}
    Unless $\mathrm{P} = \mathrm{NP}$, for any $\epsilon \in (0, \nicefrac{1}{2})$ and $k \geq 3$, every algorithm that runs in $\mathrm{poly}(n, m)$ time and outputs on communication games of communication complexity $CC(G) \leq k$ a communication policy achieving perfect utility must output, in the worst case, policies of length at least
    \begin{equation*}
        \frac{1}{2}\cdot\left(\frac{1}{2} - \epsilon\right) \cdot 2^{CC(G)}
    \end{equation*}
    bits, and at least
    \begin{equation*}
        \frac{1}{4}\cdot\left(\frac{1}{2} - \epsilon\right)\cdot N(G)
    \end{equation*}
    bits.
\end{theorem}

Any algorithm that finds optimal policies in games $G$ where $CC(G) \leq k$ for every $k$ must either take superpolynomial time or output communication schemes of length $2^{\Omega(CC(G))}$, matching, up to a constant factor, the blow up in the communication cost of \Cref{thm:one-way}.

To prove this theorem, we will reduce communication to the problem of graph coloring. For any graph $H$, we can construct an associated communication game $G_H$ as follows.

\begin{definition}[Graph Communication Games]
    \label{def:induced-graph}
    Given a simple undirected graph $H = (V, E)$, define the induced communication game $G_H$ to be the communication game with observations $\Omega_A = V$ for Alice and $\Omega_B = V$ for Bob, with actions $\mathcal{A} = \{\mathrm{EQ}, \mathrm{NEQ}\}$. The reward function $r: \mathcal{A} \times V \times V \to \{0, 1\}$ is defined as
    \begin{equation*}
        r(a; u, v) = \begin{cases}
            1 & \text{If $u = v$ and $a = \textrm{EQ}$}, \\
            1 & \text{If $u \neq v$ and $a = \textrm{NEQ}$}, \\
            0 & \text{Otherwise}.
        \end{cases}
    \end{equation*}
    The distribution over observations $\dist$ is defined so that observation pairs are drawn from the set
    \begin{equation*}
        \{(u, u) : u \in V\} \cup \{(u, v), (v, u) : \{u, v\} \in E\}
    \end{equation*}
    uniformly at random.
\end{definition}

As in \cite[Corollary 2]{ilango2020np}, we can prove that there is a tight correspondence between colorings and communication policies.

\begin{lemma}[Short Communication Policies Induce Graph Colorings]
    \label{lem:graph-reduction}
    For any graph $H = (V, E)$, let $G$ be the induced communication game. If for some $q \in \mathbb{N}$, the graph $H$ has a $q$-coloring, then $N(G) \leq 4q$ and $CC(G) \leq 1 + \ceil{\log_2(q)}$.

    Conversely, given a rectangular partition $D$ of the communication game $G$ into $q$ parts that achieves utility 1, there exists a $\mathrm{poly}(n)$ time algorithm that outputs a $q$-coloring of graph $H$.
    
    In particular, this implies that any communication policy for game $G$ that uses $\ell$ bits and achieves a utility of 1 induces a $2^\ell$ coloring of graph $H$.
\end{lemma}
\begin{proof}
    We start with the easy direction. Let $c: V \to [q]$ be a $q$-coloring of graph $H$. Consider the communication policy that operates as follows when Alice receives vertex $u$ and Bob receives vertex $v$. First, Alice sends $m_1 = c(u)$. Next, Bob sends $m_2 = \mathds{1}(m_1 = c(v))$. Finally, Charlie plays $\mathrm{EQ}$ when $m_2 = 1$ and $\mathrm{NEQ}$ otherwise.

    If $m_2 = 0$, then the agents cannot have received the same vertex, as $c(v) \neq c(u)$ implies $v \neq u$. If $m_2 = 1$, it must be that $c(u) = c(v)$. Because $c$ is a valid graph coloring, it cannot be that $\{u, v\} \in E$. Because $\{u, v\}$ is only in the support of the observation distribution when $u = v$ or $\{u, v\} \in E$, it follows that $u = v$. Thus, we have shown the correctness of the protocol.
    
    This protocol uses $1 + \ceil{\log_2(q)}$ bits, and therefore $CC(G) \leq 1 + \ceil{\log_2(q)}$ and, by \Cref{prop:communication-tiles}, $N(G) \leq 2^{1 + \ceil{\log_2(q)}} \leq 2^{2 + {\log_2(q)}} \leq 4q$.

    Now we show the other direction. Let $D = \{R_1 \times C_1, \dots, R_q \times C_q\}$ be a rectangular partition of the observation space $V \times V$. Let $\delta_D$ be the corresponding indicator function that maps $(u, v)$ to the index $i$ such that $(u, v) \in R_i \times C_i$. Consider an optimal $D$-measurable policy $\pi: V \times V \to \{\mathrm{EQ}, \mathrm{NEQ}\}$---that is, a policy of the form $\pi(u, v) = \tilde{\pi}(\delta_D(u, v))$ for a function $\tilde{\pi}: [q] \to \{\mathrm{EQ}, \mathrm{NEQ}\}$. Because $D$ achieves utility 1,
    \begin{equation*}
        \E_{u, v \sim \dist}[r(\pi(u, v); u, v)] = 1.
    \end{equation*}
    Consider the induced $q$-coloring function $c: V \to [q]$, defined so that the color of a vertex $u \in V$ is the index of the rectangle in $D$ that the pair $(u, u)$ is in, i.e., $\delta_D(u, u)$. It remains to show that this is a valid coloring.

    Notice that for vertices $u, v \in V$, if $\delta_D(u, u) = \delta_D(v, v)$, then it must be the case that $\delta_D(u, v) = \delta_D(u, u) = \delta_D(v, v)$. This is the standard fooling set argument, which follows because $u \in R_{\delta_D(u, u)}$ and $v \in C_{\delta_D(v, v)}$, so $(u, v) \in R_{\delta_D(u, u)} \times C_{\delta_D(v, v)} = R_{\delta_D(u, u)} \times C_{\delta_D(u, u)}$ and therefore $\delta_D(u, v) = \delta_D(u, u)$.

    Because they are in the same rectangle in $D$, it must be that $\pi$ plays the same action for all of $(u, u)$, $(v, v)$, and $(u, v)$. But for $(u, u)$ and $(v, v)$ the only action with utility 1 is to play $\mathrm{EQ}$, and for $(u, v)$ the action with utility 1 is to play the different $\mathrm{NEQ}$. Therefore, any tiling that places all three observation pairs in the same rectangle must get at least one pair wrong. By the contrapositive, for any edge $\{u, v\} \in E$, if all of $\pi(u, u) = \mathrm{EQ}, \pi(v, v) = \mathrm{EQ}, \pi(u, v) = \mathrm{NEQ}$ are true, then $\delta_D(u, u) \neq \delta_D(v, v)$, and so the color given to the two vertices is different, respecting the edge constraint.
\end{proof}

With this reduction in hand, we can use known results on the hardness of graph coloring even when exponentially more colors than the chromatic number are allowed to show that any efficient communication protocol must use exponentially more bits than the communication complexity.

\begin{proof}[Proof of \Cref{thm:perfect-hardness}]
    By \cite{wrochna2020improved}, it is NP-hard to color a graph $H$ with $\binom{q}{\floor{q/2}} - 1$ colors in polynomial time when $H$ is colorable with $q \geq 4$ colors. Let $k \geq 3$ be arbitrary, let $q = 2^{k - 1}$ and consider the communication game $G$ induced by a $q$-colorable graph $H$ (\Cref{def:induced-graph}). 
    By \Cref{lem:graph-reduction}, we can see that $CC(G) \leq k$.

    Suppose that there exists an algorithm that outputs, in polynomial time, a communication policy with, for some fixed $\epsilon > 0$, at most
    \begin{equation*}
        \left(\frac{1}{2} - \epsilon\right)\cdot 2^{CC(G) - 1} \leq \left(\frac{1}{2} - \epsilon\right) \cdot 2^{k - 1}
    \end{equation*}
    bits. By \Cref{lem:graph-reduction}, this induces a perfect coloring of graph $H$ with
    \begin{equation*}
        2^{\left(\frac{1}{2} - \epsilon\right) 2^{k - 1}}
    \end{equation*}
    colors. Note that the logarithm of the minimum number of colors an efficient algorithm needs in this graph is
    \begin{equation*}
        \log_2\binom{q}{\floor{q/2}} \geq \log_2\left(\left(\frac{q}{\floor{q/2}}\right)^{\floor{q/2}}\right) \geq \log_2(2^{\floor{q/2}}) = \floor{q/2} = 2^{k -2}.
    \end{equation*}
    But the logarithm of the number of colors in the coloring we output is
    \begin{equation*}
        \left(\frac{1}{2} - \epsilon\right) 2^{k - 1},
    \end{equation*}
    which for any fixed $\epsilon > 0$, must be smaller than $2^{k - 2}$, i.e., smaller than the logarithm of the minimum number of colors needed for any polynomial time algorithm. Contradiction!
    
    We can run the same argument for the tiling number. Indeed, $N(G) \leq 4q$. Suppose a perfect communication policy with $\frac{1}{4}\left(\frac{1}{2} - \epsilon\right)N(G)$ bits is found. This means a perfect policy with $\left(\frac{1}{2} - \epsilon\right)\cdot q$ bits is found, which induces by \Cref{lem:graph-reduction} a perfect coloring with $2^{\left(\frac{1}{2} - \epsilon\right)\cdot q}$ colors. By the same argument above, this is impossible unless $\mathrm{P} = \mathrm{NP}$.
\end{proof}

As a remark, the hardness result here strongly relies on the communication game using a non-product distribution over observations. We suspect that proving a hardness result for communication games with product distributions and only one correct action per observation pair is extremely difficult, and there is a possibility that communication is possible without an exponential blowup in complexity in this case. This is because a proof that estimating the minimum number of bits needed for perfect play is computationally hard in this case would immediately falsify the famous log-rank conjecture \citep{ilango2020np, gaspers2025np}.\footnote{The log-rank conjecture, if true, would bound the communication complexity of a function $f: \mathcal{X} \times \mathcal{Y} \to \mathbb{Z}_2$ polynomially by the log of the rank of the matrix $M_f \in \mathbb{Z}_2^{\abs{\mathcal{X}} \times \abs{\mathcal{Y}}}$ with $M_f(i, j) = f(i, j)$ \citep{lovasz1988lattices}. The rank of $M_f$ can easily be computed in polynomial time, and so this would define a procedure for efficiently approximating the communication complexity. If the communication complexity of a function is hard to approximate, then the log-rank conjecture must therefore be false.}

%% file: sections/06-Aumann-Agreement-Guarantees.tex
\section{Aumann Agreement in Low Complexity Games}

In this section, we will study the Aumann agreement protocol in the low complexity regime. The section will proceed in three parts. First, we adapt Aumann agreement to this discrete communication game, having the agents come to agreement over their belief of the optimal action. Second, we will tightly characterize the performance of agreement in terms of the communication complexity: we will show the number of bits communicated can be arbitrarily large on low complexity games, and we will show that the utility of agreement decays exponentially in the communication complexity. Finally, we will observe that weak-learning and rectangular substitutes assumptions in prior work~\cite{frongillo2023agreement,kong_false_2022,collina2025collaborative}
 imply that the communication game has communication complexity at most a constant, meaning that the protocol introduced in this paper works under strictly weaker assumptions than prior work.

Because it will be easy to do so, in this section we will state the results for the more general $N \geq 2$ agent communication game. In an $N$-agent communication game, there is a collection of observation sets $\Omega_1, \dots, \Omega_N$, a joint distribution over observations $\dist \in \Delta(\Omega_1 \times \cdots \times \Omega_N)$, and a reward function $r: \acts \times \Omega_1 \times \cdots \times \Omega_N \to [-1, 1]$. An action taker, Charlie, takes an action at the end of communication.

Moreover, Aumann agreement, instead of achieving additive guarantees on utility, will achieve multiplicative ones. So in this section all results will be stated for communication games with non-negative reward functions. Note that these results can still apply for the more general communication games studied above after a rescaling of the reward function.

We study a natural instantiation of the Aumann agreement protocol for decision problems, stated in \Cref{alg:AA-communication}, where instead of communicating posteriors, the agents communicate what they expect the best action to be. Agents break ties when there are multiple best actions using a public, fixed tie-breaking rule.

\begin{algorithm}[ht]
\caption{The Aumann Agreement Protocol For Two Agents}
\label{alg:AA-communication}
\begin{algorithmic}[1]
\Require A number of iterations $T$, and a communication game $G_1$.
\For{$T$ iterations}
    \State Let $a^*_A$ be an optimal action from Alice's perspective
    \begin{equation*}
        a^*_A \in \argmax_{a \in \acts} \E_{\omega_B \sim \dist_{t}}[r(a; \omega_A, \omega_B) \mid \omega_A].
    \end{equation*}
    \State Let $a^*_B$ be an optimal action from Bob's perspective
    \begin{equation*}
         a^*_B \in \argmax_{a \in \acts} \E_{\omega_A \sim \dist_{t}}[r(a; \omega_A, \omega_B) \mid \omega_B].
    \end{equation*}
    \State Alice communicates $a^*_A$ and Bob communicates $a^*_B$. This takes $\log(m)$ bits each.
    \State Let $G_{t + 1}$ be the new subgame after communication in $G_t$. Let $\mu_{t + 1}$ be the distribution over observations in $G_{t + 1}$.
\EndFor
\end{algorithmic}
\end{algorithm}

\subsection{Defining Agreement}

The standard way of defining agreement between two agents --- as is done in \cite{aaronson_complexity_2005, frongillo2023agreement, collina_tractable_2024, collina2025collaborative} ---  is to say that on Line 4 of \Cref{alg:AA-communication}, the messages that Alice and Bob send are exactly the same.

But agreement is not always lasting in general; agents can fall in and out of agreement, and the first time they enter agreement may not be the last. For a fairer comparison, we instead study the guarantees on the agents' actions after they have reached agreement for the final time, what we call lasting agreement. Agreement is lasting when running the Aumann agreement protocol indefinitely can never change what the agents report as their expected best action. That is, when it is the case that for all agents, reporting their expected best action reveals no more information because at this point in the conversation their expected best action is the same regardless of what they could be observing.

\begin{definition}[$\epsilon$-Lasting Agreement]
    In a communication game $G$, the agents are said to be in $\epsilon$-lasting agreement that some action $a \in \acts$ is optimal when for all agents $i \in [N]$ and all observations $\omega_i \in \Omega_i$ with $\Pr_\dist(\omega_i) > 0$:
    \begin{equation*}
        \E_{\omega_{-i}}[r(a; \omega_{-i}, \omega_i) \mid \omega_i] \geq \max_{b \in \acts}\E_{\omega_{-i}}[r(b; \omega_{-i}, \omega_i) \mid \omega_i] - \epsilon.
    \end{equation*}
    When $\epsilon = 0$, we will say the agents have reached lasting agreement.
\end{definition}

Because more information never hurts, any upper bound we prove about the utility of the agents' actions at lasting agreement must also upper bound the utility at first agreement.

\subsection{Quantifying the Quality of Decisions Made Under Agreement}

Notice that, eventually, running \Cref{alg:AA-communication} for long enough must make the agents come to agreement. This was first shown by \cite{geanakoplos_we_1982}. However, unlike what intuition would suggest, agents having reached agreement over what action they think is best is no guarantee that this action will be any good.

\begin{proposition}[Agreement Does Not Imply Accuracy]
    \label{prop:AA-upper-bound}
    For any choice of $k, N \in \mathbb{N}$, for any $\delta \in (0, 1)$, and any choice of $\alpha \in (0, 1]$, there exists a communication game $G$ with $N$ agents and $N_{\alpha(1 - \delta)}(G) \in \Theta(k)$ where perfect play results in utility at least $1 - \delta$, such that running the Aumann agreement protocol for any number of steps results in an expected utility of no more than
    \begin{equation*}
        \left(\frac{1}{N_{\alpha(1 - \delta)}(G)}\right)^{\frac{N - 1}{N}}.
    \end{equation*}
\end{proposition}
\begin{proof}
    Let $\ell = \ceil{k^{\frac{1}{N}}}$. The $N$ agents are each given a number uniformly at random in $\mathbb{Z}_\ell$, with agent $i$ receiving number $\omega_i$. The set of actions is $\mathcal{A} = \mathbb{Z}_\ell^{N - 1}$. The goal is for the agents to play
    \begin{equation*}
        (\omega_1 + \omega_2,~ \omega_2 + \omega_3,~ \dots,~ \omega_{N - 1} + \omega_N),
    \end{equation*}
    a list of $N - 1$ partial sums modulo $\ell$. They receive a reward of $1 - \delta$ when they succeed and a reward of 0 otherwise. To break ties, when the answer is $(0, 0, \dots, 0)$, they receive a reward of 1 when they succeed.

    Notice that because there are only $N - 1$ sums and $N$ variables, every $(N-1)$-tuple in $\mathbb{Z}_\ell^{N - 1}$ is equally likely to be the answer when numbers are sampled uniformly at random, even conditioned on any observation any agent has received.
    
    Because being correct when the answer is all zero has slightly higher reward, as soon as the game begins, they are all in lasting agreement that $(0, 0,\dots, 0)$ is the best action. Playing action $(0, 0, \dots, 0)$ always in this game results in a utility of
    \begin{equation*}
        \frac{1}{\ell^{N - 1}}.
    \end{equation*}
    However, had the agents played perfectly, they could have achieved a utility of at least $1 - \delta$.

    Notice that $CC_{\alpha(1 - \delta)}(G) \leq N \log(\ell)$, as Charlie can act optimally when each agent reveals their full $\log(\ell)$ bit observation. By partitioning as finely as possible so that each observation tuple is its own part, Charlie can act optimally. This bounds the partition number by $N_{\alpha(1 - \delta)}(G) \leq \ell^N$. This gives a total utility of at most
    \begin{equation*}
        \frac{1}{\ell^{N-1}}
    =
    \left(\frac{1}{\ell^N}\right)^{(N-1)/N}
    \leq\left(\frac{1}{N_{\alpha(1 - \delta)}(G)}\right)^{\frac{N - 1}{N}}
    \end{equation*}
    for Aumann agreement.

    It remains to show that $N_{\alpha(1 - \delta)}(G) \in \Omega(k)$. Let $D = \{R_1, \dots, R_{N_{\alpha(1 - \delta)}(G)}\}$ be a rectangular tiling of the observation space $\Omega_1 \times \Omega_2 \times \cdots \times \Omega_N$. Suppose the tiling had value at least $\alpha(1 - \delta)$ and let $\pi: [N_{\alpha(1 - \delta)}(G)] \to \acts$ be the associated policy. Write each rectangle as $R_i = R_i^{(1)} \times \cdots \times R_i^{(N)}$ where $R_i^{(j)} \subseteq \Omega_j$ for all $j \in [N]$.

    Notice that for all $i \in [N_{\alpha(1 - \delta)}(G)]$, $\pi(i)$ can only be correct on at most $\min_{j \in [N]} \abs{R_i^{(j)}}$ observations. Indeed, for any agent $j \in [N]$, for any fixed observation $\omega_j \in \Omega_j$, action $\pi(i) = (a_1, \dots, a_{N - 1})$ can only be correct on one possible observation tuple $\omega_{-j} \in \Omega_{-j}$. This is because the equation
    \begin{equation*}
        (\omega_1 + \omega_2,~ \omega_2 + \omega_3,~ \dots,~ \omega_{N - 1} + \omega_N) = (a_1, \dots, a_{N - 1})
    \end{equation*}
    becomes fully determined once any observation is fixed for any agent.

    Thus, the utility of policy $\pi$ is at most
    \begin{align*}
        &\frac{1}{\ell^N}\sum_{i = 1}^{N_{\alpha(1 - \delta)}(G)} \min_{j \in [N]} \abs{R_i^{(j)}} \\
        &\leq \frac{1}{\ell^N}\sum_{i = 1}^{N_{\alpha(1 - \delta)}(G)} \left(\prod_{j \in [N]} \abs{R_i^{(j)}}\right)^{1/N} && (\text{Because it is a geometric mean)} \\
        &= \frac{1}{\ell^N}\sum_{i = 1}^{N_{\alpha(1 - \delta)}(G)} 1 \cdot \abs{R_i}^{1/N} && (\text{Because the $R_i$ are rectangles})\\
        &\leq \frac{1}{\ell^N}\left(\sum_{i = 1}^{N_{\alpha(1 - \delta)}(G)} 1^{N/(N - 1)}\right)^{(N - 1)/N} \cdot \left(\sum_{i = 1}^{N_{\alpha(1 - \delta)}(G)} \abs{R_i}\right)^{1/N} && (\text{By Hölder's inequality}) \\
        &\leq \frac{1}{\ell^N} \cdot N_{\alpha(1 - \delta)}(G)^{(N - 1)/N} \cdot \left(\ell^N\right)^{1/N} && (\text{Because the $R_i$ tile the space}) \\
        &= \frac{N_{\alpha(1 - \delta)}(G)^{(N - 1)/N}}{\ell^{N - 1}} \\
        &\leq \left(\frac{N_{\alpha(1 - \delta)}(G)}{k}\right)^{(N - 1)/N}.
    \end{align*}
    Since the utility of the policy by assumption is $\alpha(1 - \delta)$, it follows after rearranging that
    \begin{equation*}
        N_{\alpha(1 - \delta)}(G) \geq (\alpha(1 - \delta))^{N/(N - 1)} \cdot k \in \Omega(k). \qedhere
    \end{equation*}
\end{proof}

Up to a multiplicative factor in $\alpha$, this is the strongest lower bound you can prove. When Aumann agreement between two agents terminates, it is guaranteed to result in an expected utility of at least $\alpha N_\alpha(G)^{-1/2}$.

\begin{theorem}[Agreement is Accurate On Low Complexity Games]
    \label{thm:AA-guarantee}
    Let $\alpha \in [0, 1]$ and $\epsilon > 0$ be arbitrary and $G$ be a communication game with only non-negative rewards.

    Suppose Aumann Agreement is run until the agents arrive at $\epsilon$-lasting agreement. Then they must achieve a utility of at least
    \begin{equation*}
        \frac{\alpha}{N_\alpha(G)^{\frac{N - 1}{N}}} - \epsilon \cdot N_\alpha(G)^{1/N}.
    \end{equation*}
\end{theorem}
\begin{proof}
    Let $z = N_\alpha(G)$ and let $D = \{R_1, \dots, R_z\}$ be a rectangular tiling of $\Omega_1 \times \cdots \times \Omega_N$ of value $\alpha$ and let $a_1, \dots, a_z$ be the associated actions played in each rectangle. Run the Aumann agreement protocol until $\epsilon$-lasting agreement is reached. Let $\tau$ be the transcript of messages produced, let $G_\tau$ be the subgame after Aumann Agreement produces transcript $\tau$, let $\dist_\tau$ be the distribution over observations in $G_\tau$, and let $\alpha_\tau$ be the utility $D$ gets under distribution $\dist_\tau$. We start by showing that the average utility in this subgame must be large enough.

    Let $U(a)$ be the average utility of playing action $a$ across all possible observation tuples,
    \begin{equation*}
        U(a) = \sum_{\omega \in \Omega_1 \times \cdots \times \Omega_N} r(a; \omega)\dist_\tau(\omega).
    \end{equation*}
    The key quantity we will use in this proof is the sum of an action along an axis. Define
    \begin{equation*}
        U_i(a; \omega_i) = \sum_{\omega_{-i} \in \Omega_{-i}} r(a; \omega_i, \omega_{-i})\dist_\tau(\omega_i, \omega_{-i})
    \end{equation*}
    to be the average utility agent $i$ receives when she plays $a$ while observing $\omega_i$.

    Because the agents are in $\epsilon$-lasting agreement, it must be the case that there exists an action $a^* \in \acts$ such that for all agents $i \in [N]$, all positive probability observations $\omega_i \in \Omega_i$, and all actions $a \in \acts$: 
    \begin{equation*}
        \E_{\omega_{-i}}[r(a^*; \omega_{-i}, \omega_i) \mid \omega_i] \geq \E_{\omega_{-i}}[r(a; \omega_{-i}, \omega_i) \mid \omega_i] - \epsilon,
    \end{equation*}
    or written in terms of the $U_i$,
    \begin{equation*}
        U_i(a^*; \omega_i) + \epsilon \cdot \Pr(\omega_i) \geq U_i(a; \omega_i).
    \end{equation*}

    The key idea will be to write all relevant quantities in terms of $U_i(a^*; \omega)$, and the lower bound on the performance of Aumann agreement will pop out.

    Define $U(R_j; a_j)$ to be the average utility achieved when playing action $a_j$ in rectangle $R_j$,
    \begin{equation*}
        U(R_j; a_j) = \sum_{\omega \in R_j} r(a_j; \omega)\dist_\tau(\omega).
    \end{equation*}
    Then $\alpha_\tau = \sum_{j = 1}^z U(R_j;a_j)$.

    Because it is a rectangle, we can write each $R_j$ as a product $R_j^{(1)} \times \cdots \times R_j^{(N)}$, where each $R_j^{(i)} \subseteq \Omega_i$ is a set of observations that agent $i$ could receive.

    Notice that because utilities are non-negative, we can expand the sum along any axis. For any agent $i$,
    \begin{equation*}
        U(R_j; a_j) \leq \sum_{\omega_i \in R_j^{(i)}} \sum_{\omega_{-i} \in \Omega_{-i}} r(a_j; \omega_i, \omega_{-i})\dist_\tau(\omega_i, \omega_{-i}) = \sum_{\omega_i \in R_j^{(i)}} U_i(a_j; \omega_i).
    \end{equation*}
    Because the agents are in $\epsilon$-agreement,
    \begin{equation*}
        U(R_j; a_j) - \epsilon \cdot \Pr\left(\omega_i \in R_j^{(i)}\right) \leq \sum_{\omega_i \in R_j^{(i)}} U_i(a^*; \omega_i).
    \end{equation*}
    We can lower bound the left hand side by
    \begin{equation*}
        U(R_j; a_j) - \epsilon \leq \sum_{\omega_i \in R_j^{(i)}} U_i(a^*; \omega_i).
    \end{equation*}
    Because utilities are non-negative, we can additionally say
    \begin{equation*}
        \max(0, U(R_j; a_j) - \epsilon) \leq \sum_{\omega_i \in R_j^{(i)}} U_i(a^*; \omega_i).
    \end{equation*}
    Multiplying these for all agents, we get that
    \begin{equation*}
        \max(0, U(R_j; a_j) - \epsilon)^N \leq \prod_{i \in [N]} \left(\sum_{\omega_i \in R_j^{(i)}}U_i(a^*; \omega_i)\right).
    \end{equation*}
    This quantity is important, so we'll give it the name $Z_j = \prod_{i \in [N]} \left( \sum_{\omega_i \in R_j^{(i)}}U_i(a^*; \omega_i)\right)$. So, $\max(0, U(R_j; a_j) - \epsilon)^N \leq Z_j$ for all $j$ and so $\max(0, U(R_j; a_j) - \epsilon) \leq Z_j^{1/N}$ for all $j$.

    A useful property about the $Z_j$ is that they sum to $U(a^*)^N$. We can see this by a double counting argument. Indeed, for each agent $i$, we can sum the utilities from their perspective to get that
    \begin{equation*}
        U(a^*) = \sum_{\omega_i \in \Omega_i} U_i(a^*; \omega_i).
    \end{equation*}
    Multiplying these equalities, we see that
    \begin{equation*}
        U(a^*)^N = \sum_{\omega_1 \in \Omega_1}\cdots\sum_{\omega_N \in \Omega_N} U_1(a^*; \omega_1)\cdot \dots \cdot U_N(a^*; \omega_N).
    \end{equation*}
    The rectangles partition the space, and so the sum above can be written as
    \begin{equation*}
        U(a^*)^N = \sum_{j = 1}^z \sum_{(\omega_1, \dots, \omega_N) \in R_j} U_1(a^*; \omega_1)\cdot \dots \cdot U_N(a^*; \omega_N).
    \end{equation*}
    The key property that makes this theorem hold is that the $R_j$ are rectangles, and this is where we will use this property. Because we can write $R_j = R_j^{(1)} \times \cdots \times R_j^{(N)}$, we can factorize the sum above into
    \begin{equation*}
        U(a^*)^N = \sum_{j = 1}^z \prod_{i = 1}^{N}\sum_{\omega_i \in R_j^{(i)}} U_i(a^*; \omega_i) = \sum_{j = 1}^z Z_j,
    \end{equation*}
    as claimed.

    From here, we are almost done. Recall that $\alpha_\tau = \sum_{j = 1}^z U(R_j; a_j)$. So, $\alpha_\tau - z \epsilon = \sum_{j = 1}^z (U(R_j; a_j) - \epsilon) \leq \sum_{j = 1}^z \max(0, U(R_j; a_j) - \epsilon)$. Recall that for any $j$, $\max(0, U(R_j; a_j) - \epsilon) \leq Z_j^{1/N}$. Combining the two inequalities, we get
    \begin{align*}
        \alpha_\tau - z\epsilon &\leq \sum_{j = 1}^z Z_j^{1/N}.
    \end{align*}
    Applying Hölder's inequality, we get
    \begin{align*}
        \alpha_\tau - z\epsilon &\leq \sum_{j = 1}^z Z_j^{1/N} \\
                           &= \sum_{j = 1}^z \left(1 \cdot Z_j^{1/N}\right) \\
                           &\leq \left(\sum_{j = 1}^z 1^{N/(N - 1)}\right)^{(N - 1)/N} \cdot \left(\sum_{j = 1}^z Z_j\right)^{1/N} \\
                           &= z^{(N - 1)/N}\left(\sum_{j = 1}^z Z_j\right)^{1/N}.
    \end{align*}
    Recall that $\sum_{j = 1}^z Z_j = U(a^*)^N$, and so we finally get
    \begin{equation*}
        \alpha_\tau - z\epsilon \leq z^{(N - 1)/N} \cdot U(a^*).
    \end{equation*}
    Rearranging, this means that the expected utility that Aumann agreement gets when transcript $\tau$ is produced is at least
    \begin{equation*}
        \frac{\alpha_\tau}{N_\alpha(G)^{\frac{N - 1}{N}}} - \epsilon \cdot N_\alpha(G)^{1/N}.
    \end{equation*}
    Finally, notice that $\E_\tau[\alpha_\tau] \geq \alpha$. This is because $$\alpha \leq \E[r(\pi_D(\omega); \omega)] = \E_\tau[\E[r(\pi_D(\omega); \omega) \mid \tau]] = \E_\tau[\alpha_\tau],$$ where $\pi_D$ is the utility maximizing $D$-measurable policy.
    Taking an expectation over the transcript produced by Aumann agreement $\tau$, the expected utility is therefore at least
    \begin{equation*}
        \frac{\alpha}{N_\alpha(G)^{\frac{N - 1}{N}}} - \epsilon \cdot N_\alpha(G)^{1/N},
    \end{equation*}
    as claimed.
\end{proof}

Intuitively, \Cref{prop:AA-upper-bound} and \Cref{thm:AA-guarantee} together show that when information is split up among more and more agents, agreement becomes less and less reliable.

\subsection{Lasting Agreement Can Take Exponentially Too Long to Reach}

Moreover, in the worst case, despite results like \cite{aaronson_complexity_2005} proving that in a constant number of steps Alice and Bob will both report the same action, arriving at lasting true agreement may take on the order of $\abs{\Omega_A}$ steps in the worst case, meaning Aumann agreement sometimes sends exponentially more bits than just communicating the observation directly!

\begin{proposition}
    [Lasting Agreement Can Take Too Long To Reach]
    \label{prop:AA-too-long}
    For any $k \geq 3$, there exist communication games $G$ with $k$ observations for each agent, constant communication complexity $CC(G) \leq 2$, and constant partition number $N(G) \leq 4$ such that to arrive at lasting agreement, the agents must communicate for $\Omega(k)$ steps in expectation.
\end{proposition}
\begin{proof}
    Consider the following communication game. Alice and Bob's observations are sampled with probabilities proportional to
    \begin{equation*}
        \mu(i, j) \propto \begin{cases}
            \lambda & i = j, \\
            1 & \abs{i - j} = 1, \\
            0 & \text{otherwise},
        \end{cases}
    \end{equation*}
    for some fixed $\lambda \in (1, 2)$. Charlie's action set is $\acts = \{\mathrm{EQ}, \mathrm{NEQ}\}$. Charlie must play action $\mathrm{EQ}$ if Alice and Bob's numbers are equal and action $\mathrm{NEQ}$ if they are not. The agents receive a utility of 1 if Charlie is right and 0 otherwise.

    Notice if Alice has received a 1, then the probability that Bob has received a 1 is
    \begin{equation*}
        \Pr(\mathrm{EQ}\mid \omega_A=1)=\frac{\lambda}{\lambda+1} > 1/2.
    \end{equation*}
    So from Alice's perspective, the best action is $\mathrm{EQ}$ when she observes a 1. By this principle, if either agent receives a number on the boundary of $[k]$, then from their perspective the best action to play in expectation is $\mathrm{EQ}$. 

    Otherwise, if Alice has received a 2, then the probability that Bob has received a 2 is
    \begin{equation*}
        \Pr(\mathrm{EQ}\mid \omega_A=2)=\frac{\lambda}{\lambda+2} < 1/2.
    \end{equation*}
    So from Alice's perspective, the best action is $\mathrm{NEQ}$. The same is true whenever Alice or Bob receive an observation in the interior of $[k]$.

    Therefore, at the first step of Aumann agreement, the agents effectively communicate to each other whether or not their observation is on the boundary of $[k]$. If one of them is, the agents become in agreement on what the best action is and the protocol terminates. Otherwise, we enter the subgame where Alice and Bob's observations lie in $\{2, \dots, k - 1\}$. This continues until one of Alice's or Bob's observation lies on the boundary. This is a smaller version of the same game with two fewer observations, so the same analysis applies and again Alice and Bob simply communicate whether or not they are on the boundary.

    So, Aumann agreement takes $ \min(\omega_A, \omega_B, k + 1 - \omega_A, k + 1 - \omega_B)$ steps.
    In expectation, it can be seen that the number of steps grows linearly. Notice that the probability that $\omega_A = i$ for any $i$ is at least
    \begin{equation*}
        \Pr(\omega_A=i) \geq \frac{\lambda + 1}{k\lambda+2(k-1)} \geq \frac{1}{4k}.
    \end{equation*}
    Therefore, the probability that $\omega_A$ is contained in $[k/4, 3k/4]$ initially is at least $1/8$.
    In this case, the number of steps Aumann agreement runs for is at least $k/4 - 1$. Therefore, the expected value is at least $(1/8) \cdot (k/4 - 1) \in \Omega(k)$.
    
    However, consider the following optimal communication policy. Alice and Bob communicate their observations modulo $2$ in one bit each. Because their numbers differ by at most 1, their numbers are equal if and only if they are equal modulo two. Thus, Charlie can act optimally given this information, which implies that $CC(G) \leq 2$ and $N(G) \leq 2^2 = 4$, as claimed!
\end{proof}
\begin{remark}
    As a side note, notice that \Cref{prop:AA-too-long} does not contradict \cite{aaronson_complexity_2005}. The agents are in agreement over their belief of the optimal action in no more than two iterations. The problem is that agreement, while easy to reach, is often not lasting.    
\end{remark}

\subsection{Weak-Learning Implies Constant Communication Complexity}

\cite{collina2025collaborative} get around these results by showing that when the communication game satisfies a property they call $w$-weak learning, Aumann agreement always converges in constant time to meaningful solutions. We will show that when a communication game is $w$-weak learnable, it must also have a constant (in $w$) communication complexity, meaning $w$-weak learning is a low communication complexity assumption in disguise.

\begin{definition}
    [$w$-weak learning]
    \citep{collina2025collaborative}
    \label{def:weak-learning}
    Let $w: [0, 1] \to [0, 1]$ be a convex increasing function with $w(\gamma) \leq \gamma$ for all $\gamma \in [0, 1]$. A communication game $G$ is said to be $w$-weak learnable when the following property holds.

    Let $\gamma \in (0, 1]$ be arbitrary.
    For all possible distributions over observations $\dist \in \Delta(\Omega_A \times \Omega_B)$, if the utilities of all actions are small enough, i.e.
    \begin{equation*}
        \max_{a \in \acts} \E_{\omega_A, \omega_B \sim \dist}\left[r(a; \omega_A, \omega_B)\right] \leq 1 - \gamma,
    \end{equation*}
    then at least one of Alice's or Bob's optimal policy achieves a utility at least $w(\gamma)$ greater than the best constant action, i.e.
    \begin{align*}
        &\max\left(\max_{\pi_A \in \Pi_A} \E_{\omega_A, \omega_B \sim \dist}[r(\pi_A(\omega_A); \omega_A, \omega_B)],~ \max_{\pi_B \in \Pi_B} \E_{\omega_A, \omega_B \sim \dist}[r(\pi_B(\omega_B); \omega_A, \omega_B)]\right) \\&\qquad \qquad \geq \max_{a \in \acts} \E_{\omega_A, \omega_B \sim \dist}[r(a; \omega_A, \omega_B)] + w(\gamma),
    \end{align*}
    where $\Pi_A$ and $\Pi_B$ are the sets of all functions of the form $\pi_A: \Omega_A \to \acts$ and $\pi_B: \Omega_B \to \acts$ respectively.
\end{definition}

\begin{proposition}
    [Weakly learnable games have constant communication complexity]\citep{collina2025collaborative}
    \label{prop:cc-weak}
    Suppose a communication game $G$ with utilities in $[0, 1]$ is $w$-weak learnable. Let $\gamma \in (0, 1]$ be arbitrary. If $w(\gamma) > 0$, the communication complexity is at most the constant $CC_{1 - \gamma}(G) \leq 2\ceil{\log(m)}\ceil{\frac{(1 - \gamma)}{w(\gamma)}}$ and by extension the partition number is at most $N_{1 - \gamma}(G) \leq (2m)^{2\ceil{(1 - \gamma)/w(\gamma)}}$.
\end{proposition}
\begin{proof}
    This result is implied by \cite[Theorem 7.7]{collina2025collaborative}, but for completeness, we reprove it directly here. To show this, it suffices to show that after $(1 - \gamma)/w(\gamma)$ steps of Aumann agreement, the agents must achieve a utility of at least $1 - \gamma$.

    Let $\tau_t$ be the random variable taking the value of the sequence of actions sent back and forth by Alice and Bob up to iteration $t$ of the Aumann Agreement protocol. Let $\Phi$ be the function mapping the sequence of actions that Alice and Bob send to each other to the utility Charlie achieves at time step $t$. That is, $\Phi: \acts^* \to \mathbb{R}$, where $\acts^* = \bigcup_{k \geq 0} \acts^k$, and
    \begin{equation*}
        \Phi(\tau) = \begin{cases}
            \max_{a \in \acts} \E_{\omega_A, \omega_B}[r(a; \omega_A, \omega_B) \mid \tau] & \text{If } \Pr(\tau) > 0, \\
            0 & \text{Otherwise.}
        \end{cases}
    \end{equation*}
    By the $w$-weak learning assumption (cf. \Cref{def:weak-learning}), at least one of Alice or Bob must outperform a constant action by $w(1 - \Phi(\tau_t))$.

    That is, let $\pi^*_A = \argmax_{\pi_A \in \Pi_A} \E[r(\pi_A(\omega_A); \omega_A, \omega_B) \mid \tau_t]$ be Alice's best policy, and $\pi^*_B = \argmax_{\pi_B \in \Pi_B} \E[r(\pi_B(\omega_B); \omega_A, \omega_B) \mid \tau_t]$ be Bob's current best policy. We know
    \begin{align*}
        \max\left( \E[r(\pi_A^*(\omega_A); \omega_A, \omega_B) \mid \tau_t],~ \E[r(\pi_B^*(\omega_B); \omega_A, \omega_B)\mid \tau_t]\right) \geq \Phi
        (\tau_t)+ w(1 - \Phi(\tau_t)).
    \end{align*}
    After a step of Aumann agreement, Alice and Bob both reveal $\pi^*_A(\omega_A)$ and $\pi^*_B(\omega_B)$ to Charlie, and so it follows that the optimal policy Charlie can play now is at least as good as Alice's or Bob's best policy.

    To see this formally, consider the partition of the observation space $\Omega_A \times \Omega_B$ into best-response blocks for Alice and Bob, i.e. into sets $B_{a, b} = \{(\omega_A, \omega_B) \mid \pi_A^*(\omega_A) = a, \pi^*_B(\omega_B) = b\}$ for all $a, b \in \acts$. Let $B$ be the $\sigma$-algebra induced by the sets $B_{a, b}$.

    If, after a step of Aumann agreement, Alice communicates action $b \in \acts$ and Bob communicates action $c \in \acts$, the communication protocol enters the subgame $B_{b, c}$, with the distribution over observations $\dist_{t + 1}$ being $\dist_t$ conditioned on being in $B_{b, c}$. Note that $\pi^*_A$ and $\pi^*_B$ are both $B$-measurable.

    Because in each $B_{b, c}$ block, Alice and Bob play the fixed action $b$ and $c$ respectively, the best utility Charlie can achieve when she knows they are in the $B_{b, c}$ block is at least as good as Alice or Bob: it follows that $$\max_{a \in \acts} \E[r(a; \omega_A, \omega_B) \mid B_{b, c}, \tau_t] \geq \max(\E[r(\pi_A^*(\omega_A); \omega_A, \omega_B) \mid B_{b, c}, \tau_t], \E[r(\pi_B^*(\omega_B); \omega_A, \omega_B) \mid B_{b, c}, \tau_t]).$$

    So, in expectation, Charlie's utility becomes
    \begin{align*}
        &\E\left[\max_{a \in \acts} \E[r(a; \cdot, \cdot) \mid B, \tau_t](\omega_A, \omega_B) ~\middle
        |~ \tau_t\right] \\
        &= \sum_{b, c \in \acts} \dist(B_{b, c} \mid \tau_t)\max_{a \in \acts} \E[r(a; \omega_A, \omega_B) \mid B_{b, c}, \tau_t] \\
        &\geq \sum_{b, c \in \acts} \dist(B_{b, c} \mid \tau_t) \max\left(\E[r(\pi_A^*(\omega_A); \omega_A, \omega_B) \mid B_{b, c}, \tau_t],~ \E[r(\pi_B^*(\omega_B); \omega_A, \omega_B) \mid B_{b, c}, \tau_t]\right) \\
        &\geq \max\left(\sum_{b, c \in \acts} \dist(B_{b, c} \mid \tau_t)\E[r(\pi_A^*(\omega_A); \omega_A, \omega_B) \mid B_{b, c}, \tau_t],~ \sum_{b, c \in \acts} \dist(B_{b, c} \mid \tau_t)\E[r(\pi_B^*(\omega_B); \omega_A, \omega_B) \mid B_{b, c}, \tau_t]\right) \\
        &= \max\left(\E[r(\pi_A^*(\omega_A); \omega_A, \omega_B)\mid \tau_t],~ \E[r(\pi_B^*(\omega_B); \omega_A, \omega_B) \mid \tau_t]\right),
    \end{align*}
    
    By the $w$-weak learning assumption, it follows that
    \begin{align}
        \E[\Phi(\tau_{t + 1}) \mid \tau_t] &= \E\left[\max_{a \in \acts} \E[r(a; \cdot, \cdot) \mid B](\omega_A, \omega_B) ~\middle
        |~ \tau_t\right] \notag\\
        &\geq \max\left(\E[r(\pi_A^*(\omega_A); \omega_A, \omega_B) \mid \tau_t],~ \E[r(\pi_B^*(\omega_B); \omega_A, \omega_B) \mid \tau_t]\right) \notag \\
        &\geq \Phi(\tau_t) + w(1 - \Phi(\tau_t)).  \label{eq:induction-step-weak-learning}
    \end{align}
    Let $T = \ceil{(1 - \gamma)/w(\gamma)}$. Suppose for the sake of contradiction that $\E_{\omega_A, \omega_B \sim \dist}[\Phi(\tau_T)] < 1 - \gamma$. After $T$ steps of Aumann agreement, it must be the case that Charlie's utility in expectation is 
    \begin{align*}
        \E_{\omega_A, \omega_B \sim \dist}[\Phi(\tau_T)] &= \E_{\tau_{T - 1}}\left[\E_{\omega_A, \omega_B \sim \dist}[\Phi(\tau_T) \mid \tau_{T - 1}]\right] \\
                         &\geq \E\left[\Phi(\tau_{T - 1}) + w(1 - \Phi(\tau_{T - 1}))\right] && (\text{By \Cref{eq:induction-step-weak-learning}}) \\
                         &\geq \E[\Phi(\tau_{T - 1})] + w(1 - \E[\Phi(\tau_{T - 1})]) && (\text{By the convexity of $w$)} \\
                         &\geq \E[\Phi(\tau_0)] + \sum_{t = 0}^{T - 1} w(1 - \E[\Phi(\tau_t)]) &&(\text{Applying the first three steps $T$ times}) \\
                         &\geq \sum_{t = 0}^{T - 1} w(1 - \E[\Phi(\tau_t)]).  && (\text{Because utilities are non-negative})
    \end{align*}
    Notice that expected utilities are non-decreasing, as more information can never hurt. Indeed, because $\E[\Phi(\tau_T)] \geq \E\left[\Phi(\tau_{T - 1}) + w(1 - \Phi(\tau_{T - 1}))\right]$, and $w$ is non-negative, $\E[\Phi(\tau_T)] \geq \E\left[\Phi(\tau_{T - 1}) \right]$. Because $\E_{\omega_A, \omega_B \sim \dist}[\Phi(\tau_T)] < 1 - \gamma$, it follows that $\E_{\omega_A, \omega_B \sim \dist}[\Phi(\tau_t)] < 1 - \gamma$ for all $t \in [T]$.

    Because $w$ is non-decreasing by assumption, this means that
    \begin{equation*}
        \E_{\omega_A, \omega_B \sim \dist}[\Phi(\tau_T)] \geq \sum_{t = 0}^{T - 1} w(1 - \E[\Phi(\tau_t)]) \geq \sum_{t = 0}^{T - 1} w(1 - (1 - \gamma)) \geq T \cdot w(\gamma) \geq 1 - \gamma.
    \end{equation*}
    
    Contradiction! Thus, after $T$ steps, the expected utility for the agents must be at least $1 - \gamma$.
    Each step of Aumann agreement communicates $2\ceil{\log(m)}$ bits, and we are done!
\end{proof}

%% file: sections/07-related-work.tex
\section{Additional Related Work}
Beyond the agreement literature discussed in the introduction, related work comes from both communication-complexity approaches to protocol design and economic models of strategic communication. We discuss the closest connections here.

\paragraph{Strategic communication and information design.}
A separate line of work in economics and computer science studies communication among agents whose objectives, information, abilities, or communication costs may be partially misaligned. In cheap-talk, Bayesian persuasion, and information-design models, the central question is how a sender, mediator, or platform can shape what agents learn in order to influence downstream actions \citep{crawford1982strategic,kamenica2011bayesian,bergemann2019information}. Algorithmic work in this area studies the computational complexity of designing such information structures under a range of constraints, including public and private persuasion, limited or non-robust communication channels, and constrained communication languages \citep{dughmi2016algorithmic,le2019persuasion,gradwohl2022algorithms,haghtalab2022communicating,haghtalab2024platforms,haghtalab2024leakage}. Our work is orthogonal to this literature by taking a different perspective: we remove the strategic aspects that arise from mismatches in agents' goals, abilities, or costs, and focus instead on the algorithmic problem of collaborative communication itself. We believe an interesting direction for future work is to study multi-agent communication with limited bandwidth and computational complexity that also suffers from some misalignment of objectives.

\paragraph{Related problems in clustering.} Communication can be seen as a very general clustering problem, and clustering in its full generality --- with arbitrary distance functions that, for example, do not necessarily correspond to a metric --- has been studied in the literature from time to time. For example, a reduction from the clustering-like problem of segmentation, as studied by \cite{kleinberg2004segmentation, alon1999two}, can be shown to the problem of communication. \cite{balcan2008discriminative} show how the Frieze-Kannan weak regularity lemma can be used to solve general clustering problems. Clustering to facilitate communication in cooperative settings has been studied empirically in work like \cite{lauffer_who_2023}.

\paragraph{Regularity, indistinguishability, and multicalibration.}
Our use of coarsenings as small abstractions that preserve the behavior of a class of tests is naturally connected to outcome indistinguishability and to recent work relating graph regularity, pseudorandomness, outcome indistinguishability, and multicalibration \citep{dwork2021outcome,dwork2023pseudorandomness,casacuberta2024complexity,casacuberta2025global}. This perspective is part of a broader line of work using multicalibration and indistinguishability-style guarantees to obtain downstream decision guarantees, including omniprediction, loss outcome indistinguishability, online multicalibration, panprediction, and multi-objective learning \citep{gopalan2021omnipredictors,gopalan2022loss,garg2023oracle,balakrishnan2025panprediction}. The models and guarantees are different: in our setting the tests are induced by short communication protocols, and indistinguishability by itself is not sufficient. Rather, the indistinguishable approximation must itself be a coarsening of the agents' observation spaces.

Moreover, \cite{tao2006szemer} has already noted a connection between regularity and information, introducing an information-theoretic variant of the Szemerédi regularity lemma. The author does not give an efficient algorithm to compute this. Our strengthening of the weak regularity lemma can be seen as a computationally efficient instantiation of the information-theoretic variant for finding cut-norm decompositions. As a remark, the vicious cycle we encounter when attempting to use indistinguishability without coarsening has also been encountered by \cite{dwork2026efficient}. Their proposed solution, a ``supersimulator'', uses a clever variation of the boosting argument in which the complexity of the tests increases with the complexity of the object being constructed. This is strong enough to allow for $\epsilon$-indistinguishability to $\mathrm{poly}(1/\epsilon)$ combinations of test functions. For communication, where we ask for indistinguishability with respect to exponentially in $1/\epsilon$ more complex test functions, a supersimulator would create partitions of size $2 \uparrow\uparrow (1/\epsilon^2)$, a power tower with $1/\epsilon^2$ exponentiations. Coarsening strikes a balance, allowing us to get away with a partition with much, much fewer parts when we only need to be $\epsilon$-indistinguishable to a certain class of exponentially more complex test functions.

\paragraph{One-way protocols and learning.}
\cite{kremer_randomized_1999} provide an efficient algorithm for one-way communication of Boolean functions, where Bob takes the action but only Alice is allowed to speak. They do this by connecting one-way communication under product distributions over observations to PAC learning. In particular, their approach gives a general efficient communication protocol whose bit complexity is small when the VC dimension of the rows of the communication matrix is small.

This is close in spirit to our goal of algorithmically finding useful protocols, but the setting and benchmark are different. The product-distribution assumption is essential for their guarantee, while our model allows correlated observations. Their setting is also a function-computation setting, whereas in our model Charlie may have many possible actions with different utilities for the same observation pair. Moreover, VC dimension is a worst-case structural parameter, whereas in our setting the relevant communication complexity can be much smaller for the particular game and distribution at hand. \cite{jain2009new} extend this learning-based approach to joint distributions, but their communication bound depends linearly on the mutual information between Alice's and Bob's observations, which can be as large as $\log n$ even when the communication complexity is constant. Relatedly, \cite{feldman2014sample} prove an equivalence between one-way communication under general distributions over inputs and private PAC learning.

\paragraph{Low-rank protocols for functions.}
Another source of general-purpose communication protocols comes from work around the log-rank conjecture. For a function $f$ with communication matrix $M_f$, a number of protocols exploit low rank of $M_f$ \citep{nisan1995rank, lovett2016communication, rothvoss2014direct, sudakov2025matrix}. When $M_f$ has constant rank, there is an efficient $\bO{\mathrm{rank}(M_f)}$-bit protocol: the agents agree on a set of $\rank(M_f) \leq 2^{CC(f)}$ columns forming a basis, and Alice sends her observed row's entries on those basis columns. This already gives a general-purpose $\bO{2^{CC(f)}}$-bit protocol for functions. The Nisan--Wigderson protocol \citep{nisan1995rank}, together with observations from \citep{lovett2016communication}, improves this to $\bO{\sqrt{\mathrm{rank}(M_f)\log(\mathrm{rank}(M_f))}}$ bits, and this was later improved to $\bO{\sqrt{\mathrm{rank}(M_f)}}$ by \citep{sudakov2025matrix}. These results induce general-purpose $\bO{2^{CC(f)/2}}$-bit protocols for functions.

These protocols are powerful in the function setting, but they do not directly extend to our decision-centric model with correlated observations and potentially many actions of varying utility. Our lower bound shows that no analogue with subexponential dependence on $CC(G)$ can hold in this generality: efficient algorithms may require $\Omega(2^{CC(G)})$ bits. Thus the low-rank improvements for functions cannot transfer wholesale to our setting. There are also computational obstacles in some of these function-protocol approaches. For example, one route uses a large-submatrix step \cite[Lemma 1.3]{lovett2016communication}, which is related to submatrix optimization problems that are hard to approximate in polynomial time under the Small Set Expansion Hypothesis and $\mathrm{NP} \not\subseteq \mathrm{BPP}$ \citep{manurangsi2018inapproximability}. Balancing the resulting Nisan--Wigderson protocol tree can also take time exponential in its size; in the worst case this can be as large as $2^{\sqrt{\mathrm{rank}(M_f)}} \leq 2^{2^{CC(f)/2}}$, since $\mathrm{rank}(M_f)$ can be as large as $2^{CC(f)}$. While \cite{rothvoss2014direct} give an efficient convex-optimization-based way of finding large monochromatic rectangles for constant-rank matrices, these techniques remain tied to the function-computation setting.

\paragraph{Hardness of communication complexity.}
\citep{harsha2007communication} study the relationship between communication complexity and computation, giving functions with low communication complexity where every low enough complexity communication protocol requires superpolynomial compute to execute. From a different perspective on the computational hardness of communication, a line of work gives lower bounds on the time required to compute the value of the communication complexity itself. Because finding the shortest protocol would be one way of computing the communication complexity, these results are closely related to the problem of effective communication. \cite{kushilevitz_complexity_2009} prove that it is cryptographically hard to compute communication complexity. \cite{hirahara2021hardness} prove that it is NP-hard to compute the minimum number of bits needed by a communication protocol with a limited number of rounds. Recently, \cite{hirahara2025communication} and \cite{gaspers2025np} have shown that computing communication complexity is NP-hard. Our hardness result is complementary: rather than asking whether the optimal communication complexity can be computed, we show that even when a short high-utility protocol exists, computational efficiency may force an exponential blowup in the number of communicated bits.

%% file: sections/08-Discussion.tex
\section{Discussion}

This paper introduces the problem of effective multi-agent communication and shows that coarsening observations can be used as an algorithmic method for keeping communication simultaneously short, useful, and efficiently computable. The question is a natural one for the TCS community, combining communication, computation, and information abstraction in multi-agent environments. It also points toward a broader motivation: a principled perspective on \emph{agentic memory}. How should an agent, or a collection of agents, decide what information from a long stream of observations, computations, or interactions is worth carrying forward? This question is increasingly concrete in systems built from large language model agents such as Codex~\citep{openai_codex} and Claude~\citep{anthropic_claude}, where useful information may be distributed across tools, sub-agents, documents, previous reasoning steps, or long interaction histories that cannot all fit in the context window. The results in this paper suggest one possible principle: memory should be judged not by whether it reconstructs the past, but by whether it enables the agent to make good decisions in the future. Coarsening the information space is a natural way to communicate in a language native to the models, rather than by designing arbitrary encoding and decoding schemes.

We believe a principled perspective on \emph{agentic memory} is fertile ground for theoretical computer science, with the potential to address algorithmic questions at the frontier of AI agentic systems. We end with a brief discussion of directions where the TCS community could help build this foundation.

\paragraph{Stateful environments and agentic memory.}
The communication protocols in this paper are non-adaptive: the messages Alice and Bob send do not depend on each other's earlier messages. This raises the question of when adaptivity is useful for effective communication, and whether adaptive protocols can outperform \Cref{alg:One-Way-Communication}. Investigating adaptivity is a natural next step.

The model in this paper is also a one-shot information gathering model in which agents observe private information, communicate, then take an action. However, many natural agentic systems are stateful, in which information gathering itself is correlated with earlier messages and actions. This suggests studying multi-agent communication and coarsenings in sequential settings such as MDPs and  POMDPs.

\paragraph{Mixed motives and strategic communication.}
The model in this paper assumed that all agents share the same payoff (the game of common interest). While this is perhaps the right starting point for isolating the algorithmic problem of useful communication, many multi-agent settings involve partially aligned or mixed-motive agents.
In such settings, an agent may want to communicate information that is useful for one objective while withholding information relevant to another. Understanding how mixed motives shape information revelation under communication and computational constraints, perhaps by bridging Bayesian persuasion with the algorithmic perspective on communication introduced here, could provide a principled way to reason about what agents should reveal, remember, or compress when incentives are not perfectly aligned.

\paragraph{Decision-preserving regularity broadly.}
At a technical level, our main tool is a regularity lemma that produces a small abstraction of a game while preserving the value of all short protocols. We believe that understanding the broader connections between notions of agentic memory via coarsenings, regularity, indistinguishability (as well as their connections to multicalibration, omniprediction, and the broader perspective on multi-objective learning) presents an exciting agenda for the TCS community in line with recent developments connecting these notions ~\citep{hebert2018multicalibration,dwork2021outcome,gopalan2021omnipredictors,haghtalab2023unifying,garg2023oracle,balakrishnan2025panprediction,casacuberta2024complexity,casacuberta2025global}.

\section*{Acknowledgments}
The authors would like to thank Cam Allen, Natalie Collina, Parikshit Gopalan, Brian Lee, and Abhishek Shetty for helpful discussions and useful feedback on earlier drafts of this work.
This work was supported in part by the National Science Foundation under grant CCF-2145898, by the Office of Naval Research under grant N00014-24-1-2159, a Google Research Scholar Award, an Alfred P. Sloan fellowship, a Schmidt Science AI2050 fellowship, and by a gift from Coefficient Giving (formerly Open Philanthropy) supporting the work of the Center for Human-Compatible AI at Berkeley. Mark Bedaywi is supported by a Cooperative AI PhD Fellowship and a Canada Graduate Research Scholarship.

%% file: sections/09-extra-proofs.tex
\section{The Random Communication Policy Performs Poorly}

Here we show that randomly partitioning the observation space is very ineffective. We can prove that for any $n$, there exist games with $n$ observations and a communication complexity of 1 such that a randomly sampled protocol either sends $\Omega(\log(n))$ bits --- on the order of that required to represent the full observation --- or essentially achieves the same utility as not sending anything.

\begin{restatable}[The Average Protocol Fails]{proposition}{RandomProtocol}
    \label{prop:grid-must-fail}
    For any $n, k \in \mathbb{N}$, there exists a communication game $G$ with $n$ observations, two actions, and $CC(G) \leq 1$, where the optimal one-bit scheme achieves a utility of 1, but a uniformly random partition of the observation space into 
    $2^k$ parts induces a protocol whose expected utility is at most
    \begin{equation*}
        \frac{1}{2} \cdot \left(1 + \sqrt{\frac{2^k}{n}}\right).
    \end{equation*}
    So, if $k \in o(\log(n))$, a randomly sampled communication protocol achieves in expectation no more than $\nicefrac{1}{2}\cdot (1 + o(1))$ utility.
\end{restatable}
\begin{proof}
    Consider the following family of communication games. Some set $S$ is fixed and known to all agents. Alice receives a number $\omega$ from 1 to $n$ uniformly at random, Bob receives no observation, and Charlie must guess whether $\omega \in S$. The communication complexity is at most $1$, as Alice can simply communicate $\mathds{1}(\omega \in S)$. This is true regardless of the choice of $S$. The partition number of this game is 2, as the partition simply consists of $S$ and $S^C$.

    We will construct the sequence of games one game at a time. For any choice of $n$, we will show that some set $S$ exists for which Alice communicating which part in a random partition her secret number is in provides Charlie with as little information as possible.

    We will use the notation $\mathcal{P}_i(A)$ to mean the collection of all partitions of $A$ into $i$ parts.
    Consider a random $k$-bit partitioning protocol instead. This is a random partitioning of $[n]$ into $2^k$ parts, $(T_1, \dots, T_{2^k}) \sim \mathrm{Unif}\left[\mathcal{P}_{2^k}([n])\right]$, where Alice communicates which part she is in. The argument will be that $k$ must be substantially large before signaling which part of a random partition you are in is useful to Charlie on average.
    We use the probabilistic method to simplify the analysis below: we will assume that $S$ is randomly sampled for now. The course of the argument will be to show that random $S$ using random partitions must fail and to conclude that there must exist some bad set $S$ for which random partitions fail.

    Suppose $S$ is sampled uniformly at random, so that each element in $[n]$ is in $S$ with probability $0.5$, and these events are independent. Let $V(T)$ be the reward Charlie can achieve under the knowledge that Alice's observation is contained in $T$. The expected utility of a randomly chosen partition of size $2^k$ is
    \begin{equation*}
        \E_S\E_{T_1, \dots, T_{2^k} \sim \mathrm{Unif}[\mathcal{P}_{2^k}([n])]}\left[\sum_{i = 1}^{2^k}V(T_i)\cdot\frac{\abs{T_i}}{n} \right]
        = \E_{T_1, \dots, T_{2^k} \sim \mathrm{Unif}[\mathcal{P}_{2^k}([n])]}\left[\sum_{i = 1}^{2^k}\frac{\abs{T_i}}{n}\E_SV(T_i) \right].
    \end{equation*}

    For an arbitrary set $T \subseteq [n]$, what is the value of knowing $T$ for a random $S$? Knowing Alice's observation is contained in $T$, your possible actions are to guess that $\omega \in S$ and in expectation receive $\frac{\abs{T \cap S}}{\abs{T}}$, or guess that $\omega \notin S$ and receive $\frac{\abs{T \setminus S}}{\abs{T}}$. This means
    \begin{align}
        \E_SV(T)
        &= \E_S \left[\max\left(\frac{\abs{T \cap S}}{\abs{T}}, \frac{\abs{T \setminus S}}{\abs{T}}\right)\right] \nonumber\\
        &= \frac{1}{2}\E_S \left[\frac{\abs{T \cap S}}{\abs{T}} + \frac{\abs{T \setminus S}}{\abs{T}} + \abs{\frac{\abs{T \cap S}}{\abs{T}} - \frac{\abs{T \setminus S}}{\abs{T}}}\right] && \left(\text{Since} \max(a, b) = \frac{1}{2}(a + b + \abs{a - b})\right)\nonumber\\
        &= \frac{1}{2} + \frac{1}{2\abs{T}}\E_S \left[\big|\abs{T \cap S} - \abs{T \setminus S}\big|\right] \nonumber\\
        &= \frac{1}{2} + \frac{1}{2\abs{T}}\E_S \left[\big|\abs{T \cap S} - (\abs{T} - \abs{T \cap S})\big|\right] \nonumber\\ 
        &= \frac{1}{2} + \E_S \left[\left|\frac{\abs{T \cap S}}{\abs{T}} - \frac{1}{2}\right|\right]. \label{eq:grid-fails-continue}
    \end{align}
    Let $X_i$ be the indicator random variable that $i \in S$. By construction, $X_i \sim \text{Bern}(0.5)$, and $X_i \bot X_j$ for each $i \neq j$. This means $\E_S\sum_{t \in T} X_t = 0.5\abs{T}$ and $\text{Var}_S\left(\sum_{t \in T} X_t\right) = \sum_{t \in T}\text{Var}_S\left(X_t\right) = 0.25\abs{T}$. By Cauchy-Schwarz
    \begin{align*}
        \E_S \left[\left|\frac{\sum_{t \in T}X_t}{\abs{T}} - \frac{1}{2}\right|\right]
        &= \E_S \left[\sqrt{\left(\frac{\sum_{t \in T}X_t}{\abs{T}} - \frac{1}{2}\right)^2}\right] \\
        &\leq \sqrt{\E_S \left[\left(\frac{\sum_{t \in T}X_t}{\abs{T}} - \frac{1}{2}\right)^2\right]} \\
        &= \sqrt{\text{Var}_S\left(\frac{1}{\abs{T}}\sum_{t \in T}X_t\right)} \\
        &= \frac{1}{2\sqrt{\abs{T}}}. 
    \end{align*}
    Plugging this back into \Cref{eq:grid-fails-continue} and using Cauchy-Schwarz twice more, we get that
    \begin{align*}
        &\E_S\E_{\left(T_1, \dots, T_{2^k}\right) \sim \mathrm{Unif}\left[\mathcal{P}_{2^k}([n])\right]}\left[\sum_{i = 1}^{2^k}V(T_i)\cdot\frac{\abs{T_i}}{n} \right] \\
        &\leq \E_{\left(T_1, \dots, T_{2^k}\right) \sim \mathrm{Unif}\left[\mathcal{P}_{2^k}([n])\right]}\left[\sum_{i = 1}^{2^k}\frac{\abs{T_i}}{n} \left(\frac{1}{2} + \frac{1}{2\sqrt{\abs{T_i}}}\right) \right] \\
        &= \frac{1}{2n}\E_{\left(T_1, \dots, T_{2^k}\right) \sim \mathrm{Unif}\left[\mathcal{P}_{2^k}([n])\right]}\left[\sum_{i = 1}^{2^k}\abs{T_i}\right] + \frac{1}{2n}\E_{\left(T_1, \dots, T_{2^k}\right)\sim \mathrm{Unif}\left[\mathcal{P}_{2^k}([n])\right]}\left[\sum_{i = 1}^{2^k}\sqrt{\abs{T_i}} \right] \\
        &\leq \frac{1}{2} + \frac{1}{2n}\E_{\left(T_1, \dots, T_{2^k}\right) \sim \mathrm{Unif}\left[\mathcal{P}_{2^k}([n])\right]}\left[\sqrt{2^k\sum_{i = 1}^{2^k}\abs{T_i}} \right] \\
        &\leq \frac{1}{2} + \frac{1}{2n}\sqrt{2^k\cdot \E_{\left(T_1, \dots, T_{2^k}\right) \sim \mathrm{Unif}\left[\mathcal{P}_{2^k}([n])\right]}\left[\sum_{i = 1}^{2^k}\abs{T_i} \right]} \\
        &= \frac{1}{2} + \frac{1}{2}\sqrt{\frac{2^k}{n}}.
    \end{align*}
    Since this is the expected value of the random partition communication protocol for a randomly chosen set $S$, there must then exist a set $S$ such that the value Charlie gets on average is at most $(1 + \sqrt{\frac{2^k}{n}})/2$. If $k$ grows at least on the order of $\log(n)$, you might as well just communicate the full state you are in. Otherwise, if $k$ grows asymptotically slower than $\log(n)$, this tends towards $1/2$ for large enough games, performing virtually no better than not communicating at all! 
\end{proof}